\documentclass[conference,letterpaper]{IEEEtran}

\usepackage[utf8]{inputenc} 
\usepackage[T1]{fontenc}
\usepackage{url}
\usepackage{ifthen}
\usepackage{cite}
\usepackage[cmex10]{amsmath} % Use the [cmex10] option to ensure complicance
\usepackage{xcolor}
\usepackage{amssymb}
\usepackage{bm}

\usepackage{graphicx}    % For including images
\usepackage{subcaption}  % For subfigures and subcaptions

\usepackage{multicol}

\usepackage{multirow}

\newenvironment{proof}{\par\noindent\textit{Proof:}\ }{\hfill$\square$\par}

\newtheorem{theorem}{Theorem}
\newtheorem{definition}{Definition}

\newtheorem{proposition}{Proposition}
\newtheorem{lemma}{Lemma}

\definecolor{americanrose}{rgb}{1.0, 0.01, 0.24}
\definecolor{ao(english)}{rgb}{0.0, 0.5, 0.0}

\begin{document}
\title{Differential Confounding Privacy and Inverse Composition} 

% %%% Single author, or several authors with same affiliation:
% \author{%
%  \IEEEauthorblockN{Author 1 and Author 2}
% \IEEEauthorblockA{Department of Statistics and Data Science\\
%                    University 1\\
 %                   City 1\\
  %                  Email: author1@university1.edu}% }

%%% Several authors with up to three affiliations:
\author{%
  \IEEEauthorblockN{Tao Zhang}
  \IEEEauthorblockA{Computer Science and Engineering \\
                    Washington University in St. Louis\\
                    tz636@nyu.edu}
  \and
  \IEEEauthorblockN{Bradley A. Malin}
  \IEEEauthorblockA{Department of Biostatistics\\
  Vanderbilt University\\
  b.malin@vumc.org}\\
  \and
  \IEEEauthorblockN{Netanel Raviv and Yevgeniy Vorobeychik}
  \IEEEauthorblockA{Computer Science and Engineering \\ 
                    Washington University in St. Louis\\
                    \{netanel.raviv,  yvorobeychik\}@wustl.edu}
}

\maketitle

%%%%%%
%% Abstract: 
%% If your paper is eligible for the student paper award, please add
%% the comment "THIS PAPER IS ELIGIBLE FOR THE STUDENT PAPER
%% AWARD." as a first line in the abstract. 
%% For the final version of the accepted paper, please do not forget
%% to remove this comment!
%%

\begin{abstract}
Differential privacy (DP) has become the gold standard for privacy-preserving data analysis, but its applicability can be limited in scenarios involving complex dependencies between sensitive information and datasets. 
To address this, we introduce \textit{differential confounding privacy} (DCP), a specialized form of the Pufferfish privacy (PP) framework that generalizes DP by accounting for broader relationships between sensitive information and datasets.
DCP adopts the $(\epsilon, \delta)$-indistinguishability framework to quantify privacy loss.
We show that while DCP mechanisms retain privacy guarantees under composition, they lack the graceful compositional properties of DP. To overcome this, we propose an \textit{Inverse Composition (IC)} framework, where a leader-follower model optimally designs a privacy strategy to achieve target guarantees without relying on worst-case privacy proofs, such as sensitivity calculation. Experimental results validate IC's effectiveness in managing privacy budgets and ensuring rigorous privacy guarantees under composition.
\end{abstract}

\section{Introduction}

Information privacy is a critical challenge in the digital age, driven by the widespread use of data-driven technologies and personal data. Addressing this requires robust, interpretable, and quantifiable privacy protection techniques.
The problem can be framed as follows: Let $X$ be a dataset processed by a (possibly randomized) mechanism $\mathcal{M}: \mathcal{X} \to \mathcal{Y}$, producing an output $Y = \mathcal{M}(X)$ that is publicly observed, where $\mathcal{X}$ and $\mathcal{Y}$ are the input and output spaces, respectively. The dataset $X$ encodes information about sensitive information (or the \textit{secret}) $S$. The goal of privacy protection is to quantify how much an attacker can infer about the secret $S$ from the output $Y$.

Differential privacy (DP) \cite{dwork2006calibrating} has become the gold standard for privacy-preserving data analysis. In DP, the secret $S$ is an integral part of a dataset, encompassing attributes, entries, values, or membership. DP ensures that any two \textit{adjacent datasets} $x, x' \in \mathcal{X}$, differing by a single entry ($x \simeq x'$), produce nearly indistinguishable outputs, limiting the information inferable about the secret.
A randomized mechanism $\mathcal{M}: \mathcal{X} \to \mathcal{Y}$ is \textit{$(\epsilon, \delta)$-differentially private} ($(\epsilon, \delta)$-DP) with $\epsilon \geq 0$ and $\delta \in [0,1]$ if it is \textit{$(\epsilon, \delta)$-indistinguishable} ($(\epsilon, \delta)\textup{-}\mathtt{ind}$) for all $x \simeq x'$:
\begin{equation}\label{eq:standard_DP_def}
    \sup_{\mathcal{W} \subset \mathcal{Y}} \left( \textbf{Pr}\left[\mathcal{M}(x) \in \mathcal{W}\right] - e^{\epsilon} \textbf{Pr}\left[\mathcal{M}(x') \in \mathcal{W}\right] \right) \leq \delta.
\end{equation}
Pufferfish privacy (PP) \cite{kifer2014pufferfish} generalizes differential privacy (DP) by extending the adjacency concept from datasets to \textit{distributions} over datasets, conditioned on the adversary’s prior knowledge \cite{desfontaines2019sok}. This allows PP to incorporate domain-specific assumptions and protect customizable sensitive properties of the data—extending beyond individual data points or records.
Appendix \ref{app:Puffish Privacy} shows the formal definition of PP.

In DP, the secret and the dataset enter the $(\epsilon,\delta)$-indistinguishability inequality through a deterministic one-to-one relationship: once the dataset is fixed, the secret (i.e., data point) is uniquely determined, and adjacent datasets correspond to distinct secret values. Hence, the randomness in the relationship between the secret and the mechanism output arises solely from the randomness of the mechanism itself.
However, in real-world systems such as multi-agent networks, machine learning pipelines, or distributed systems, sensitive information often arises from complex interactions and dependencies that go beyond individual records, data attributes, or features that are observable or directly inferable from the dataset (e.g., an individual’s activity at a specific time).
These systems involve derived knowledge, intermediate information, or correlations generated within internal mechanisms.
Here are a few illustrative examples (see Appendix \ref{app:examples} for more detailed examples).

\paragraph{Example: Online Learning and Personalized Recommendation Systems}
Online platforms use user interactions \(X\) (e.g., clicks, watch history) to infer secrets \(S\), such as latent preferences or user profiles.
The platform updates \(S\) based on \(X\), and \(S\) shapes future recommendations, influencing user behavior. Over time, \(S\) and \(X\) evolve together in a feedback loop (\(S \leftrightarrow X\)). Unlike static datasets, \(S\) is dynamically updated and not directly present in \(X\).

\paragraph{Example: Medical Diagnosis and Screening Pipelines}
In healthcare, tests and screenings generate datasets \(X\) (e.g., lab results, imaging data). The secret \(S\) is a diagnosis or health label inferred from \(X\), such as “high risk of condition \(S\).”
The causal flow starts with raw data \(X\), which is analyzed to derive \(S\). Feedback loops occur when a tentative diagnosis triggers further tests, creating new data that updates \(X\). Since \(S\) depends on multiple variables in \(X\) and evolves dynamically, \(S \not\subseteq X\).

To capture such general relationships between $S$ and $X$, we introduce \textit{differential confounding privacy} (DCP), a specialized form of the PP framework, to quantify the privacy leakage of the secret $S$ from the output $Y$. DCP adopts the \((\epsilon, \delta)\) scheme of DP and PP to parameterize probabilistic (near) indistinguishability.

Composition, the ability to run multiple private algorithms on the same dataset while preserving privacy, is a fundamental property of effective privacy frameworks.
Graceful composition properties are key to the success of DP \cite{dwork2022differential} and are essential for the practicality and scalability of any useful privacy framework. 
We show that the composition of DCP mechanisms remains DCP but lacks the graceful composition properties of DP even for independent mechanisms. Specifically, the \textit{optimal composition} bound for DP \cite{kairouz2015composition,murtagh2015complexity} underestimates the cumulative privacy loss in DCP, and the \textit{basic composition} property \cite{dwork2006our,kairouz2015composition} is generally not satisfied. This presents challenges in determining optimal or conservative privacy bounds for DCP under composition.

To tackle these challenges, we introduce the \textit{inverse composition} (IC) framework, which reimagines the approach to privacy composition. In contrast to traditional methods that depend on worst-case guarantees, such as sensitivity calculations (generally NP-hard \cite{xiao2008output}) and combinatorial composition analysis (generally \#P-hard \cite{murtagh2015complexity}), IC adopts a leader-follower model to strategically design privacy strategies.
The leader’s privacy strategy elicits an optimal response from the follower, ensuring that privacy guarantees are met without explicitly constraining the privacy strategy or the mechanism. IC provides a systematic method to achieve $(\epsilon, \delta)$-DCP guarantees under composition, offering a flexible framework for managing privacy budgets and addressing the complexities of DCP.

\textbf{Organization:} Section II defines the DCP framework. Section III analyzes the composition of DCP mechanisms and their limitations compared to DP. Section IV presents the inverse composition framework, and Section V provides numerical validation. Section VI concludes the paper. Technical details and proofs are in the appendices of the full version \cite{zhang2024confounding}.

\subsection{Related Work}

% The notion of DCP builds on the general Pufferfish privacy (PP) framework \cite{kifer2014pufferfish} that generlizes the deterministic and one-to-one relationship between the secret and the dataset considered by differential privacy (DP) frameworks \cite{dwork2006calibrating}.
% A key challenge in PP frameworks is the scarcity of suitable mechanisms. Notable instantiations include Blowfish privacy \cite{he2014blowfish}, Wasserstein mechanisms for correlated data \cite{song2017pufferfish}, attribute privacy \cite{zhang2022attribute}, and mutual-information PP \cite{nuradha2023pufferfish}.

Differential confounding privacy (DCP) builds on the general Pufferfish privacy (PP) framework \cite{kifer2014pufferfish}, which extends the deterministic, one-to-one relationship between secrets and datasets assumed in Differential privacy (DP) \cite{dwork2006calibrating} to support more general and expressive privacy specifications.
Despite its expressive power, a key challenge in the PP framework is the scarcity of practical and broadly applicable mechanisms. Notable instantiations include Blowfish privacy \cite{he2014blowfish}, Wasserstein mechanisms for correlated data \cite{song2017pufferfish}, attribute privacy \cite{zhang2022attribute}, and mutual-information-based PP \cite{nuradha2023pufferfish}.

\textbf{Composition of PP }  
PP does not always compose \cite{kifer2014pufferfish,song2017pufferfish}. \cite{kifer2014pufferfish} introduce a sufficient condition, \textit{universally composable evolution scenarios} (UCES), for linear self-composition, analogous to DP’s basic composition property \cite{dwork2006calibrating}. UCES assumes a deterministic one-to-one mapping between a secret and a dataset, leveraging DP’s composition benefits by eliminating randomness in this relationship, which corresponds to an extremely confident attacker. However, in CP, where causation may flow from datasets to secrets, UCES is often infeasible. This complexity arises when another internal mechanism controls the state-dataset relationship, preventing the elimination of randomness.

\textbf{Classic Composition of DP }  
DP’s composition property underpins its success. The advanced composition of $(\epsilon, \delta)$-DP \cite{dwork2010boosting} tightens cumulative privacy bounds by considering probabilistic behavior. \cite{kairouz2015composition} derive optimal composition for homogeneous $(\epsilon, \delta)$-DP mechanisms, while \cite{murtagh2015complexity} show that finding tight bounds for heterogeneous mechanisms is $\#P$-complete. Recent works on interactive mechanisms confirm that they compose comparably to standard DP mechanisms \cite{vadhan2021concurrent,lyu2022composition,vadhan2023concurrent}.

\textbf{Fine-Grained Composition of DP }  
Functional approaches like R\'enyi DP (RDP) \cite{mironov2017renyi}, privacy profiles \cite{balle2018improving}, $f$-DP \cite{dong2022gaussian}, and privacy loss distribution (PLD) \cite{sommer2019privacy,koskela2020computing} provide fine-grained characterizations of privacy guarantees. These methods compose by aggregating functional terms, enabling mechanism-specific analyses, tighter privacy bounds, better privacy-utility trade-offs, and efficient privacy accounting \cite{bun2016concentrated,abadi2016deep,balle2018improving,wang2019subsampled,koskela2020computing,gopi2021numerical,zhu2022optimal,koskela2022individual}. Conversions to $(\epsilon, \delta)$-DP are achieved through mechanism standardization, but lossless conversions avoiding $\#P$-completeness under composition remain an open problem.

\section{Differential Confounding Privacy}

We formally define \textit{differential confounding privacy (DCP)} as follows. Let $\{\theta, \mathcal{G}\}$ describe the relationship between the secret $S$ and the dataset $X$, where:
\begin{itemize}
    \item $\theta\in\Theta \subseteq \Delta(\mathcal{S} \times \mathcal{X})$ determines a joint density $P_{\theta}(s,x)$, where $\Theta$ is a class of joint distributions over $S$ and $X$. 
    \item $\mathcal{G}: \mathcal{X} \to \mathcal{S}$ is a randomized mapping such that $P_\theta(s, x) > 0$ if and only if $\textbf{Pr}[\mathcal{G}(x) = s] > 0$.
\end{itemize}
In addition, let $P_{\theta}(s)$ and $P_{\theta}(x|s)$ be the marginal and conditional densities derived from $P_{\theta}(s,x)$.
In addition, let $\mathtt{D}(\cdot|\theta): \mathcal{S} \times \mathcal{S} \to \mathbb{R}$ be a metric satisfying $\mathtt{D}(s, s' | \theta) \leq \mathtt{d}$ for a given $\mathtt{d} \geq 0$ for any $s,s'\in\mathcal{S}$ with $P_{\theta}(s)>0$ and $P_{\theta}(s')>0$, which defines when two secrets $s$ and $s'$ are considered \textit{adjacent}.
Let $\mathcal{Q}$ be the set of such adjacent secrets.

DCP aims to protect the privacy of the secret $S$ by ensuring that any adjacent secret pair from the output of a randomized mechanism $\mathcal{M}(\gamma)$, where $\gamma: \mathcal{X} \to \Delta(\mathcal{Y})$ denotes the probability density function of the underlying distribution of the mechanism.
The corresponding mapping $\mathcal{M}(\cdot; \gamma): \mathcal{X} \to \mathcal{Y}$ takes an input dataset $x \in \mathcal{X}$ (but not the secret $S$) and produces an output $y = \mathcal{M}(x; \gamma) \in \mathcal{Y}$.
%
% \textcolor{lightgray}{$$ggggggggggg$$}
%
% DCP aims to protect the privacy of the secret $S$ by ensuring that any adjacent secret pair from the output of a randomized mechanism $\mathcal{M}: \mathcal{X} \to \mathcal{Y}$, which takes the dataset $X$ as input but not the secret $S$.
% Let $\gamma: \mathcal{X} \to \Delta(\mathcal{Y})$ denote the probability density function of the underlying distribution of the mechanism $\mathcal{M}$. 
%
% We represent the mechanism by $\mathcal{M}(\gamma)$. The corresponding mapping $\mathcal{M}(\cdot; \gamma): \mathcal{X} \to \mathcal{Y}$ takes an input dataset $x \in \mathcal{X}$ and produces an output $y = \mathcal{M}(x; \gamma) \in \mathcal{Y}$.
%
For any measurable set $\mathcal{W} \subseteq \mathcal{Y}$, the probability is defined as:
\begin{equation}\label{eq:prior_likelihood}
    \textbf{Pr}^{\theta,s}_{\gamma}\left[Y \in \mathcal{W}\right] \equiv \int_{x \in \mathcal{X}} \int_{y \in \mathcal{W}} \gamma\left(y \middle| x\right) P_{\theta}(x|s) \, dy \, dx.
\end{equation}
For discrete spaces, the density functions are replaced by probability mass functions, and integrals are replaced by summations. This work focuses on the continuous case. For simplicity, we may omit $\theta$ in the notation $\textbf{Pr}^{\theta,s}_{\gamma}$ when $\theta$ is fixed and does not vary in the discussion.

\begin{definition}[$(\epsilon, \delta, \theta,\mathcal{G})$-Differential Confounding Privacy]
A randomized mechanism $\mathcal{M}(\gamma): \mathcal{X} \to \mathcal{Y}$ is \textit{$(\epsilon, \delta, \theta,\mathcal{G})$-differentially confounding private ($(\epsilon, \delta,\theta,\mathcal{G})$-DCP)} with $\epsilon \geq 0$ and $\delta \in [0,1]$ if $\mathcal{M}$ satisfies \textit{$(\epsilon, \delta, \theta,\mathcal{G})$-indistinguishability ($(\epsilon, \delta, \theta,\mathcal{G})\textup{-}\mathtt{ind}$)}:
\begin{equation}\label{eq:confounding_privacy_def}
    \sup_{(s,s')\in\mathcal{Q} } \sup_{\mathcal{W} \subset \mathcal{Y}} 
    \left( \textbf{Pr}^{s}_{\gamma}\left[Y \in \mathcal{W}\right] - e^{\epsilon} \textbf{Pr}^{s'}_{\gamma}\left[Y \in \mathcal{W}\right] \right) \leq \delta.
\end{equation}
% %
% For simplicity, we omit $\mathtt{D}$ and $\mathtt{d}$ in the notation and refer to the property as $(\epsilon, \delta)$-DCP and $(\epsilon, \delta)\textup{-}\mathtt{ind}$.
% %
\end{definition}

DCP ensures that any two adjacent secrets $(s, s') \in \mathcal{Q}$ link to nearly indistinguishable outputs, thereby limiting the information that can be inferred about the secret. Specifically, the $(\epsilon, \delta, \theta, \mathcal{G})\textup{-}\mathtt{ind}$ condition bounds the difference in probabilities of $\mathcal{M}$ producing an output in any measurable set $\mathcal{W} \subset \mathcal{Y}$ for $s$ and $s'$, up to an additive factor $\delta$ and multiplicative factor $e^\epsilon$. Smaller $\epsilon$ and $\delta$ imply stronger privacy guarantees.
For simplicity, we omit $\mathtt{D}$ and $\mathtt{d}$ in the notation and refer to the property as $(\epsilon, \delta)$-DCP and $(\epsilon, \delta)\textup{-}\mathtt{ind}$.

\textbf{Example:} \textit{Privacy Risks in Navigation Analytics. }
Let $\mathbf{x}$ be the data from a connected vehicle, including GPS logs, driving speed, acceleration patterns, daily mileage, time-of-day usage, and onboard sensor readouts. Let $\mathbf{s}$ be the driver’s hidden risk profile (e.g., “high likelihood of accident or reckless driving”), inferred by an insurer or telematics service using a machine-learning model.
A mechanism $\mathcal{M}(\gamma)$ performs a navigation analytics service that uses the raw data $\mathbf{x}$ to identify traffic congestion hotspots and suggest better route planning for municipal authorities or map providers. The output $\mathbf{y} = \mathcal{M}(\mathbf{x}; \gamma)$ might include publicly accessible heatmaps of traffic speeds, congestion levels, and dangerous intersections across the city.
Even without direct access to $\mathbf{x}$, an insurance company could leverage $\mathbf{y}$ to infer details about individual drivers. For instance, granular heatmaps could reveal where and when specific vehicles frequently drive or expose unique behaviors like repeated harsh braking on certain routes. By cross-referencing this information with partial records (e.g., billing data) and applying advanced machine-learning models, the insurer could refine risk scores or reconstruct the driver’s hidden risk profile $\mathbf{s}$.

In this example, the secret $s$ is inferred information derived from $x$. Consequently, DP becomes inapplicable, as it is designed to protect individual input data records rather than properties inferred from them.
% %
% Although mathematically the PP framework and the DCP appear similar, since the conditional density $P_{\theta}(x | s)$ in (\ref{eq:prior_likelihood}) is derived from the intrinsic probabilistic relationship between $x$ and $s$, rather than the adversary’s prior knowledge of the dataset, the standard PP framework does not directly apply to this scenario.

% Let dataset $x$ be the aggregate patient health metrics (e.g., as medical imaging and blood-glucose levels) and demographic details, and let the secret $s$ be the specific patient's disease risk profile (e.g., ``high risk of diabetes").
% The secret $s$ can be produced by a disease screening model using the dataset $x$; i.e., $s=\mathcal{G}(x)$.
% Consider that a county public health authority builds a population-level dashboard using aggregated lab results stored as $x$ from multiple clinic to track the prevalence of pre-diabetes or elevated glucose levels.
% The output $y=\mathcal{M}(x;\gamma)$ is "Area A has a $15\%$ increase in elevated glucose levels compared to last month". 

% In healthcare, tests and screenings generate datasets \(X\) (e.g., lab results, imaging data). The secret \(S\) is a diagnosis or health label inferred from \(X\), such as “high risk of condition \(S\).”
% The causal flow starts with raw data \(X\), which is analyzed to derive \(S\). Feedback loops occur when a tentative diagnosis triggers further tests, creating new data that updates \(X\). Since \(S\) depends on multiple variables in \(X\) and evolves dynamically, \(S \not\subseteq X\).

\section{Characterizing the Composition}

In this section, we characterized the composition properties of DCP mechanisms using copula and privacy loss random variables. 
We demonstrate that, unlike DP, DCP does not always exhibit composition as gracefully as DP when the same dataset (or overlapping datasets \cite{dwork2010boosting}) linked to the secret are used across multiple private data processing mechanisms.
The detailed version of this section is provided in Appendix \ref{app:section_III_details}.

\textbf{Composition }
In the DCP framework, as more computations are performed on the dataset $x$, it is important to understand how the privacy of the secret $s = \mathcal{G}(x)$ degrades under composition. 
Consider $k$ independent mechanisms $\mathcal{M}_{1}(\gamma_{1}), \mathcal{M}_{2}(\gamma_{2}), \dots, \mathcal{M}_{k}(\gamma_{k})$, where each $\mathcal{M}_{i}(\cdot;\gamma_{i}): \mathcal{X} \to \mathcal{Y}_{i}$ operates with a density function $\gamma_{i}$ and output space $\mathcal{Y}_{i}$. Let $\vec{\gamma} \equiv (\gamma_{i})_{i=1}^k$ denote the collection of density functions, $\vec{Y} \equiv (Y_{i})_{i=1}^k$ the random variables corresponding to the outputs of the mechanisms, and $\vec{y} \equiv (y_{i})_{i=1}^k \in \vec{\mathcal{Y}}\equiv \prod_{i=1}^k \mathcal{Y}_{i}$.
% %
% Consider $k$ independent mechanisms $\mathcal{M}_{1}(\gamma_{1}), \mathcal{M}_{2}(\gamma_{2}), \dots, \mathcal{M}_{k}(\gamma_{k})$, where each $\mathcal{M}_{i}(\cdot;\gamma_{i}): \mathcal{X} \to \mathcal{Y}_{i}$ operates with density function $\gamma_{i}$ and output space $\mathcal{Y}_{i}$. Let $\vec{\gamma} \equiv (\gamma_{i})_{i=1}^k$, $\vec{Y} \equiv (Y_{i})_{i=1}^k$ denote the random variables corresponding to the outputs of the mechanisms, and $\vec{y} \equiv (y_{i})_{i=1}^k \in \vec{\mathcal{Y}}$, where $\vec{\mathcal{Y}} \equiv \prod_{i=1}^k \mathcal{Y}_{i}$.
% %
The composition $\mathcal{M}(\cdot;\vec{\gamma}): \mathcal{X} \to \vec{\mathcal{Y}}$ is defined as:
\[
\mathcal{M}(x; \vec{\gamma}) = \left(\mathcal{M}(x; \gamma_{1}), \dots, \mathcal{M}(x; \gamma_{k})\right).
\]
For any $\vec{\mathcal{W}} \subseteq \vec{\mathcal{Y}}$, we define:
\begin{equation}\label{eq:likelihood_composition}
    \textbf{Pr}^{s}_{\vec{\gamma}}\left[\vec{Y} \in \vec{\mathcal{W}}\right] \equiv \int\limits_{x \in \mathcal{X}} \int\limits_{\vec{y} \in \vec{\mathcal{W}}} \prod_{i=1}^{k} \gamma_{i}(y_{i} \mid x) P_{\theta}(x|s) \, d\vec{y} \, dx.
\end{equation}
The ($k$-fold) composition $\mathcal{M}(\vec{\gamma})$ satisfies $(\epsilon_{g}, \delta_{g})$-DCP for some $\epsilon_{g} \geq 0$ and $\delta_{g} \in [0, 1]$ if,
\[
\sup_{(s,s')\in\mathcal{Q}} \sup_{\vec{\mathcal{W}} \subset \vec{\mathcal{Y}}} 
\left( \textbf{Pr}^{s}_{\vec{\gamma}}\left[\vec{Y} \in \vec{\mathcal{W}}\right] - e^{\epsilon_{g}} \textbf{Pr}^{s'}_{\vec{\gamma}}\left[\vec{Y} \in \vec{\mathcal{W}}\right] \right) \leq \delta_{g}.
\]
Given a $\delta_{g}$, the tightest or optimal value of $\epsilon_{g}$ is defined as
\begin{equation}\label{eq:def_opt_comp}
    \textup{OPT}\left(\mathcal{M}(\vec{\gamma}), \delta_{g}\right) \equiv \inf\left\{\epsilon' \geq 0 \mid \mathcal{M}(\vec{\gamma}) \textup{ is $(\epsilon', \delta_{g})\textup{-}\mathtt{ind}$}\right\}.
\end{equation}

\textbf{Effective Mechanism }
For each $\gamma_{i}: \mathcal{X} \to \Delta(\mathcal{Y}_{i})$, let $\psi_{i}: \mathcal{S} \to \Delta(\mathcal{Y}_{i})$ denote the underlying density function of $\textbf{Pr}^{s}_{\gamma}[\cdot]$ given by Equation~(\ref{eq:prior_likelihood}). Specifically,
\[
\psi_{i}(y_{i} \mid s) = \int_{x \in \mathcal{X}} \gamma_{i}(y_{i} \mid x) P_{\theta}(x|s)\, dx.
\]
Additionally, let $\mathcal{N}_{i}(\cdot;\psi_{i}): \mathcal{S} \to \mathcal{Y}_{i}$ denote the corresponding \textit{effective mechanism} for all $i \in [k]$. Since $\textbf{Pr}^{s}_{\psi_{i}}[\cdot] = \textbf{Pr}^{s}_{\gamma_{i}}[\cdot]$ for all $i$, the mechanism $\mathcal{N}_{i}(\psi_{i})$ is $(\epsilon_{i}, \delta_{i})\textup{-}\mathtt{ind}$ if and only if $\mathcal{M}_{i}(s; \gamma_{i})$ is $(\epsilon_{i}, \delta_{i})$-DCP.
Let $\Psi_{i}$ denote the cumulative distribution function (CDF) associated with the density $\psi_{i}$ of each effective mechanism $\mathcal{N}_{i}$.
The composition of the effective mechanisms $\mathcal{N}(\cdot;\vec{\psi}): \mathcal{S} \to \vec{\mathcal{Y}}$, defined as $\mathcal{N}(s;\vec{\psi}) = \left(\mathcal{N}_{1}(s; \psi_{1}), \dots, \mathcal{N}_{k}(s; \psi_{k})\right)$.

If we consider the secret as a single-point dataset (see Appendix \ref{app:DCP_single_point_DP} for discussion of signle-point DP), $\textup{OPT}\left(\mathcal{N}(s;\vec{\psi}), \delta_{g}\right)$ can be calculated according to the optimal composition theory of DP \cite{murtagh2015complexity}. However, it is generally a \#P-complete problem \cite{murtagh2015complexity}.

The concept of effective mechanisms is introduced to illustrate the connection and the difference between the composition properties of DP and those of DCP, in terms of \textit{privacy loss random variable}.

\textbf{Privacy Loss Random Variable }
Privacy loss under the $(\epsilon, \delta)$ framework can be quantified through the \textit{privacy loss random variable} (PLRV). Let $\mathtt{p}$ and $\mathtt{q}$ be two probability density functions on the output space $\mathcal{Y}$. Define the PLRV as 
\[
\textbf{L}_{\mathtt{p}, \mathtt{q}}(y) = \log \frac{\mathtt{p}(y)}{\mathtt{q}(y)}.
\]
The PLRV, $\textbf{L}_{\mathtt{p}, \mathtt{q}}(Y)$ for $Y \sim \mathtt{p}$, captures the distinguishability between outputs under $\mathtt{p}$ and $\mathtt{q}$.
Let $\psi^{s}_{i}(\cdot)=\psi_{i}(\cdot|s)\in\mathcal{Y}_{i}$.
Then, for any $(s_{0}, s_{1})\in\mathcal{Q}$, the individual PLRV of each mechanism $\mathcal{M}_{i}(\gamma_{i})$ is given by $\textbf{L}_{\psi^{s_{0}}_{i}, \psi^{s_{1}}_{i}}(Y_{i})$ with $Y_{i}\sim \psi^{s_{0}}_{i}$.

If we consider a deterministic, one-to-one relationship between $S$ and $X$, DCP becomes DP.
The PLRV of the composition $\mathcal{M}(\vec{\gamma})$ coincides with the PLRV of the composition of efficient mechanisms $\mathcal{N}(\vec{\psi})$.
That is,
\begin{equation}\label{eq:PLRV_independent}
    \textbf{L}^{\textup{id}}_{\vec{\psi}_{s_{0}}, \vec{\psi}_{s_{1}}}(\vec{Y}) = \sum\nolimits_{i=1}^{k} \textbf{L}_{\psi_{i}^{s_{0}}, \psi_{i}^{s_{1}}}(Y_{i}),
\end{equation}
where $\vec{\psi}_{s}(\cdot)=\vec{\psi}(\cdot|s)=\prod^{k}_{i=1}\psi_{i}(\cdot|s)$. 
The fact that privacy loss can be either positive or negative introduces cancellations under composition, which are generally challenging to track.

As we will show shortly using \textit{copulas}, under a general $(\theta, \mathcal{G})$, the PLRV of $\mathcal{M}(\vec{\gamma})$ differs from \(\textbf{L}^{\textup{id}}{\vec{\psi}{s_{0}}, \vec{\psi}{s{1}}}(\vec{Y})\), even when the mechanisms $\{\mathcal{M}_{i}(\gamma_{i})\}^{k}_{i=1}$ are independent.

\textbf{Copula }
For any $s \in \mathcal{S}$, let $B(\cdot | s)$ denote the cumulative distribution function (CDF) and $b(\cdot | s)$ the density function of $\textbf{Pr}^{s}_{\vec{\gamma}}$ as defined in Equation~(\ref{eq:likelihood_composition}) of Appendix \ref{app:section_III_details}. To simplify, we assume the outputs $y_{i} \in \mathcal{Y}_{i}$ are univariate for all $i \in [k]$.
By Sklar's theorem~\cite{sklar1959fonctions}, there exists a \textit{copula} $C$ for every $s \in \mathcal{S}$ such that:
\[
B\left(Y_{1}, \dots, Y_{k} \mid s\right) = C\left(\Psi_{1}\left(Y_{1} \mid s\right), \dots, \Psi_{k}\left(Y_{k} \mid s\right) \mid s\right),
\]
where $C: [0,1]^{k} \to [0,1]$ is a multivariate CDF with uniform univariate marginals:
\[
C(u_{1}, \dots, u_{k}) = \textbf{Pr}\left[U_{1} \leq u_{1}, \dots, U_{k} \leq u_{k}\right],
\]
and $U_{i} \in [0,1]$ is a standard uniform random variable $\forall i \in [k]$.

Let $c$ denote the \textit{copula density function} of $C$. Then, the density function $b$ of $B$ can be expressed as:
\[
b_{s}\left(\vec{Y}\right) = c\left(\Psi_{1}\left(Y_{1} \mid s\right), \dots, \Psi_{k}\left(Y_{k} \mid s\right) \mid s\right) \prod_{i=1}^{k} \psi_{i}\left(Y_{i} \mid s\right).
\]
Given $\{\theta, \mathcal{G}\}$, the PLRV of the composition $\mathcal{M}(\vec{\gamma})$ is expressed as:
\[
\textbf{L}_{b_{s_{0}}, b_{s_{1}}}(\vec{Y}) = \textbf{L}^{\mathcal{G}}_{c_{s_{0}}, c_{s_{1}}}(\vec{Y}) + \textbf{L}^{\textup{id}}_{\psi_{s_{0}}, \psi_{s_{1}}}(\vec{Y}),
\]
where $\textbf{L}^{\mathcal{G}}_{c_{s_{0}}, c_{s_{1}}}(\vec{Y})$ depends on the copula, while $\textbf{L}^{\textup{id}}_{\psi_{s_{0}}, \psi_{s_{1}}}(\vec{Y})$, given by (\ref{eq:PLRV_independent}), is independent of the copula.

Since $\textbf{L}^{\textup{id}}_{\psi_{s_{0}}, \psi_{s_{1}}}(\vec{Y})$ is a sum of independent random variables, the \textit{privacy loss distribution} (PLD) of $\textbf{L}^{\textup{id}}_{\psi_{s_{0}}, \psi_{s_{1}}}(\vec{Y})$ corresponds to the \textit{convolution} of the individual PLDs of $\{\textbf{L}_{\psi_{i}^{s_{0}}, \psi_{i}^{s_{1}}}(Y_{i})\}_{i=1}^{k}$, analogous to composition in DP frameworks \cite{sommer2019privacy}. 
In contrast, $\textbf{L}^{\mathcal{G}}_{c_{s_{0}}, c_{s_{1}}}(\vec{Y})$, which depends on the CDFs $\{\Psi_{i}\}_{i=1}^{k}$, generally disrupts this property, preventing the PLD of $\textbf{L}_{b_{s_{0}}, b_{s_{1}}}(\vec{Y})$ from being a straightforward convolution of the individual PLDs.

% Since $\textbf{L}^{\textup{id}}_{\psi_{s_{0}}, \psi_{s_{1}}}(\vec{Y})$ is a sum of independent random variables, the the \textit{privacy loss distribution} (PLD) of $\textbf{L}^{\textup{id}}_{\psi_{s_{0}}, \psi_{s_{1}}}(\vec{Y})$ corresponds to the \textit{convolution} of the individual PLDs of $\{\textbf{L}_{\psi_{i}^{s_{0}}, \psi_{i}^{s_{1}}}(Y_{i})\}^{k}_{i=1}$, similar to the composition in DP frameworks \cite{sommer2019privacy}.
% However, $\textbf{L}^{\mathcal{G}}_{c_{s_{0}}, c_{s_{1}}}(\vec{Y})$, which depends on the CDFs $\{\Psi_{i}\}_{i=1}^{k}$, generally disrupts this property, preventing the PLD of $\textbf{L}_{b_{s_{0}}, b_{s_{1}}}(\vec{Y})$ from being a straightforward convolution of the individual PLDs.

Theorem~\ref{thm:order_composition_pre} establishes that the composition of individual DCP mechanisms remains DCP but lacks the graceful behavior of DP composition due to the loss of convolution, even for independent mechanisms. The full discussion and quantitative bounds are provided in Appendix~\ref{app:section_III_details}.

% Theorem~\ref{thm:order_composition} first shows that composition of individual DCP mechanisms is still DCP, and then demonstrates that DCP composition does not exhibit the same graceful behavior as DP composition due to the lost of convolution even when individual mechanisms are independent. The full discussion and quantitative bounds are included in Appendix~\ref{app:section_III_details}.

\begin{theorem}\label{thm:order_composition_pre}
Let $\mathcal{M}(\vec{\gamma})$ be the composition of $\{\mathcal{M}_{i}(\gamma_{i})\}_{i=1}^{k}$ be independent mechanisms. 
Let $\mathcal{N}_{i}(\psi_{i})$ be the effective mechanism of $\mathcal{M}_{i}(\gamma_{i})$, $\forall i$.
Suppose $\textbf{L}^{\mathcal{G}}_{c_{s_{0}}, c_{s_{1}}}(\vec{Y}) \neq 0$. Then, the following holds:
\begin{itemize}
    \item[(i)] $\mathcal{M}(\vec{\gamma})$ is $(\epsilon_{g}, \delta_{g})$-DCP for some $\epsilon_{g}\geq 0$ and $\delta_{g}\in[0,1]$.
    \item[(ii)] $\mathcal{M}(\vec{\gamma})$ is not generally $(\sum_{i=1}^{k} \epsilon_{i}, \sum_{i=1}^{k} \delta_{i})$-DCP.
    \item[(iii)] $\textup{OPT}\left(\mathcal{N}(s;\vec{\psi}), \delta_{g}\right)$ strictly \textit{underestimates} the privacy risk of $\mathcal{M}(\vec{\gamma})$.
\end{itemize}
\end{theorem}

\section{Inverse Composition}

We propose \textit{inverse composition} (IC) to address challenges in privacy composition by formulating an optimization framework involving two decision-makers: a \textit{leader} and a \textit{follower}.
The leader employs a \textit{privacy strategy} $\alpha: \mathcal{S} \mapsto \Delta(\mathcal{Y}_\alpha)$, where $\mathcal{Y}_{\alpha}$ is the output space, to achieve any of two tasks:
\begin{enumerate}
    \item \textbf{Task-1:} Add a new mechanism $\mathcal{M}_{k+1}(\gamma_\alpha):\mathcal{X}\mapsto \mathcal{Y}_{\alpha}$ to $\mathcal{M}(\vec{\gamma})$ and ensure the new composition, denoted by $\mathcal{M}(\vec{\gamma}, \gamma_\alpha):\mathcal{X}\mapsto \vec{\mathcal{Y}}\times\mathcal{Y}_{\alpha}$, satisfies a target $(\epsilon_g, \delta_g)$-DCP constraint, where $\gamma_{\alpha}$ is determined by $\alpha$.
    \item \textbf{Task-2:} Determine the global privacy parameters $(\epsilon_g, \delta_g)$ of the existing composition $\mathcal{M}(\vec{\gamma}) = (\mathcal{M}_1(\gamma_1), \dots, \mathcal{M}_k(\gamma_k))$
    without adding a new mechanism.
\end{enumerate}
By observing $\vec{y}_{\alpha}=(\vec{y}, y_{\alpha})\in\vec{\mathcal{Y}}_{\alpha}\equiv\vec{\mathcal{Y}}\times \mathcal{Y}_{\alpha}$, the follower adopts a \textit{response strategy} $\pi: \vec{\mathcal{Y}}_{\alpha} \mapsto \Delta(\mathcal{S})$, reacting optimally to the leader's decisions by obtaining a probability distribution (inference) over the secrets.
In both tasks, the existing $k$ mechanisms may exhibit arbitrary dependencies, captured by $\mathcal{C}$. Additionally, the mechanism $\mathcal{M}_{k+1}(\gamma_{\alpha})$ can either intrinsically depend on the existing $k$ mechanisms or introduce further dependencies through $\alpha$, e.g., via \textit{copula perturbation} (see Appendix F).
In this section, we focus on the case when $\gamma_{\alpha}$ s independent of $\vec{\gamma}$.

The idea of IC is as follows. 
The privacy strategy $\alpha$ introduces an additional PLRV $\textbf{L}_{\alpha_{s_{0}}, \alpha_{s_{1}}}$, resulting in the total PLRV:
\begin{equation}\label{eq:PLRV_perturb}
    \textbf{L}^{\alpha}_{s_{0}, s_{1}}(\vec{Y}, Y_{\alpha}) = \textbf{L}_{\hat{b}_{s_{0}}, \hat{b}_{s_{1}}}(\vec{Y}) + \textbf{L}_{\alpha_{s_{0}}, \alpha_{s_{1}}}(Y_{\alpha}),
\end{equation}
where $\textbf{L}_{\hat{b}_{s_{0}}, \hat{b}_{s_{1}}}(\vec{Y})$, given by (\ref{eq:PLRV_dependent}) of Appendix \ref{app:section_III_details}, is fixed but unknown. 
Thus, changes in $\textbf{L}^{\alpha}{s_{0}, s_{1}}$ are driven by $\textbf{L}_{\alpha_{s_{0}}, \alpha_{s_{1}}}$, with $\alpha$ also influencing the response strategy $\pi$.
Traditional \textit{forward} approaches rely on generally challenging worst-case proofs, including sensitivity calculations (NP-hard \cite{xiao2008output}) required for the proper distributions of $\alpha$, and combinatorial analysis to compute $\epsilon_{g}$ and $\delta_{g}$ (\#P-complete even for independent composition \cite{murtagh2015complexity}). 
In contrast, IC designs $\alpha$ to elicit a follower response $\pi$ with certain properties. The key insight is that if $\alpha$ induces a $\pi$ satisfying these properties, the corresponding $\alpha$ will lead to $\textbf{L}^{\alpha}_{s_{0}, s_{1}}$ that \textit{automatically} achieves the desired privacy guarantees, without requiring explicit constraints on $\alpha$.

By doing so, IC ensures that the resulting privacy strategy $\alpha$ inherently fulfills \textbf{Task-1} and \textbf{Task-2} without requiring explicit constraints on $\alpha$, thereby avoiding the need for worst-case proofs.
We formulate IC as a constrained optimization problem.
The leader and the following choose $\alpha$ and $\pi$ by minimizing an expected loss based on any \textit{strictly proper scoring rule} (SPSR) $\mathcal{L}$ taking the form \cite{savage1971elicitation}
\begin{equation}
    \begin{aligned}
        \mathcal{L}(\pi, s,\vec{y}_{\alpha}) \equiv \Phi(\pi) + \Upsilon(\pi(s|\vec{y}_{\alpha})),
    \end{aligned}
\end{equation}
where $\Phi(\cdot)$ is strict convex, and $\Upsilon$ is some pointwise function.
In addition, for any valid $p: \vec{\mathcal{Y}}_{\alpha} \to \Delta(\mathcal{S})$, the expected loss satisfies 
\[
\mathbb{E}_{\alpha, \vec{\gamma}}^p[\mathcal{L}(p, S, \vec{Y}_\alpha)] \leq \mathbb{E}_{\alpha, \vec{\gamma}}^p[\mathcal{L}(\pi, S, \vec{Y}_\alpha)],
\]
where the expectation captures randomness in $S$ and $\vec{Y}_\alpha$.
It is well-known that $\pi^{*}=p$ is the \textit{unique} minimizer of $\mathbb{E}_{\alpha, \vec{\gamma}}^p[\mathcal{L}(\pi, S, \vec{Y}_\alpha)]$ \cite{gneiting2007strictly}.

For $\tau_{g}\geq 1$ and $\delta_{g}\in(0,1]$, define:
\[
\Pi[\tau_{g}, \delta_{g}] \equiv \left\{\pi \middle| 
\begin{aligned}
    &\pi(s \mid \vec{y}_{\alpha}) \geq \tau^{-1}_{g} P_{\theta}(s), \quad \forall s, \vec{y}_{\alpha}, \\
    &\mathbb{E}_{\pi}\left[\frac{\pi(S \mid \vec{y}_{\alpha})}{P_{\theta}(S)}\right] \leq \delta_{g}\tau_{g}, \quad \forall \vec{y}_{\alpha}
\end{aligned}
\right\},
\]
where the expectation captures randomness in $S$ by $\pi(\cdot \mid \vec{y}_{\alpha})$.
%
% Here, for any $\tau_{g}\geq 1$, there exists $\epsilon_{g}\geq 0$ such that $\tau_{g} = e^{\epsilon}$.
%
Additionally, for $\delta_{g}=0$, define:
\[
\Pi[\tau_{g}, 0] \equiv \left\{\pi \middle| 
    \tau^{-1}_{g} P_{\theta}(s) \leq \pi(s \mid \vec{y}_{\alpha}) \leq \tau_{g} P_{\theta}(s), \quad \forall s, \vec{y}_{\alpha}
\right\}.
\]
\textbf{Convex Optimization}
Given $\tau_{g} \geq 1$ and $\delta_{g} \in [0,1]$, consider the following two optimization problems:
\begin{equation}\tag{\texttt{Task-1}}\label{eq:IC_implement}
    \min_{\alpha, \pi} \mathbb{E}^{\mu}_{\alpha,\vec{\gamma}}\left[\mathcal{L}(\pi, S, \vec{Y}_{\alpha})\right], \text{ s.t. } \pi \in \Pi[\tau_{g}, \delta_{g}],
\end{equation}
\begin{equation}\tag{\texttt{Task-2}}\label{eq:IC_tight}
    \min_{\tau_{g}, \alpha, \pi} \mathbb{E}^{\mu}_{\alpha,\vec{\gamma}}\left[\mathcal{L}(\pi, S, \vec{Y}_{\alpha})\right], \text{ s.t. } \pi \in \Pi[\tau_{g}, \delta_{g}].
\end{equation}
The posterior $\mu$ need not be convex in $\alpha$. However, if we choose an SPSR $\hat{\mathcal{L}}(\pi,s,\vec{y}_{\alpha})
= -\log\bigl(\pi(s\mid\vec{y}_{\alpha})\bigr)\,P_{\alpha,\vec\gamma}(\vec{y}_{\alpha})$, where $P_{\alpha,\vec{\gamma}}(\vec{y}_{\alpha}) = \sum_{s}\alpha(y_{\alpha}\mid s)\,P_{\theta}(s)$, then the factor $P_{\alpha,\vec{\gamma}}(\vec{y}_{\alpha})$ cancels inside the expectation $\mathbb{E}^{\mu}_{\alpha,\vec{\gamma}}[\mathcal{L}]$, leaving only the product $\alpha(y_{\alpha}\mid s)\,P_{\theta}(s)$. Moreover, $\Pi[\tau_{g},\delta_{g}]$ is convex in $\pi$ and $\tau_{g}$ and does not depend on $\alpha$. Hence, with $\mathcal{L}=\hat{\mathcal{L}}$, both (\ref{eq:IC_implement}) and (\ref{eq:IC_tight}) reduce to convex programs.
%
% The posterior $\mu$ is generally not convex in $\alpha$. However, for certain choices of SPSR $\mathcal{L}$, the explicit formulation of $\mu$ can be eliminated, making $\mathcal{L}$ convex in $\alpha$ and $\pi$. For example, consider $\hat{\mathcal{L}}(\pi, s, \vec{y}_{\alpha}) = -\log(\pi(s \mid \vec{y}_{\alpha})) P_{\alpha,\vec{\gamma}}(\vec{y}_{\alpha}),$
% where \( P_{\alpha,\vec{\gamma}}(\vec{y}_{\alpha}) = \sum_{s} \alpha(y_{\alpha} \mid s)P_{\theta}(s) \) is the total probability, also the denominator of $\mu$. 
% In this case, the expectation $\mathbb{E}^{\mu}_{\alpha,\vec{\gamma}}\left[\mathcal{L}(\pi, s, \vec{y}_{\alpha})\right]$ cancels \( P_{\alpha,\vec{\gamma}}(\vec{y}_{\alpha}) \), reducing $\mu$ to the product of $\alpha$ and the prior $p_{\theta}$. Additionally, the set $\Pi[\tau_{g}, \delta_{g}]$ is convex in $\pi$ and $\tau_{g}$ and independent of $\alpha$. 
% Thus, by choosing $\mathcal{L} = \hat{\mathcal{L}}$, both (\ref{eq:IC_implement}) and (\ref{eq:IC_tight}) become convex optimization.

\begin{theorem}\label{thm:inverse_composition}
Given $\{\theta,\mathcal{G}\}$, let $\{\mathcal{M}_{i}(\gamma_{i})\}_{i=1}^{k}$ be $k$ mechanisms, each satisfying $(\epsilon_{i}, \delta_{i})$-DCP. 
Let $\varepsilon(\tau_{g}, \theta)\equiv \log\left( 1 + \frac{\tau_{g}-1}{P^{*}_{\theta}}\right)$ with $P^{*}_{\theta}=\min_{s}P_{\theta}(s)> 0$.
When $\delta_{g} \neq 0$, let $\alpha^{*}$ be a solution to (\ref{eq:IC_implement}) and $\tau^{*}_{g}$ a solution to (\ref{eq:IC_tight}). Then, for any $\delta_{g} \in [0,1]$, the following hold:
\begin{itemize}
    \item[(i)] (\textbf{Task-1}) For any $\tau_{g} \geq 1$, if $\Pi[\tau_{g}, \delta_{g}] \neq \emptyset$, the composition $\mathcal{M}(\vec{\gamma}, \gamma_{\alpha^{*}})$ is $(\varepsilon(\tau_{g}, \theta), \delta_{g})$-DCP.

    \item[(ii)] (\textbf{Task-2}) The composition $\mathcal{M}(\vec{\gamma})$ is $(\varepsilon(\tau^{*}_{g}, \theta), \delta_{g})$-DCP.

    \item[(iii)] When $\delta_{g} = 0$, $\mathcal{M}(\vec{\gamma}, \gamma_{\alpha^{\dagger}})$ is $\varepsilon(\tau_{g},\theta)$-DCP if $\Pi[\tau_{g}, 0] \neq \emptyset$, for $\tau_{g}\geq 1$, and $\alpha^{\dagger}$ is a solution to (\ref{eq:IC_implement}). Furthermore, $\mathcal{M}(\vec{\gamma})$ is $\varepsilon(\tau^{\dagger}_{g},\theta)$-DCP if and only if $\tau^{\dagger}_{g}$ is a solution to (\ref{eq:IC_tight}).
\end{itemize}
\end{theorem}

Theorem~\ref{thm:inverse_composition} establishes that the privacy strategy $\alpha$ can be designed via IC to achieve \textbf{Task-1} and \textbf{Task-2} by eliciting an optimal response strategy $\pi$. When $\delta_{g} \neq 0$, the optimal $\pi$ satisfying the conditions of $\Pi[\tau_{g}, \delta_{g}]$ is sufficient for the corresponding $\alpha$ to achieve $(\varepsilon(\tau_{g},\theta), \delta_{g})$-DCP. In \textbf{Task-1}, $\textbf{L}_{\alpha_{s_{0}}, \alpha_{s_{1}}}$ is constructed to ensure that $\textbf{L}^{\alpha}_{s_{0}, s_{1}}$ satisfies $(\varepsilon(\tau_{g},\theta), \delta_{g})$-DCP without wasting privacy budget. In \textbf{Task-2}, $\textbf{L}_{\alpha_{s_{0}}, \alpha_{s_{1}}} = 0$, and the optimal $\epsilon^{*}_{g}=\varepsilon(\tau^{*}_{g},\theta)$ represents the tightest privacy parameter for $\textbf{L}^{\alpha}_{s_{0}, s_{1}}$.
When $\delta_{g} = 0$, the optimal $\pi$ satisfying the conditions of $\Pi[\tau_{g}, 0]$ is sufficient for $\alpha$ to achieve $(\varepsilon(\tau_{g},\theta), 0)$-DCP. 
However, these conclusions assume an ideal scenario where the convex optimization problems can be solved exactly. Since both (\ref{eq:IC_implement}) and (\ref{eq:IC_tight}) are infinite-dimensional, they cannot be solved directly.
The exploration of algorithmic and computational methods to address these challenges is left as future work.

\subsection{Numerical Experiments}

\begin{figure}[htb]
    \centering
    % First row of figures
    \begin{subfigure}[b]{0.23\textwidth}
        \includegraphics[width=\textwidth]{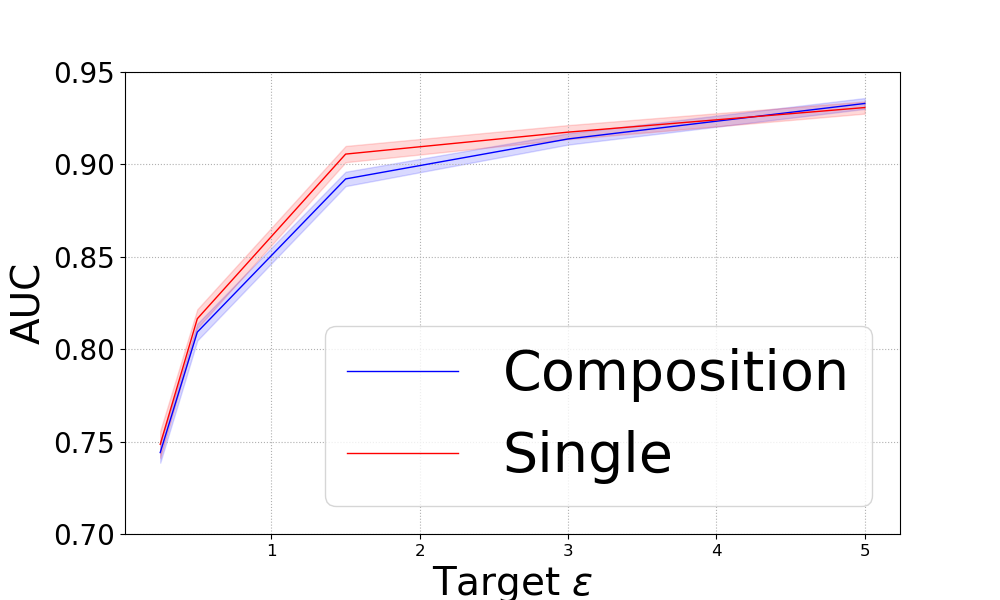}
        \caption{{\small Independent Perturbation}}
        \label{fig:fig1}
    \end{subfigure}
    %\hfill % adds horizontal space between figures
    \begin{subfigure}[b]{0.23\textwidth}
        \includegraphics[width=\textwidth]{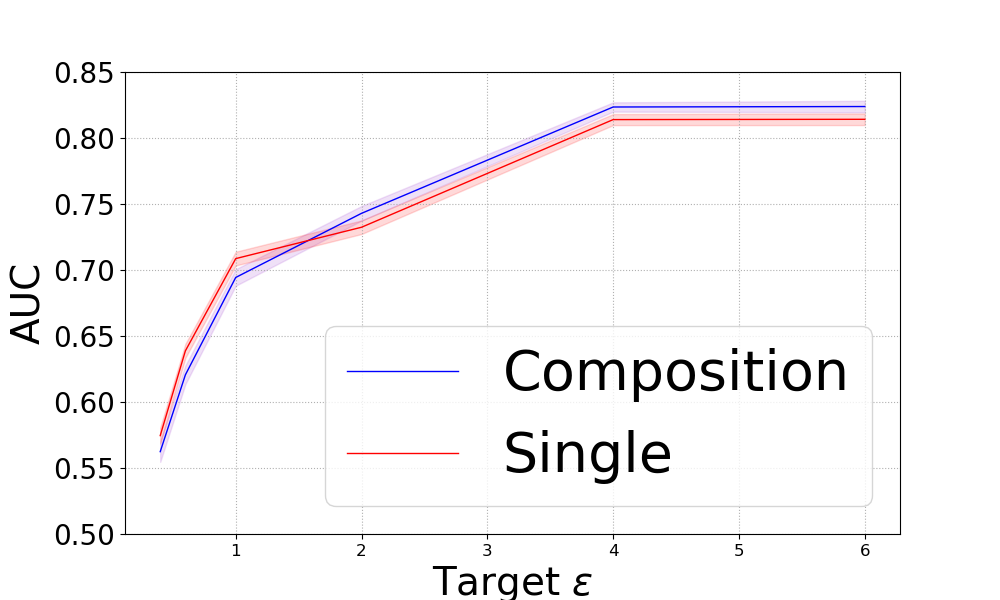}
        \caption{{\small Copula Perturbation}}
        \label{fig:fig2}
    \end{subfigure}
    \caption{\small (a) shows five tests with $\delta_{g} = \delta_{i} = 0.02$ and \(\epsilon_{g} = \{0.25, 0.25, 1.5, 3, 5\}\), corresponding to \(\epsilon_{i} = \{0.05, 0.1, 0.3, 0.6, 1\}\) for \(i \in \{1,2, 3, 4, 5\}\).
(b) shows six tests with $\delta_{g} = \delta_{i} = 0.02$ and \(\epsilon_{g} = \{0.4, 0.6, 1, 2, 4, 6\}\), corresponding to \(\epsilon_{i} = \{0.05, 0.1, 0.18, 0.3, 0.6, 1\}\) for \(i \in \{1, 2, 3, 4, 5\}\).
    % (a) shows five tests with fixed $\delta_{g} = \delta_{i} = 0.02$: $\epsilon^{1}_{g} = 0.25$ with $\epsilon^{1}_{i} = 0.05$, $\epsilon^{2}_{g} = 0.25$ with $\epsilon^{2}_{i} = 0.1$, $\epsilon^{3}_{g} = 1.5$ with $\epsilon^{3}_{i} = 0.3$, $\epsilon^{4}_{g} = 3$ with $\epsilon^{4}_{i} = 0.6$, and $\epsilon^{5}_{g} = 5$ with $\epsilon^{5}_{i} = 1$, for $i \in \{2, 3, 4, 5\}$.
    % (b) shows six tests with fixed $\delta_{g}=\delta_{i}=0.02$: $\epsilon^{1}_{g} = 0.4$ with $\epsilon^{1}_{i} = 0.05$, $\epsilon^{2}_{g} = 0.6$ with $\epsilon^{2}_{i} = 0.1$, $\epsilon^{3}_{g} = 1$ with $\epsilon^{3}_{i} = 0.18$, $\epsilon^{4}_{g} = 2$ with $\epsilon^{4}_{i} = 0.3$, $\epsilon^{5}_{g} = 4$ with $\epsilon^{5}_{i} = 0.6$, and $\epsilon^{5}_{g} = 6$ with $\epsilon^{5}_{i} = 1$,
    % for $i \in \{1,2, 3, 4, 5\}$.  
    }
    \label{fig:six_figures}
\end{figure}

In this section, we validate our IC approach through numerical experiments. 
The constrained optimization problem (inverse composition) is solved using the penalty method. 
The dataset used is from the 2016 iDASH Workshop on Privacy and Security \cite{tang2016idash}, derived from the 1000 Genomes Project \cite{10002015global}, consisting of SNVs from 800 individuals on Chromosome 10. All experiments were conducted on an NVIDIA A40 48G GPU using PyTorch. Detailed experimental settings are provided in Appendix M.

In the first experiment, we demonstrate how IC optimally adds a fifth mechanism, randomized by \(\alpha\) using output perturbation, so that the composition of five DCP mechanisms achieves the target privacy budget \((\epsilon_{g}, \delta_{g})\). 
In the second experiment, we show how IC performs \textit{copula perturbation} (see Appendix \ref{app:gaussian_copula_per}) to achieve the target privacy budget \((\epsilon_{g}, \delta_{g})\) for five DCP mechanisms. Copula perturbation using \(\alpha\): (1) manipulates the dependencies between two chosen mechanisms (e.g., \(\mathcal{M}_{1}(\gamma_{1})\) and \(\mathcal{M}_{2}(\gamma_{2})\)), and (2) transforms \(\textbf{L}_{\alpha_{s_{0}}, \alpha_{s_{1}}}\) in (\ref{eq:PLRV_perturb}) into a copula PLRV with \(\gamma_{1}\) and \(\gamma_{2}\) as marginals.

We refer to these experiments as \textit{independent perturbation} and \textit{copula perturbation}, respectively. The validity of IC is assessed by comparing the privacy loss of the composed mechanisms to that of a single \((\epsilon_{g}, \delta_{g})\)-DCP mechanism under the same membership inference attack. 
Privacy loss is evaluated through the Area Under the Curve (AUC) of membership inference attacks, which audit the extent of data leakage in machine learning models \cite{yeom2018privacy,zanella2023bayesian}. As we can see from Fig.~\ref{fig:six_figures}, the AUC under IC-based composition is comparable to that of the single mechanism, confirming that IC efficiently implements the composition.

\section{Conclusion}

We introduced Differential Confounding Privacy (DCP), extending differential privacy to settings involving complex relationships between secrets and datasets. While DCP mechanisms are composable, they lack the graceful composition properties of DP. To address this, we proposed the Inverse Composition (IC) framework, which enables systematic design of privacy strategies. Future work includes algorithmic development of IC and exploration of its approximations.

\bibliographystyle{IEEEtran}
\bibliography{references}

\begin{thebibliography}{10}
\providecommand{\url}[1]{#1}
\csname url@samestyle\endcsname
\providecommand{\newblock}{\relax}
\providecommand{\bibinfo}[2]{#2}
\providecommand{\BIBentrySTDinterwordspacing}{\spaceskip=0pt\relax}
\providecommand{\BIBentryALTinterwordstretchfactor}{4}
\providecommand{\BIBentryALTinterwordspacing}{\spaceskip=\fontdimen2\font plus
\BIBentryALTinterwordstretchfactor\fontdimen3\font minus \fontdimen4\font\relax}
\providecommand{\BIBforeignlanguage}[2]{{%
\expandafter\ifx\csname l@#1\endcsname\relax
\typeout{** WARNING: IEEEtran.bst: No hyphenation pattern has been}%
\typeout{** loaded for the language `#1'. Using the pattern for}%
\typeout{** the default language instead.}%
\else
\language=\csname l@#1\endcsname
\fi
#2}}
\providecommand{\BIBdecl}{\relax}
\BIBdecl

\bibitem{dwork2006calibrating}
C.~Dwork, F.~McSherry, K.~Nissim, and A.~Smith, ``Calibrating noise to sensitivity in private data analysis,'' in \emph{Theory of Cryptography: Third Theory of Cryptography Conference, TCC 2006, New York, NY, USA, March 4-7, 2006. Proceedings 3}.\hskip 1em plus 0.5em minus 0.4em\relax Springer, 2006, pp. 265--284.

\bibitem{kifer2014pufferfish}
D.~Kifer and A.~Machanavajjhala, ``Pufferfish: A framework for mathematical privacy definitions,'' \emph{ACM Transactions on Database Systems (TODS)}, vol.~39, no.~1, pp. 1--36, 2014.

\bibitem{desfontaines2019sok}
D.~Desfontaines and B.~Pej{\'o}, ``Sok: differential privacies,'' \emph{arXiv preprint arXiv:1906.01337}, 2019.

\bibitem{dwork2022differential}
C.~Dwork, ``Differential privacy: getting more for less,'' in \emph{Proc. Int. Cong. Math}, vol.~6, 2022, pp. 4740--4761.

\bibitem{kairouz2015composition}
P.~Kairouz, S.~Oh, and P.~Viswanath, ``The composition theorem for differential privacy,'' in \emph{International conference on machine learning}.\hskip 1em plus 0.5em minus 0.4em\relax PMLR, 2015, pp. 1376--1385.

\bibitem{murtagh2015complexity}
J.~Murtagh and S.~Vadhan, ``The complexity of computing the optimal composition of differential privacy,'' in \emph{Theory of Cryptography Conference}.\hskip 1em plus 0.5em minus 0.4em\relax Springer, 2015, pp. 157--175.

\bibitem{dwork2006our}
C.~Dwork, K.~Kenthapadi, F.~McSherry, I.~Mironov, and M.~Naor, ``Our data, ourselves: Privacy via distributed noise generation,'' in \emph{Advances in Cryptology-EUROCRYPT 2006: 24th Annual International Conference on the Theory and Applications of Cryptographic Techniques, St. Petersburg, Russia, May 28-June 1, 2006. Proceedings 25}.\hskip 1em plus 0.5em minus 0.4em\relax Springer, 2006, pp. 486--503.

\bibitem{xiao2008output}
X.~Xiao and Y.~Tao, ``Output perturbation with query relaxation,'' \emph{Proceedings of the VLDB Endowment}, vol.~1, no.~1, pp. 857--869, 2008.

\bibitem{zhang2024confounding}
T.~Zhang, B.~A. Malin, N.~Raviv, and Y.~Vorobeychik, ``Confounding privacy and inverse composition,'' \emph{arXiv preprint arXiv:2408.12010}, 2024.

\bibitem{he2014blowfish}
X.~He, A.~Machanavajjhala, and B.~Ding, ``Blowfish privacy: Tuning privacy-utility trade-offs using policies,'' in \emph{Proceedings of the 2014 ACM SIGMOD international conference on Management of data}, 2014, pp. 1447--1458.

\bibitem{song2017pufferfish}
S.~Song, Y.~Wang, and K.~Chaudhuri, ``Pufferfish privacy mechanisms for correlated data,'' in \emph{Proceedings of the 2017 ACM International Conference on Management of Data}, 2017, pp. 1291--1306.

\bibitem{zhang2022attribute}
W.~Zhang, O.~Ohrimenko, and R.~Cummings, ``Attribute privacy: Framework and mechanisms,'' in \emph{Proceedings of the 2022 ACM Conference on Fairness, Accountability, and Transparency}, 2022, pp. 757--766.

\bibitem{nuradha2023pufferfish}
T.~Nuradha and Z.~Goldfeld, ``Pufferfish privacy: An information-theoretic study,'' \emph{IEEE Transactions on Information Theory}, 2023.

\bibitem{dwork2010boosting}
C.~Dwork, G.~N. Rothblum, and S.~Vadhan, ``Boosting and differential privacy,'' in \emph{2010 IEEE 51st Annual Symposium on Foundations of Computer Science}.\hskip 1em plus 0.5em minus 0.4em\relax IEEE, 2010, pp. 51--60.

\bibitem{vadhan2021concurrent}
S.~Vadhan and T.~Wang, ``Concurrent composition of differential privacy,'' in \emph{Theory of Cryptography: 19th International Conference, TCC 2021, Raleigh, NC, USA, November 8--11, 2021, Proceedings, Part II 19}.\hskip 1em plus 0.5em minus 0.4em\relax Springer, 2021, pp. 582--604.

\bibitem{lyu2022composition}
X.~Lyu, ``Composition theorems for interactive differential privacy,'' \emph{Advances in Neural Information Processing Systems}, vol.~35, pp. 9700--9712, 2022.

\bibitem{vadhan2023concurrent}
S.~Vadhan and W.~Zhang, ``Concurrent composition theorems for differential privacy,'' in \emph{Proceedings of the 55th Annual ACM Symposium on Theory of Computing}, 2023, pp. 507--519.

\bibitem{mironov2017renyi}
I.~Mironov, ``R{\'e}nyi differential privacy,'' in \emph{2017 IEEE 30th computer security foundations symposium (CSF)}.\hskip 1em plus 0.5em minus 0.4em\relax IEEE, 2017, pp. 263--275.

\bibitem{balle2018improving}
B.~Balle and Y.-X. Wang, ``Improving the gaussian mechanism for differential privacy: Analytical calibration and optimal denoising,'' in \emph{International Conference on Machine Learning}.\hskip 1em plus 0.5em minus 0.4em\relax PMLR, 2018, pp. 394--403.

\bibitem{dong2022gaussian}
J.~Dong, A.~Roth, and W.~J. Su, ``Gaussian differential privacy,'' \emph{Journal of the Royal Statistical Society Series B: Statistical Methodology}, vol.~84, no.~1, pp. 3--37, 2022.

\bibitem{sommer2019privacy}
D.~M. Sommer, S.~Meiser, and E.~Mohammadi, ``Privacy loss classes: The central limit theorem in differential privacy,'' \emph{Proceedings on Privacy Enhancing Technologies}, vol. 2019, no.~2, pp. 245--269.

\bibitem{koskela2020computing}
A.~Koskela, J.~J{\"a}lk{\"o}, and A.~Honkela, ``Computing tight differential privacy guarantees using fft,'' in \emph{International Conference on Artificial Intelligence and Statistics}.\hskip 1em plus 0.5em minus 0.4em\relax PMLR, 2020, pp. 2560--2569.

\bibitem{bun2016concentrated}
M.~Bun and T.~Steinke, ``Concentrated differential privacy: Simplifications, extensions, and lower bounds,'' in \emph{Theory of Cryptography Conference}.\hskip 1em plus 0.5em minus 0.4em\relax Springer, 2016, pp. 635--658.

\bibitem{abadi2016deep}
M.~Abadi, A.~Chu, I.~Goodfellow, H.~B. McMahan, I.~Mironov, K.~Talwar, and L.~Zhang, ``Deep learning with differential privacy,'' in \emph{Proceedings of the 2016 ACM SIGSAC conference on computer and communications security}, 2016, pp. 308--318.

\bibitem{wang2019subsampled}
Y.-X. Wang, B.~Balle, and S.~P. Kasiviswanathan, ``Subsampled r{\'e}nyi differential privacy and analytical moments accountant,'' in \emph{The 22nd international conference on artificial intelligence and statistics}.\hskip 1em plus 0.5em minus 0.4em\relax PMLR, 2019, pp. 1226--1235.

\bibitem{gopi2021numerical}
S.~Gopi, Y.~T. Lee, and L.~Wutschitz, ``Numerical composition of differential privacy,'' \emph{Advances in Neural Information Processing Systems}, vol.~34, pp. 11\,631--11\,642, 2021.

\bibitem{zhu2022optimal}
Y.~Zhu, J.~Dong, and Y.-X. Wang, ``Optimal accounting of differential privacy via characteristic function,'' in \emph{International Conference on Artificial Intelligence and Statistics}.\hskip 1em plus 0.5em minus 0.4em\relax PMLR, 2022, pp. 4782--4817.

\bibitem{koskela2022individual}
A.~Koskela, M.~Tobaben, and A.~Honkela, ``Individual privacy accounting with gaussian differential privacy,'' \emph{arXiv preprint arXiv:2209.15596}, 2022.

\bibitem{sklar1959fonctions}
M.~Sklar, ``Fonctions de r{\'e}partition {\`a} n dimensions et leurs marges,'' in \emph{Annales de l'ISUP}, vol.~8, no.~3, 1959, pp. 229--231.

\bibitem{savage1971elicitation}
L.~J. Savage, ``Elicitation of personal probabilities and expectations,'' \emph{Journal of the American Statistical Association}, vol.~66, no. 336, pp. 783--801, 1971.

\bibitem{gneiting2007strictly}
T.~Gneiting and A.~E. Raftery, ``Strictly proper scoring rules, prediction, and estimation,'' \emph{Journal of the American statistical Association}, vol. 102, no. 477, pp. 359--378, 2007.

\bibitem{tang2016idash}
H.~Tang, X.~Wang, S.~Wang, and X.~Jiang, ``Idash privacy and security workshop,'' 2016.

\bibitem{10002015global}
{1000 Genomes Project Consortium} \emph{et~al.}, ``A global reference for human genetic variation,'' \emph{Nature}, vol. 526, no. 7571, p.~68, 2015.

\bibitem{yeom2018privacy}
S.~Yeom, I.~Giacomelli, M.~Fredrikson, and S.~Jha, ``Privacy risk in machine learning: Analyzing the connection to overfitting,'' in \emph{2018 IEEE 31st computer security foundations symposium (CSF)}.\hskip 1em plus 0.5em minus 0.4em\relax IEEE, 2018, pp. 268--282.

\bibitem{zanella2023bayesian}
S.~Zanella-Beguelin, L.~Wutschitz, S.~Tople, A.~Salem, V.~R{\"u}hle, A.~Paverd, M.~Naseri, B.~K{\"o}pf, and D.~Jones, ``Bayesian estimation of differential privacy,'' in \emph{International Conference on Machine Learning}.\hskip 1em plus 0.5em minus 0.4em\relax PMLR, 2023, pp. 40\,624--40\,636.

\bibitem{steinke2022composition}
T.~Steinke, ``Composition of differential privacy \& privacy amplification by subsampling,'' \emph{arXiv preprint arXiv:2210.00597}, 2022.

\bibitem{triastcyn2020bayesian}
A.~Triastcyn and B.~Faltings, ``Bayesian differential privacy for machine learning,'' in \emph{International Conference on Machine Learning}.\hskip 1em plus 0.5em minus 0.4em\relax PMLR, 2020, pp. 9583--9592.

\bibitem{blackwell1951comparison}
D.~Blackwell \emph{et~al.}, ``Comparison of experiments,'' in \emph{Proceedings of the second Berkeley symposium on mathematical statistics and probability}, vol.~1, no. 93-102, 1951, p.~26.

\bibitem{de2018blackwell}
H.~de~Oliveira, ``Blackwell's informativeness theorem using diagrams,'' \emph{Games and Economic Behavior}, vol. 109, pp. 126--131, 2018.

\end{thebibliography}

\newpage
\appendix

\subsection{Table of Notations}\label{app:notation}

\begin{table}[h]
\centering
\caption{Notations for Section II}
\begin{tabular}{ll}
\hline
Symbol & Description \\
\hline
$S$, $s$, $\mathcal{S}$& random variable (RV), sample, and space of secret\\
$X$, $x$, $\mathcal{X}$& RV, sample, and space of dataset\\
$Y$, $y$, $\mathcal{Y}$& RV, sample, and space of output\\
$\theta\in\Delta(\mathcal{S}\times\mathcal{X})$ & joint probability of secret and dataset \\
$\mathcal{G}:\mathcal{X}\mapsto \mathcal{S}$ & a randomized mapping derived from $\theta$\\
$P_{\theta}(s,x)$ &  joint density of $s$ and $x$\\
$\mathtt{D}(s,s'|\theta)\leq \mathtt{d}$ &  a metric to define the adjacency between $s$ and $s'$\\
$\mathcal{Q}$ & set of adjacent states specified by $\mathtt{D}(\cdot|\theta)$ and $\mathtt{d}$\\
$\mathcal{M}(\gamma)$ & randomized mechanism with density function $\gamma$\\
                      & $y=\mathcal{M}(x;\gamma)$ with density $\gamma(y|x)$\\
$\gamma:\mathcal{X}\mapsto \Delta(\mathcal{Y})$& density function of $\mathcal{M}(\gamma)$\\
$y=\mathcal{M}(x;\gamma)$ & $\mathcal{M}(x:\gamma)$ outputs $y$ when $x$ is the input dataset \\
$\textbf{Pr}^{\theta, s}_{\gamma}[\cdot] = \textbf{Pr}^{ s}_{\gamma}[\cdot] $ & probability induced by $\theta$, $s$, and $\gamma$\\
\hline
\end{tabular}
\end{table}

\begin{table}[h]
\centering
\caption{Notations for Section III and Appendix \ref{app:section_III_details}}
\begin{tabular}{ll}
\hline
Symbol & Description \\
\hline
$\mathcal{M}(\mathcal{\vec{\gamma}})$ & composition of $\mathcal{M}_{1}(\gamma_{1}), \dots, \mathcal{M}_{k}(\gamma_{k})$, \\
&with $\vec{\gamma}=(\gamma_{i=1})^{k}_{i=1}$\\
$\vec{Y}, \vec{y}, \vec{\mathcal{Y}}$  & RV, sample, and space of the joint outputs \\
&of $\mathcal{M}(x;\vec{\gamma})$ for any dataset $x$\\
$\psi_{i}:\mathcal{S}\mapsto \Delta(\mathcal{Y}_{i})$& density function of output $y_{i}$ of $\mathcal{M}_{i}(\gamma_{i})$\\
&conditioning on the secret $s$\\
$\mathcal{N}_{i}(\psi_{i}):\mathcal{S}\mapsto\mathcal{Y}_{i}$ & effective mechanism of $\mathcal{M}_{i}(\gamma_{i})$, with density $\psi_{i}$\\
$\mathcal{N}(\vec{\psi}):\mathcal{S}\mapsto\mathcal{Y}$ & composition of $\mathcal{N}_{i}(\psi_{1}), \dots, \mathcal{N}_{k}(\psi_{k})$ \\
&as if they are independent\\
$y_i = \mathcal{N}_{i}(s;\psi_{i})$ & if $S=s$, $X=x$, and $y_i=\mathcal{M}_{i}(x;\gamma_{i})$, \\
&then $y_i = \mathcal{N}_{i}(s;\psi_{i})$\\
$B(\cdot|s)$ & CDF of the joint $\vec{Y}$ when the secret is $s$\\
$b(\cdot|s)$ & density function of $B(\cdot|s)$\\
$C(u_{1}, \dots, u_{k}|s)$ & copula of $B(\cdot|s)$, a multivariate CDF \\
&with uniform univariate marginals over $[0,1]$ \\
$c(\cdot|s)$ & copula density of $C(\cdot|s)$\\
$\textbf{L}_{\mathtt{p}, \mathtt{q}}(Y)$& privacy loss random variable (PLRV) with $Y\sim \mathtt{p}$,\\
& where $\mathtt{p}$ and $\mathtt{q}$ are two densities\\
$\mathbf{L}^{\mathcal{G}}_{c_{s_{0}}, c_{s_{1}}}(\vec{Y})$& PLRV induced by the copula densities when\\
&$(\theta, \mathcal{G})$ captures the relationship between $S$ and $X$\\
$\textbf{L}^{\textup{id}}_{\psi_{s_{0}}, \psi_{s_{1}} }(\vec{Y})$ & PLRV induced by $\mathcal{N}(\vec{\psi})$\\
$\mathcal{C}^{j} \subseteq \{\mathcal{M}_i\}_{i=1}^k$ & $j$-th collection of mechanisms that are \\
&mutually dependent\\
$\mathcal{C} = \{\mathcal{C}^{j}\}_{j=1}^{m}$ & captures the dependencies among mechanisms\\
$\textbf{L}^{\mathcal{C}}_{\hat{c}_{s_{0}}, \hat{c}_{s_{1}}}(\vec{Y})$ & PLRV induced by $\mathcal{C}$\\
 $\underline{\textbf{opt}}_{\delta_g}$ &  optimal epsilon value for a given $\delta_{g}$ under \\
 &composition $\mathcal{N}(\vec{\psi})$ when dependencies \\
&induced by $(\theta, \mathcal{G})$ are ignored\\
 $\underline{\textbf{dt}}_{\epsilon_g}$ & optimal delta value for a given $\epsilon_{g}$ under\\
 &composition $\mathcal{N}(\vec{\psi})$ when dependencies \\
 &induced by $(\theta, \mathcal{G})$ are ignored\\
$\textbf{opt}_{\delta_g}$ & true optimal epsilon value given $\delta_{g}$ under \\
&composition $\mathcal{M}(\vec{\gamma})$\\
$\textbf{dt}_{\epsilon_g}$ & true optimal delta value given $\epsilon_{g}$ under \\
&composition $\mathcal{M}(\vec{\gamma})$\\
$\overline{\textbf{opt}}_{\delta_g}$ & optimal epsilon value given $\delta_{g}$ under\\
&composition $\mathcal{M}(\vec{\gamma})$ when the PLRV induced by\\
&copula is considered as an independent mechanism\\
$\overline{\textbf{dt}}_{\epsilon_g}$ & optimal delta value given $\epsilon_{g}$ under\\
&composition $\mathcal{M}(\vec{\gamma})$ when the PLRV induced by\\
&copula is considered as an independent mechanism\\
\hline
\end{tabular}
\end{table}

\begin{table}[h]
\centering
\caption{Notations for Section IV}
\begin{tabular}{ll}
\hline
Symbol & Description \\
\hline
$\alpha:\mathcal{S}\mapsto \mathcal{Y}_{\alpha}$& leader's privacy strategy\\
$Y_{\alpha}$, $y_{\alpha}$, $\mathcal{Y}_{\alpha}$& RV, sample, and space of output of $\alpha$\\
$\mathcal{M}_{k+1}(\gamma_{\alpha})$& the effective mechanism determined by $\alpha$\\
$\vec{Y}_{\alpha}$, $\vec{y}_{\alpha}$, $\vec{\mathcal{Y}}_{\alpha}$& RV, sample, and space of the joint output of \\
&$\mathcal{M}_{k+1}(\gamma_{\alpha})$ and $\mathcal{M}(\vec{\gamma})$\\
$\pi:\vec{\mathcal{Y}}_{\alpha}\mapsto \Delta(\mathcal{S})$& follower's response strategy\\
$\mathcal{L}(\pi, s, \vec{y}_{\alpha})$& any strictly proper scoring rule function\\
\hline
\end{tabular}
\end{table}

\subsection{Pufferfish Privacy}\label{app:Puffish Privacy}

Pufferfish privacy (PP) \cite{kifer2014pufferfish} generalizes the concept of adjacent datasets in DP to \textit{adjacent distributions} of datasets \cite{desfontaines2019sok}, defined by discriminative pairs of secrets representing sensitive information to be protected. This framework enables PP to incorporate domain-specific knowledge and handle data correlations more effectively than DP.

A typical PP framework consists of:
\begin{enumerate}
    \item a set of \textit{secrets} $\mathcal{S}$, containing measurable subsets of $\mathcal{X}$;
    \item a \textit{set of secret pairs} $\mathcal{Q} \subseteq \mathcal{S} \times \mathcal{S}$;
    \item a class of data distributions $\mathbb{D} \subseteq \Delta(\mathcal{X})$, representing domain knowledge as prior beliefs. 
\end{enumerate}

PP ensures privacy by maintaining that the posterior distinguishability of any secret pairs in $\mathcal{Q}$, given the outputs of a mechanism, remains close to their prior distinguishability. A randomized mechanism $\mathcal{M}: \mathcal{X} \to \mathcal{Y}$ is \textit{$(\epsilon, \delta)$-private in the pufferfish framework} ($(\epsilon, \delta)$-PP) with $\epsilon \geq 0$ and $\delta \in [0,1]$ if, for all $P_X \in \mathbb{D} $, $(s, s') \in \mathcal{Q}$ with $P_X(s) > 0$, $P_X(s') > 0$, and measurable $\mathcal{W} \subseteq \mathcal{Y}$:
\begin{equation}\label{eq:def_pufferfish}
\sup_{\mathcal{W} \subset \mathcal{Y}} \left( \textbf{Pr}^{s}\left[\mathcal{M}(X) \in \mathcal{W}\right] - e^{\epsilon} \textbf{Pr}^{s'}\left[\mathcal{M}(X) \in \mathcal{W}\right] \right) \leq \delta,
\end{equation}
where $\textbf{Pr}^{s}[\cdot]$ denotes the probability conditioned on the secret being $s$.

In Pufferfish Privacy (PP), each probability distribution $P_X \in \mathbb{D}$ corresponds to an attacker characterized by their prior knowledge $P_X$, which the system aims to protect against. This belief $P_X$ represents the attacker’s understanding of how the dataset was generated \cite{kifer2014pufferfish}. The set $\mathbb{D}$ is also referred to as \textit{data evolution scenarios}. Importantly, the PP framework does not explicitly account for any intrinsic probabilistic relationships between the secret $s$ and the dataset $x$.

% \textcolor{red}{
% Furthermore, while the secret $s$ can naturally be derived from the dataset $x$, the PP framework mathematically defines the relationship such that the secret determines the distribution of the dataset from the attacker’s perspective. Therefore, the standard PP framework does not leverage the general causality between the secret and the dataset in its privacy definition.
% }

% In PP, each probability distribution $P_{X}\in \mathbb{D}$ corresponds to an attacker (characterized by their prior knowledge $P_{X}$), which the system aims to protect against. 
% This belief $P_{X}$ represents the attacker's belief in how the dataset was generated \cite{kifer2014pufferfish}.
% The set $\mathbb{D}$ is also referred to as \textit{data evolution scenarios}.
% Importantly, the PP framework does not explicitly account for any intrinsic probabilistic relationships between the secret and the dataset.

% Furthermore, the secret can naturally be derived from the dataset. However, the PP framework mathematically defines the relationship such that the secret determines the distribution of the dataset from the attacker’s perspective. Therefore, the standard PP framework does not leverage the general causality between the secret and the dataset in the privacy definition.

Differential privacy (DP) is a special case of PP when the secret coincides with the dataset ($\mathcal{S} = \mathcal{X}$), $\mathcal{Q} \subseteq \mathcal{X} \times \mathcal{X}$ contains all pairs of adjacent datasets, and $\mathbb{D}  = \Delta(\mathcal{X})$.

% \textcolor{red}{\textbf{Differences between DP, BDP, PP, and DCP}}

\subsection{Detailed Characterization of Section III}\label{app:section_III_details}

In this section, we present the detailed version of Section III.
Part (i) of Theorem \ref{thm:order_composition_pre} is shown by Proposition \ref{prop:existence_composition_DCP}, and parts (ii) and (iii) of Theorem \ref{thm:order_composition_pre} is captured by Theorem \ref{thm:order_composition}.

\textbf{Composition }
In the DCP framework, as more computations are performed on the dataset $x$, it is important to understand how the privacy of the secret $s = \mathcal{G}(x)$ degrades under composition. 
Consider $k$ independent mechanisms $\mathcal{M}_{1}(\gamma_{1}), \mathcal{M}_{2}(\gamma_{2}), \dots, \mathcal{M}_{k}(\gamma_{k})$, where each $\mathcal{M}_{i}(\cdot;\gamma_{i}): \mathcal{X} \to \mathcal{Y}_{i}$ operates with density function $\gamma_{i}$ and output space $\mathcal{Y}_{i}$. Let $\vec{\gamma} \equiv (\gamma_{i})_{i=1}^k$, $\vec{Y} \equiv (Y_{i})_{i=1}^k$ denote the random variables corresponding to the outputs of the mechanisms, and $\vec{y} \equiv (y_{i})_{i=1}^k \in \vec{\mathcal{Y}}$, where $\vec{\mathcal{Y}} \equiv \prod_{i=1}^k \mathcal{Y}_{i}$.

The composition $\mathcal{M}(\cdot;\vec{\gamma}): \mathcal{X} \to \vec{\mathcal{Y}}$ is defined as:
\[
\mathcal{M}(x; \vec{\gamma}) = \left(\mathcal{M}(x; \gamma_{1}), \dots, \mathcal{M}(x; \gamma_{k})\right).
\]
For any $\vec{\mathcal{W}} \subseteq \vec{\mathcal{Y}}$, we define:
\begin{equation}\label{eq:likelihood_composition}
    \textbf{Pr}^{s}_{\vec{\gamma}}\left[\vec{Y} \in \vec{\mathcal{W}}\right] \equiv \int\limits_{x \in \mathcal{X}} \int\limits_{\vec{y} \in \vec{\mathcal{W}}} \prod_{i=1}^{k} \gamma_{i}(y_{i} \mid x) P_{\theta}(x|s) \, d\vec{y} \, dx.
\end{equation}
The ($k$-fold) composition $\mathcal{M}(\vec{\gamma})$ satisfies $(\epsilon_{g}, \delta_{g})$-DCP for some $\epsilon_{g} \geq 0$ and $\delta_{g} \in [0, 1]$ if,
\[
\sup_{(s,s')\in\mathcal{Q}} \sup_{\vec{\mathcal{W}} \subset \vec{\mathcal{Y}}} 
\left( \textbf{Pr}^{s}_{\vec{\gamma}}\left[\vec{Y} \in \vec{\mathcal{W}}\right] - e^{\epsilon_{g}} \textbf{Pr}^{s'}_{\vec{\gamma}}\left[\vec{Y} \in \vec{\mathcal{W}}\right] \right) \leq \delta_{g}.
\]

\paragraph{Effective Mechanisms}
For each $\gamma_{i}: \mathcal{X} \to \Delta(\mathcal{Y}_{i})$, let $\psi_{i}: \mathcal{S} \to \Delta(\mathcal{Y}_{i})$ denote the underlying density function of $\textbf{Pr}^{s}_{\gamma}[\cdot]$ given by (\ref{eq:prior_likelihood}). Specifically,
\[
\psi_{i}(y_{i} \mid s) = \int_{x \in \mathcal{X}} \gamma_{i}(y_{i} \mid x) P_{\theta}(x|s)\, dx.
\]
Additionally, let 
\[
\mathcal{N}_{i}(\cdot;\psi_{i}): \mathcal{S} \to \mathcal{Y}_{i}
\]
denote the corresponding \textit{effective mechanism} of $\mathcal{M}_{i}(\gamma_{i})$ for all $i \in [k]$ such that, for a given realized pair $s,x$, 
\[
\textbf{Pr}^{s}_{\psi_{i}}\left[\mathcal{N}_{i}(s;\psi_{i})\in\mathcal{W}_{i} \right] = \textbf{Pr}^{s}_{\psi_{i}}\left[\mathcal{M}_{i}(x;\gamma_{i})\in\mathcal{W}_{i} \right]. 
\]
Since $\textbf{Pr}^{s}_{\psi_{i}}[\cdot] = \textbf{Pr}^{s}_{\gamma_{i}}[\cdot]$ for all $i$, the mechanism $\mathcal{N}_{i}(\psi_{i})$ is $(\epsilon_{i}, \delta_{i})\textup{-}\mathtt{ind}$ if and only if $\mathcal{M}_{i}(s; \gamma_{i})$ is $(\epsilon_{i}, \delta_{i})$-DCP.
Let $\Psi_{i}$ denote the cumulative distribution function (CDF) associated with the density $\psi_{i}$ of each effective mechanism $\mathcal{N}_{i}$.

\subsubsection{Copula}

For any $s \in \mathcal{S}$, let $B(\cdot | s)$ denote the cumulative distribution function (CDF) and $b(\cdot | s)$ the density function of $\textbf{Pr}^{s}_{\vec{\gamma}}$ as defined in (\ref{eq:likelihood_composition}). For ease of exposition, we focus on where the output $y_{i} \in \mathcal{Y}_{i}$ is univariate, for all $i\in[k]$.
It is straightforward to verify that the $i$-th marginal density of $B$ ($b$) equals $\Psi_{i}$ ($\psi_{i}$) for all $i \in [k]$. 

By Sklar's theorem \cite{sklar1959fonctions}, there exists a \textit{copula} $C$ for every $s \in \mathcal{S}$ such that:
\[
B\left(Y_{1}, \dots, Y_{k} \mid s\right) = C\left(\Psi_{1}\left(Y_{1} \mid s\right), \dots, \Psi_{k}\left(Y_{k} \mid s\right) \mid s\right),
\]
where $C: [0,1]^{k} \to [0,1]$ captures the dependency among mechanisms. The copula $C$ is a multivariate CDF with uniform univariate marginals:
\[
C(u_{1}, u_{2}, \dots, u_{k}) = \textbf{Pr}\left[U_{1} \leq u_{1}, U_{2} \leq u_{2}, \dots, U_{k} \leq u_{k}\right],
\]
where each $U_{i} \in [0,1]$ is a standard uniform random variable for all $i \in [k]$. 

Let $c$ denote the \textit{copula density function} of $C$. Then, the density function $b$ of $B$ can be expressed as:
\[
b_{s}\left(\vec{Y}\right) = c\left(\Psi_{1}\left(Y_{1} \mid s\right), \dots, \Psi_{k}\left(Y_{k} \mid s\right) \mid s\right) \prod_{i=1}^{k} \psi_{i}\left(Y_{i} \mid s\right),
\]
For simplicity, we let 
\begin{itemize}
    \item $c_{s}(\cdot) = c(\Psi_{1}(\cdot | s), \dots, \Psi_{k}(\cdot | s) | s)$, and
    \item $\psi_{s}(\cdot) = (\psi^{s}_{1}(\cdot), \dots, \psi^{s}_{k}(\cdot))$, with $\psi^{s}_{i}(\cdot) = \psi_{i}(\cdot | s)$.
\end{itemize}

\subsubsection{Privacy Loss Random Variable}

Quantifying privacy risk using the $(\epsilon, \delta)$ scheme can be characterized through the \textit{privacy loss random variable} (PLRV). Let $p$ and $q$ be two probability density functions on the output space $\mathcal{Y}$ corresponding to two randomized mechanisms. Define $\textbf{L}_{\textup{p} \| \textup{q}}: \mathcal{Y} \to \mathbb{R}$ by:
\[
\textbf{L}_{\textup{p}, \textup{q}}(y) \equiv \log\left(\frac{\textup{p}(y)}{\textup{q}(y)}\right).
\]
The PLRV is then given by $\textbf{L}_{\textup{p}, \textup{q}}(Y)$ for $Y \sim p$, which is not symmetric in $\textup{p}$ and $\textup{q}$. The corresponding \textit{privacy loss distribution} (PLD) is denoted by $\textup{PLD}(p \| q)$.

Now suppose the mechanisms $\{\mathcal{M}_{i}\}_{i=1}^{k}$ are \textit{independent}. Given $\{\theta, \mathcal{G}\}$, the PLRV of the composition $\mathcal{M}(\vec{\gamma})$ is given by:
\begin{equation}\label{eq:privacy_loss_RV}
\textbf{L}_{b_{s_{0}}, b_{s_{1}}}(\vec{Y}) = \textbf{L}^{\mathcal{G}}_{c_{s_{0}}, c_{s_{1}}}(\vec{Y}) + \textbf{L}^{\textup{id}}_{\psi_{s_{0}}, \psi_{s_{1}} }(\vec{Y}),
\end{equation}
where $\textbf{L}^{\mathcal{G}}_{c_{s_{0}}, c_{s_{1}}}(\vec{Y})$ is determined by the copula densities, and 
\[
\textbf{L}^{\textup{id}}_{\psi_{s_{0}}, \psi_{s_{1}} }(\vec{Y})=\sum^{k}\nolimits_{i=1}\textbf{L}_{\psi_{i}^{s_{0}},\psi_{i}^{s_{1}} }(Y_{i}),
\]
with each $\textbf{L}_{\psi_{i}^{s_{0}}, \psi_{i}^{s_{1}}}(Y_{i})$ determined by the individual marginal density $\psi_{i}$.

\begin{definition}[Invertable $\mathcal{G}$]
    The mapping $\mathcal{G}: \mathcal{X} \to \mathcal{S}$ is \textit{invertible} if for any secret sample $s$, there exists a dataset $x$ such that $P_{\theta}(x | s) = 1$. 
    $\mathcal{G}$ is \textit{non-invertible} if it is not invertible.
\end{definition}

\begin{proposition}\label{prop:invertible_G}
Suppose that the mechanisms $\{\mathcal{M}_{i}\}_{i=1}^{k}$ are independent. Then, the following holds.
\begin{itemize}
    \item If $\mathcal{G}$ is invertible, then $\textbf{L}^{\mathcal{G}}_{c_{s_{0}}, c_{s_{1}}}(\vec{Y}) = 0$ in (\ref{eq:privacy_loss_RV}) for all $s_{0}, s_{1}$.

    \item If $\mathcal{G}$ is non-invertible, then $\textbf{L}^{\mathcal{G}}_{c_{s_{0}}, c_{s_{1}}}(\vec{Y}) \neq 0$ in (\ref{eq:privacy_loss_RV}) for all $s_{0}, s_{1}$.
\end{itemize}

\end{proposition}

Proposition~\ref{prop:invertible_G} implies that a non-invertible mapping $\mathcal{G}$ generally introduces additional dependency among mechanisms, even when the mechanisms themselves are independent.

To account for \textit{arbitrary dependencies} among mechanisms, we group the collection of $k$ mechanisms $\{\mathcal{M}_i\}_{i=1}^k$ into sub-collections, denoted as $\mathcal{C}^{j} \subseteq \{\mathcal{M}_i\}_{i=1}^k$ for $j \in \{1, 2, \dots, m\}$ with $m \geq 1$. Each sub-collection consists of mechanisms that are mutually dependent, and no mechanism outside a sub-collection is mutually dependent with all mechanisms within that sub-collection. Let $\mathcal{C} = \{\mathcal{C}^{j}\}_{j=1}^{m}$.

We assume that each individual mechanism $\mathcal{M}_{i}(\gamma_{i})$ has a well-defined marginal distribution $\gamma_{i}$, meaning that the marginal probability distribution of the output of $\mathcal{M}_{i}$ is $\gamma_{i}$, even if $\mathcal{M}_{i}$ exhibits dependencies with other mechanisms.
Appendix \ref{app:well_defined_dependence} explains the notion of \textit{well-defined dependency structure.}

Sklar's theorem implies that each sub-collection $\mathcal{C}^{j}$ has a well-defined copula, denoted by $C^{j}(\cdot | s)$, with copula density $\hat{c}^{j}(\cdot | s)$, which may depend on the secret $s$. For simplicity, let $\hat{c}^{j}_{s} = \hat{c}^{j}(\cdot | s)$.
The PLRV induced by the copula density $\hat{c}^{j}$ is denoted as $\textbf{L}^{j}_{\hat{c}^{j}_{s_{0}}, \hat{c}^{j}_{s_{1}}}(\tilde{Y}^{j})$, where $\tilde{Y}^{j}$ represents the aggregated output of the mechanisms in $\mathcal{C}^{j}$.

\begin{proposition}\label{prop:existence_composition_DCP}
The composition $\mathcal{M}(\vec{\gamma})$ of $\{\mathcal{M}_{i}\}_{i=1}^{k}$ with $\mathcal{C}$ is $(\epsilon_{g}, \delta_{g})$-DCP for some $\epsilon_{g} \geq 0$ and $\delta_{g} \in [0,1]$.
\end{proposition}

Proposition~\ref{prop:existence_composition_DCP} confirms that the composition of DCP mechanisms satisfies DCP but does not specify the privacy parameters or a method for their computation.

Similar to (\ref{eq:privacy_loss_RV}), the PLRV of the composition $\mathcal{M}(\vec{\gamma})$ with dependencies $\mathcal{C}$ can be expressed as:
\begin{equation}\label{eq:PLRV_dependent}
    \textbf{L}_{\hat{b}_{s_{0}}, \hat{b}_{s_{1}}}(\vec{Y}) = \textbf{L}^{\mathcal{G}}_{c_{s_{0}}, c_{s_{1}}}(\vec{Y}) + \textbf{L}^{\mathcal{C}}_{\hat{c}_{s_{0}}, \hat{c}_{s_{1}}}(\tilde{Y}) + \textbf{L}^{\textup{id}}_{\psi_{s_{0}}, \psi_{s_{1}}}(\vec{Y}),
\end{equation}
where $\textbf{L}^{\mathcal{C}}_{\hat{c}_{s_{0}}, \hat{c}_{s_{1}}}(\vec{Y})= \sum^{m}\nolimits_{j=1} \textbf{L}^{j}_{\hat{c}^{j}_{s_{0}},\hat{c}^{j}_{s_{1}}}(\tilde{Y}^{j})$.

Thus, the PLRV of the composition is the additive combination of contributions from the marginals, the copulas of dependent mechanisms, and the copula induced by the relationship between $S$ and $X$. This decomposition enables separate analysis of the effects of the relationship between $S$ and $X$ and the dependencies among mechanisms on privacy loss.
When all PLRVs are independent, the composition corresponds to the convolution of the associated privacy loss distributions (PLDs) \cite{sommer2019privacy}. 

However, $\textbf{L}_{\hat{b}_{s_{0}}, \hat{b}_{s_{1}}}(\vec{Y})$ in (\ref{eq:PLRV_dependent}) generally does not correspond to the convolution of the PLDs of each PLRV, even when the mechanisms are independent (i.e., $\textbf{L}^{\mathcal{C}}_{\hat{c}_{s_{0}}, \hat{c}_{s_{1}}} = 0$) due to a non-zero $\textbf{L}^{\mathcal{G}}_{c_{s_{0}}, c_{s_{1}}}(\vec{Y})$.

\subsubsection{Composition of PP}

When the relationship between $S$ and $X$ follows a chain rule, DCP reduces to Pufferfish privacy for a specific metric $(\mathtt{D}, \mathtt{d})$. An invertible mapping $\mathcal{G}$ (where $L^{\mathcal{G}}_{c_{s_{0}} \| c_{s_{1}}}(\vec{Y}) = 0$) is both necessary and sufficient for the \textit{universally composable evolution scenario} of Pufferfish privacy \cite{kifer2014pufferfish}, ensuring the linear self-composition property \cite{kifer2014pufferfish}, analogous to DP's basic composition.

\subsubsection{Composition of DP}

If the mapping $\mathcal{G}$ is invertible and, for any pair $s_0 \neq s_1$, there exist $x \neq x'$ such that 
\[
P_{\theta}(s_0, x) = 1 \textup{ and } P_{\theta}(s_1, x') = 1,
\]
then DCP reduces to DP.

In this case, the composition $\mathcal{M}(\vec{\gamma})$ is equivalent to the composition of the effective mechanisms $\mathcal{N}(\cdot;\vec{\psi}): \mathcal{S} \to \vec{\mathcal{Y}}$, defined as 
\[
\mathcal{N}(s;\vec{\psi}) = \left(\mathcal{N}_{1}(s; \psi_{1}), \dots, \mathcal{N}_{k}(s; \psi_{k})\right).
\]
DP composes gracefully under the assumption that each $\mathcal{N}_{i}(\psi_{i})$ is $(\epsilon_{i}, \delta_{i})\textup{-}\mathtt{ind}$.
\begin{itemize}
    \item The $k$-fold composition is $(\epsilon_{\mathtt{dp}}, \delta_{\mathtt{dp}})\textup{-}\mathtt{ind}$.

    \item The $k$-fold composition satisfies the \textbf{basic composition theory}: $\epsilon_{\mathtt{dp}} = \sum_{i=1}^{k} \epsilon_{i}, \quad \delta_{\mathtt{dp}} = \sum_{i=1}^{k} \delta_{i}.$
    \item For any $\delta_{\mathtt{dp}} \in [0,1]$, the optimal $\epsilon_{\mathtt{dp}}= \textup{OPT}(\mathcal{N}(\vec{\psi}), \delta_{\mathtt{dp}})$ where
\begin{equation*}%\label{eq:def_opt_comp}
    \textup{OPT}\left(\mathcal{N}(\vec{\psi}), \delta_{\mathtt{dp}}\right) \equiv \inf\left\{\epsilon' \geq 0 \mid \mathcal{N}(\vec{\psi}) \textup{ is $(\epsilon', \delta_{\mathtt{dp}})\textup{-}\mathtt{ind}$}\right\}.
\end{equation*}
The \textbf{optimal composition theory} of DP (Theorem 1.5 of \cite{murtagh2015complexity}) provides a method to compute the tightest $\epsilon_{\mathtt{dp}}$ from $\{\epsilon_{i}, \delta_{i}\}_{i=1}^{k}$ for a given $\delta_{\mathtt{dp}} \in [0,1]$. However, this computation is \#P-complete \cite{murtagh2015complexity}.

\end{itemize}

% The \textit{basic composition} rule states: $\epsilon_{\mathtt{dp}} = \sum_{i=1}^{k} \epsilon_{i}, \quad \delta_{\mathtt{dp}} = \sum_{i=1}^{k} \delta_{i}.$

% For any $\delta_{\mathtt{dp}} \in [0,1]$, the optimal $\epsilon_{\mathtt{dp}}= \textup{OPT}(\mathcal{N}(\vec{\psi}), \delta_{\mathtt{dp}})$ where
% \begin{equation}\label{eq:def_opt_comp}
%     \textup{OPT}\left(\mathcal{N}(\vec{\psi}), \delta_{\mathtt{dp}}\right) \equiv \inf\left\{\epsilon' \geq 0 \mid \mathcal{N}(\vec{\psi}) \textup{ is $(\epsilon', \delta_{\mathtt{dp}})\textup{-}\mathtt{ind}$}\right\}.
% \end{equation}
% The \textit{optimal composition theory} of DP (Theorem 1.5 of \cite{murtagh2015complexity}) provides a method to compute the tightest $\epsilon_{\mathtt{dp}}$ from $\{\epsilon_{i}, \delta_{i}\}_{i=1}^{k}$ for a given $\delta_{\mathtt{dp}} \in [0,1]$. However, this computation is \#P-complete \cite{murtagh2015complexity}.

We have a natural question: \textit{does the composition of DCP exhibit similarly graceful behavior when all mechanisms are independent?} The answer, however, is negative.

For a given composition $\mathcal{M}(\vec{\gamma})$, define the following notations for \textit{any $(s_{0}, s_{1})\in \mathcal{Q}$}, where we omit explicit dependence on $(s_{0}, s_{1})$ for simplicity.

\paragraph{}
Let $\underline{\textbf{opt}}_{\delta_g} \equiv \textup{OPT}(\mathcal{N}(\vec{\psi}), \delta_g)$ denote the optimal $\epsilon_g$ when ignoring dependencies among mechanisms for a given $\delta_g$, which can be computed by DP's optimal composition theory \cite{murtagh2015complexity}.

\paragraph{} Let $\textbf{opt}_{\delta_g} \equiv \inf\{\epsilon_g \mid \mathcal{M}(\vec{\gamma}) \textup{ is } (\epsilon_g, \delta_g)\textup{-DCP}\}$ represent the true optimal $\epsilon_g$ of $\mathcal{M}(\vec{\gamma})$ for a given $\delta_g$.

\paragraph{} Define $\widehat{\mathcal{N}}(c): \mathcal{S} \to \vec{\mathcal{Y}}$ as a additional mechanism that induces a PLRV coinciding with $\textbf{L}^{\mathcal{G}}_{c_{s_0}, c_{s_1}}(\vec{Y})$ in (\ref{eq:privacy_loss_RV}), where $c$ is the density of the CDF $C(u_1,\dots, u_k)$ independent of $\{\Psi_{i}\}^{k}_{i=1}$.

\paragraph{} Let $\overline{\textbf{opt}}_{\delta_g} \equiv \inf\{\epsilon_g \mid \mathcal{M}(\vec{\gamma}, c) \textup{ is } (\epsilon_g, \delta_g)\textup{-DCP}\}$, where $\mathcal{M}(\vec{\gamma}, c): \mathcal{S} \to \vec{\mathcal{Y}} \times \vec{\mathcal{Y}}$ is the composition of $k+1$ independent mechanisms $\mathcal{M}_1(\gamma_1), \dots, \mathcal{M}_k(\gamma_k), \widehat{\mathcal{N}}(c)$.

\paragraph{} The privacy profile (delta value) is given by 
\[
\delta_\epsilon \equiv \mathbb{E}_{\textbf{L} \sim \textup{PLD}} \left[\left(1 - e^{\epsilon - \textbf{L}}\right)^+\right],
\]
for a PLD \cite{steinke2022composition}. Let $\underline{\textbf{dt}}_{\epsilon_g}$, $\textbf{dt}_{\epsilon_g}$, and $\overline{\textbf{dt}}_{\epsilon_g}$ denote the privacy profiles corresponding to the PLDs of the same composition settings as $\underline{\textbf{opt}}_{\delta_g}$, $\textbf{opt}_{\delta_g}$, and $\overline{\textbf{opt}}_{\delta_g}$, respectively.

\begin{theorem}\label{thm:order_composition}
Suppose that $k$ mechanisms $\{\mathcal{M}_{i}(\gamma_{i})\}^{k}_{i=1}$ are \textit{independent}.
For any $\epsilon_{g}\geq 0$ and $\delta_{g}\in[0,1)$, the following holds for all $(s_{0}, s_{1})\in\mathcal{Q}$.
Suppose $\textbf{L}^{\mathcal{G}}_{c_{s_{0}}, c_{s_{1}}}(\vec{Y})\neq 0$. Then, 
\begin{itemize}
    \item[(i)] $\mathcal{M}(\vec{\gamma})$ is generally \textit{not} $(\sum^{k}_{i=1}\epsilon_{i}, \sum^{k}_{i}\delta_{i})$-DCP;
    \item[(ii)] $\underline{\textbf{opt}}_{\delta_{g}}< \textbf{opt}_{\delta_{g}}\leq \overline{\textbf{opt}}_{\delta_{g}} \text{ and } \underline{\textbf{dt}}_{\epsilon_{g}}< \textbf{dt}_{\epsilon_{g}}\leq \overline{\textbf{dt}}_{\epsilon_{g}}.$
\end{itemize}
\end{theorem}

Theorem~\ref{thm:order_composition} shows that the composition of DCP mechanisms is less graceful than that of DP mechanisms, even when the DCP mechanisms are independent. Specifically, if the PLRV induced by the relationship between $S$ and $X$ satisfies $\textbf{L}^{\mathcal{G}}_{c_{s_{0}}, c_{s_{1}}}(\vec{Y}) \neq 0$, the basic composition generally fails. 

Furthermore, the DP optimal composition theory \cite{murtagh2015complexity} underestimates the true aggregated privacy loss for DCP composition in this case.
When $\textbf{L}^{\mathcal{G}}_{c_{s_{0}}, c_{s_{1}}}(\vec{Y}) = 0$, the DCP composition is reduced to DP composition.

\subsection{Examples of Non-Chain $S\leftrightharpoons X$ }\label{app:examples}

In this section, we discuss real-world scenarios where the secret $S$ and the input dataset $X$ do not follow the simple chain-rule relationship $S \rightarrow X$. Our DCP framework addresses these scenarios by protecting the privacy of the secret $S$ even when the dataset $X$ is processed in ways not directly tied to $S$. These settings feature either a feedback loop $S \leftrightarrow X$ or a causal direction $X \rightarrow S$ (i.e., $S \leftarrow X$), rather than $S$ being simply “contained” within $X$.

In these cases, the sensitive secret $S$ is not merely an intrinsic component of the dataset $X$. Instead, $S$ may be:
(i) \textit{Derived from $X$:} An aggregated statistic, diagnosis, or hidden preference estimated by a model.
(ii) \textit{Jointly influencing $X$:} Through feedback loops, creating a bidirectional relationship $S \leftrightarrow X$. These distinctions highlight the need for more flexible privacy frameworks.

In real-world data pipelines, the process from $X$ to $S$ (e.g., via a machine-learning model or decision-making rule) is rarely a purely deterministic function. It often incorporates various sources of randomness or uncertainty, such as measurement noise, algorithmic randomization, or natural variability in the underlying phenomena being measured. This randomness further complicates the relationship between $X$ and $S$.

Classical adjacency-based definitions of differential privacy assume row-level or record-level changes in $X$ to reason about privacy. However, when $S$ is derived from, updated by, or mutually dependent on $X$, these row-level assumptions fail to capture the complexity of the relationships and the potential for information leakage.

% In this section, we discuss real-world scenarios where the secret $S$ and the input dataset $X$ do not follow the simple chain-rule causality $S \rightarrow X$. 
% Our DCP framework aims to protect the privacy of the secret when the dataset $X$ is used by mechanisms that process data in ways not directly tied to the secret $S$. These scenarios feature either a feedback loop $S \leftrightarrow X$ or a causal direction where $X \rightarrow S$ (i.e., $S \leftarrow X$), rather than $S$ being simply “contained” within $X$.

% In such cases, the sensitive secret $S$ does not merely reside \textit{inside} the dataset $X$. Instead, $S$ may be:
% (i) \textit{Derived from $X$:} An aggregated statistic, diagnosis, or hidden preference estimated by a model.
% (ii) \textit{Jointly influencing $X$:} Through feedback loops, creating a bidirectional relationship $S \leftrightarrow X$.

% In most real-world data pipelines, the process from $X$ to $S$ (e.g., machine-learning model or decision-making rule) is rarely a purely deterministic function. It often incorporates various sources of randomness or uncertainty, such as measurement noise and errors, or algorithmic randomization, or natural variability in the underlying phenomena being measured.

% Classical adjacency-based definitions of differential privacy assume row-level or record-level changes in $X$ to reason about privacy. However, when $S$ is derived from, updated by, or mutually dependent on $X$, these row-level assumptions are no longer sufficient.

\paragraph{Multi-Agent Sensing and State Estimation}

In large-scale sensor networks or multi-agent robotic systems, each agent gathers partial observations, denoted as $X$. The dataset $X$ consists of raw sensor readings, such as environmental measurements. The secret $S$, on the other hand, represents the global state of the environment, which could include factors like the location of a target, the temperature distribution, or the status of a system.

The process involves a causal relationship between the dataset and the secret. Agents first collect sensor readings ($X$), which are then fed into an estimation process that derives the secret state $S$, such as an aggregated or fused estimate. Based on this current estimate, agents may act, for example, moving sensors to new locations, which in turn alters the environment and influences subsequent sensor readings.

This system does not follow a simple chain of causality because the state $S$ is derived from the collective sensor data ($S \leftarrow X$) and simultaneously influences the future dataset due to feedback from agent actions ($S \leftrightarrow X$). Unlike cases where the secret is a direct subset of the dataset, here the secret $S$ emerges from agent readings and has a dynamic, bidirectional relationship with future observations.

\paragraph{Medical Diagnosis and Screening Pipelines}

In a healthcare setting, patients undergo multiple stages of tests, including blood tests, imaging, and genetic screening, which produce complex data. A final diagnosis or sensitive health condition is determined by combining the results of these tests.

The dataset $X$ consists of the patient’s preliminary test results and demographic data. The secret $S$, on the other hand, represents a diagnosis or sensitive health label, such as “high risk of condition $S$,” which is \textit{inferred} through a multi-stage screening protocol.

The process involves a causal flow where the raw dataset $X$, comprising lab values, images, and other data, is analyzed by a classifier or medical expert system. This system then derives the secret $S$, representing the diagnosis or health label. Once a tentative diagnosis is reached, follow-up tests or changes in the patient’s condition may provide new data, feeding it back into the diagnostic pipeline.

This system does not follow a simple chain of causality because the diagnosis $S$ is not a straightforward field within $X$. Instead, $S$ is a function of multiple correlated variables within $X$, potentially with feedback loops where new tests are ordered or the dataset evolves based on an initial diagnosis. As such, $S$ and $X$ are \textit{jointly dependent} rather than having a purely subset-based relationship ($S \not\subseteq X$).

\paragraph{Online Learning or Personalized Recommendation Systems}

An online platform, such as a streaming service or e-commerce site, collects click and consumption data over time to adjust its recommendations based on perceived user preferences. 
The dataset $X$ consists of observed user interactions, including clicks, dwell times, and watch history. The secret $S$ represents the user’s true preference profile, such as underlying tastes or personal attributes, or latent variables within the recommendation engine’s model.

The causal relationship in this system operates as follows. The platform updates the latent user profile ($S$) based on the newly observed dataset ($X$). The user’s future behavior is influenced by the platform’s recommendations, which are shaped by the current estimate of $S$. Over time, both $S$ and $X$ evolve together, creating a bidirectional feedback loop ($S \leftrightarrow X$).

This system does not follow a simple chain of causality because the preference $S$ is a hidden state or parameter, not a direct entry within $X$. Instead, $S$ is updated dynamically by user actions, and feedback loops emerge when preferences shift in response to the platform’s recommendations.

\paragraph{Cyber-Physical Systems With Feedback Control}

A smart grid or industrial control system uses sensor readings to adjust actuators in real-time. The system’s configuration or control policy might represent a proprietary or sensitive ``secret".

The dataset $X$ consists of real-time sensor data, including power flows, temperature readings, and network states. The secret $S$ captures the internal control policy or hidden control variables that the operator wants to keep confidential, such as a proprietary control algorithm or target operating set point.

The causal process operates as follows. Observations $X$ are fed into a controller that determines $S$, which represents the next system set point or policy choice. This set point then influences subsequent sensor readings, creating a feedback loop between the data and the control policy.

This system does not follow a simple chain of causality because standard differential privacy (DP) adjacency treats $S$ as part of the data. However, in reality, $S$ is an emergent or dynamically updated control policy that is not strictly present in the sensor readings. The relationship between $S$ and $x$ is bidirectional and may exhibit randomness due to system dynamics.

%%%%%%%%%%%%%%%%%%%%%%%%%
\subsection{DCP as Single-Point DP}\label{app:DCP_single_point_DP}

\section{CP As Single-Point DP}

Consider that the prior knowledge $\theta\in\Theta$ is a uniform distribution (or other non-informative prior) over $\mathcal{S}$.
In addition, let $\mathtt{D}(s,s'|\theta)$ represent Hamming distance and $\mathtt{d}=1$.
Thus, any adjacent pair $s$ and $s'$ with $\mathtt{D}(s,s'|\theta)\leq 1$ differ in only one entry (denoted by $s\sim s'$).
Thus, in this case, the (hypothetical) sensitive information is an entry of a secret, where the relationship between the sensitive information and the secret is deterministic.
Then, $(\epsilon, \delta)$-DCP is a notion of $(\epsilon, \delta)$-DP when we treat the dataset $x$ as some intermediate terms in the probabilistic pipeline from the secret to the output.

When we consider adjacent secrets with priors and $\mathcal{G}$ is deterministic and invertible, the DCP framework also mathematically coincides with the \textit{Bayesian differential privacy} (BDP).

\begin{definition}[$(\epsilon_{\omega}, \delta_{\omega})$-Bayesian Differential Privacy \cite{triastcyn2020bayesian}]
A randomized mechanism $\mathcal{M}:\mathcal{X}\mapsto \mathcal{Y}$ is \textup{$(\epsilon_{\omega}, \delta_{\omega})$-Bayesian differentially private ($(\epsilon_{\omega}, \epsilon_{\omega})$-BDP)} with $\epsilon_{\omega}\geq 0$ and $\delta_{\omega}\in[0,1]$, if for any $x\simeq x'$ differing in a single entry $z\sim \omega$, we have
\[
    \textup{Pr}\left[\mathcal{M}(x)\in \widehat{\mathcal{Y}}\right]\leq e^{\epsilon_{\omega}} \textup{Pr}\left[\mathcal{M}(x')\in \widehat{\mathcal{Y}}\right] + \delta_{\omega}, \forall \widehat{\mathcal{Y}}\subseteq \mathcal{Y},
    \]
    where the probability is taken over the randomness of the output response $y$ and the data entry $z$.
\end{definition}

The notion of BDP enjoys many properties of the standard DP \cite{triastcyn2020bayesian}, including post-processing, composition, and group privacy. 
The core idea of BDP lies in defining \textit{typical scenarios}.
A scenario is considered as \textit{typical} when all sensitive data is drawn from the same distribution \cite{triastcyn2020bayesian}.
This is motivated by the fact that many machine learning models are often designed and trained for specific data distributions, and such prior knowledge is usually known by the attackers.
The BDP framework adjusts the noise according to the data distribution and can provide a better expected privacy guarantee.

In the general DCP framework, we treat the entire secret as sensitive information and consider the probability space over all states, while the BDP only considers the space over a single differing entry.
Since the common data entries of the adjacent datasets are assumed to be known by the (worst-case) attacker, these common entries do not impact the distributions of the mechanism. 
In addition, both standard DP and BDP consider independent data.

That is, the privacy guarantees by DP and BDP do not consider the correlation between data entries.
Therefore, we can mathematically treat the DCP framework as the BDP framework when the \textit{virtual dataset} (that contains sensitive entry) is a single-point dataset and the true dataset $x$ is some intermediate term that is not publicly observable.
However, such a single-point (B)DP equivalence is valid only when there is one mechanism. 
When multiple DCP mechanisms are composed, this equivalence fails (See Section III for more detail).

%%%%%%%%%%%%%%%%%%%%%%%%
\subsection{Well-Defined Dependence Structure}\label{app:well_defined_dependence}

In this section, we describe what the \textit{well-defined dependence structure} is (with slight abuse of notations for simplicity).

\subsubsection*{\textbf{D.1} Probability Space and Random Mechanisms}

Let \((\mathcal{X}, \mathcal{F}, \mathbb{P})\) be a probability space, where:
\begin{itemize}
    \item \(\mathcal{X}\) is the set of possible input datasets,
    \item \(\mathcal{F}\) is a \(\sigma\)-algebra of measurable subsets of \(\mathcal{X}\),
    \item \(\mathbb{P}\) is the probability measure defined on \((\mathcal{X}, \mathcal{F})\).
\end{itemize}

Let \(M_{1}, M_{2}, \dots, M_{k}\) be mechanisms (random variables) that introduce randomness to ensure privacy. Each mechanism \(M_i\) maps a dataset \(x \in \mathcal{X}\) and a random input \(\omega \in \Omega'\) (from a probability space \((\Omega', \mathcal{F}', \mathbb{P}')\)) to an output in a measurable space \((\mathcal{Y}_{i}, \mathcal{B}_{i})\). Formally:
\[
M_i: \mathcal{X} \times \Omega' \to \mathcal{Y}_i,
\]
where \(\mathcal{B}_{i} \subseteq 2^{\mathcal{Y}_{i}}\) is a \(\sigma\)-algebra on \(\mathcal{Y}_{i}\).

For a fixed dataset \(x \in \mathcal{X}\), each \(M_i(x)\) is a random variable defined as:
\[
M_i(x) = M_i(x; \omega), \quad \omega \sim \mathbb{P}'.
\]

\begin{enumerate}
    \item For a fixed \(x \in \mathcal{X}\), denote by \(F_i^{(x)}\) the (marginal) distribution function (CDF) of \(M_i(x)\):
    \begin{itemize}
        \item If \(M_i(x)\) has a density \(f_i^{(x)}\), then:
        \[
        F_i^{(x)}(m) = \int_{-\infty}^m f_i^{(x)}(y) \, dy, \quad \forall m \in \mathcal{Y}_i.
        \]
    \end{itemize}
    \item Let \(F^{(x)}\) be the joint distribution function of the vector:
    \[
    \mathbf{M}(x) = \bigl(M_1(x), M_2(x), \dots, M_k(x)\bigr).
    \]
    The joint CDF \(F^{(x)}\) is defined over the product space 
    \((\mathcal{Y}_1 \times \mathcal{Y}_2 \times \cdots \times \mathcal{Y}_k, \mathcal{B}_1 \otimes \mathcal{B}_2 \otimes \cdots \otimes \mathcal{B}_k)\) as:
    \[
    \begin{aligned}
        &F^{(x)}(m_1, m_2, \dots, m_k) \\
        &= \mathbb{P}\big(M_1(x) \leq m_1, \dots, M_k(x) \leq m_k\big).
    \end{aligned}
    \]
\end{enumerate}

A fundamental question in dependence in probability theory and statistics is how and in what ways $F^{x}$ deviates from independence.
More generally, how the joint law $F^{x}$ \textit{``couples"} the marginals $\{F^{x}_{i}\}^{k}_{i=1}$.

\subsection*{\textbf{D.2} Dependence Structure via Copulas}

\textbf{Probability-Integral Transforms }
We consider that each marginal $F^{x}_{i}$ is continuous for all $x\in\mathcal{X}$.
Define the probability-integral transform for each mechanism \(M_i(s)\):
\[
U_{i}(s, \omega) \;:=\; F_{i}^{(s)}\bigl(M_{i}(s, \omega)\bigr).
\]
Then each \(U_{i}\) is a random variable taking values in \([0,1]\). Indeed, by the usual probability-integral-transform result, \(U_{i} \sim \text{Uniform}(0,1)\) when \(M_{i}(s, \omega)\) is continuous.

\textbf{The Copula as the Joint Law of \(\mathbf{U}(s)\)}
Consider the random vector
\[
\begin{aligned}
    \mathbf{U}(s) &=\bigl(U_{1}(s),U_{2}(s),\dots,U_{k}(s)\bigr)\\
    &= \bigl(F_{1}^{(s)}(M_{1}(s)),\; F_{2}^{(s)}(M_{2}(s)),\;\dots,\; F_{k}^{(s)}(M_{k}(s))\bigr).
\end{aligned}
\]
This vector \(\mathbf{U}(s)\) lives in the unit cube \([0,1]^{k}\). Define the copula \(C^{(s)}\) for the dataset \(s \in \mathcal{S}\) as: for $\mathbf{u} \in [0,1]^{k}$,
\[
C^{(s)}(\mathbf{u})
\;:=\; \mathbb{P}\!\Bigl(U_{1}(s) \le u_{1},\, U_{2}(s) \le u_{2},\,\dots,\,U_{k}(s) \le u_{k}\Bigr).
\]
Then \(C^{(s)}\) is itself a distribution function on \([0,1]^{k}\), called the \emph{copula} associated with the dataset \(s\).

\textbf{Sklar’s Theorem}
By Sklar’s theorem, the original joint distribution \(F^{(s)}\) of \(\mathbf{M}(s)\) and the marginals \(F_{i}^{(s)}\) satisfy:
\[
F^{(s)}(m_{1},\dots,m_{k})
\;=\;
C^{(s)}\bigl(F_{1}^{(s)}(m_{1}), \dots, F_{k}^{(s)}(m_{k})\bigr).
\]
Moreover, we have the following.
\begin{itemize}
    \item \textbf{Existence:} \(C^{(s)}\) always exists for any valid joint \(F^{(s)}\) and marginals \(F_{i}^{(s)}\).
    \item \textbf{Uniqueness:} If each \(F_{i}^{(s)}\) is continuous, \(C^{(s)}\) is unique. (If some marginals are not continuous, the copula need not be unique but it still exists.)
\end{itemize}

Thus, from a rigorous standpoint, the dependence structure among \(\mathbf{M}(s)\) is precisely the copula \(C^{(s)}\)—that is, the joint distribution of the transformed vector \(\mathbf{U}(s)\).

If some marginal \(F_{i}^{(s)}\) is not continuous, we still define
\[
U_{i}(s) = F_{i}^{(s)}(M_{i}(s)).
\]
Then each \(U_{i}(s)\) is not necessarily \(\text{Uniform}(0,1)\) in the usual sense (it may have jumps/atoms), but \(\mathbf{U}(s)\) still lives in \([0,1]^{k}\). A copula \(C^{(s)}\) satisfying
\[
F^{(s)}(m_{1},\dots,m_{k})
\;=\; C^{(s)}\bigl(F_{1}^{(s)}(m_{1}),\dots,F_{k}^{(s)}(m_{k})\bigr)
\]
still exists (by Sklar’s theorem), although it may not be unique when there are atoms.

Putting it all together, a concise definition of \textit{well-defined dependence structure} is as follows.

\begin{definition}
    Let \(\mathbf{M}(s)\) be a \(k\)-dimensional random vector with continuous marginals \(F_{i}^{(s)}\). Define
\[
U_{i}(s) = F_{i}^{(s)}(M_{i}(s)),
\]
and let \(\mathbf{U}(s) = \bigl(U_{1}(s), \dots, U_{k}(s)\bigr)\). The joint distribution of \(\mathbf{U}(s)\), which is a probability measure on \([0,1]^{k}\), is called the \emph{copula} \(C^{(s)}\). This measure \(C^{(s)}\) completely characterizes the dependence structure among the components of \(\mathbf{M}(s)\).
\end{definition}

Intuitively, the marginals \(F_{i}^{(s)}\) “remove” the specific scale and shape of each mechanism’s distribution (making each \(U_{i}(s)\) uniform on \([0,1]\)).
Then, what remains (i.e., how the \(U_{i}(s)\) jointly vary in the unit cube) is the pure dependence structure among the mechanisms, independent of the marginal properties.

To summarize, the independence of \(\mathbf{M}(s)\) corresponds to the case where the copula \(C^{(s)}\) is the product copula, i.e., \(C^{(s)}(\mathbf{u}) = u_{1} \cdot u_{2} \cdots u_{k}\). Any other copula \(C^{(s)}\) indicates non-trivial dependence.
If \(\mathbf{M}(s)\) has a well-defined joint distribution \(F^{(s)}\) with marginals \(F_{i}^{(s)}\), a corresponding copula \(C^{(s)}\) always exists and uniquely captures the dependence structure (assuming continuous marginals). Specifically, the dependence among \(\mathbf{M}(s)\) is represented by the equivalence class of all joint distributions on \((\mathcal{Y}_{1} \times \cdots \times \mathcal{Y}_{k})\) that share the same marginals \(F_{i}^{(s)}\). Concretely, this dependence structure is identified with the copula \(C^{(s)}\), which is the joint distribution of the transformed vector \(\bigl(F_{1}^{(s)}(M_{1}(s)), \dots, F_{k}^{(s)}(M_{k}(s))\bigr)\).

%%%%%%%%%%%%%%%%%%%%%%%%
\subsection{Gaussian Copula Perturbation}\label{app:gaussian_copula_per}

For any $\theta\in\Theta$, we let $\mathcal{Q}\equiv\{(s,s')|\mathtt{D}(s,s'|\theta)\leq\mathtt{d}\}$ for simplicity.
We start by defining \textit{bivariate Gaussian copula}.

\begin{definition}[Bivariate Gaussian Copula]
A bivariate copula \( C^{\mathtt{g}}_{\rho}:[0,1]^{2} \to [0,1] \) is called a \textup{Gaussian copula} if it is defined as  
\[
C^{\mathtt{g}}_{\rho}(u_{1}, u_{2}) = \Phi_{\rho}\big(\Phi^{-1}(u_{1}), \Phi^{-1}(u_{2})\big),
\]
where:
\begin{itemize}
    \item \(\Phi\) is the CDF of the standard normal distribution, mapping \(u \in [0,1]\) to \(\mathbb{R}\) via its inverse \(\Phi^{-1}\),
    \item \(\Phi_{\rho}\) is the joint CDF of a bivariate normal distribution with zero means, unit variances, and correlation coefficient \(\rho\).
\end{itemize}
This construction ensures \(C^{\mathtt{g}}_{\rho}\) satisfies the properties of a copula, mapping \([0,1]^2 \to [0,1]\) while capturing dependence structure via \(\rho\).
\end{definition}

\paragraph{Pseudo-random Sample Generation}
Let \(\xi_{1}\) and \(\xi_{2}\) be the density functions of two independent noise distributions \(V_{1} \in \mathcal{V}_{1}\) and \(V_{2} \in \mathcal{V}_{2}\), with corresponding CDFs \(\Xi_{1}\) and \(\Xi_{2}\). Additionally, let \(Z_{1} \sim N(\mu_{1}, \mathtt{var}_{1})\) and \(Z_{2} \sim N(0, 1)\), with \(\textup{F}_{1}(\cdot)\) as the CDF of \(N(\mu_{1}, \mathtt{var}_{1})\). For any \(\rho \neq 0\), the pseudo-random sample generation process is as follows:
\begin{enumerate}
    \item \textbf{Sample Generation}: Generate \((z_{1}, z_{2})\) where \(z_{1} \sim N(\mu_{1}, \mathtt{var}_{1})\) and \(z_{2} \sim N(0, 1)\).
    
    \item \textbf{Transform to Uniform Samples}: Compute \(u_{1} = \textup{F}_{1}(z_{1})\) and 
    \[
    u_{2} = \textup{F}_{\rho}\big(\rho z_{1} + \sqrt{1 - \rho^{2}} z_{2}\big),
    \]
    where \(\textup{F}_{\rho}\) is the CDF of \(N\big(\rho \mu_{1}, \rho^{2} \mathtt{var}_{1} + (1 - \rho^{2})\big)\).
    
    \item \textbf{Transform to Noise Samples}: Obtain \(v_{1} = \Xi^{-1}_{1}(u_{1})\) and \(v_{2} = \Xi^{-1}_{2}(u_{2})\), where \(\Xi^{-1}_{1}\) and \(\Xi^{-1}_{2}\) are the inverses of the CDFs corresponding to \(\xi_{1}\) and \(\xi_{2}\), respectively.
\end{enumerate}

Let \(\textup{PsedR}_{\rho}(\cdot|\Xi_{1}, \Xi_{2}):\mathbb{R} \times \mathbb{R} \to \mathcal{V}_{1} \times \mathcal{V}_{2}\) denote the deterministic mapping from \((z_{1}, z_{2})\) to \((v_{1}, v_{2})\) via the steps outlined above, such that 
\[
(v_{1}, v_{2}) = \textup{PsedR}_{\rho}(z_{1}, z_{2}|\Xi_{1}, \Xi_{2}).
\]
Additionally, let \(\overline{\textup{PR}}_{\rho}(Z_{1}, Z_{2}|\Xi_{1}, \Xi_{2})\) represent the joint distribution of \(V_{1}\) and \(V_{2}\) induced by \(N(\mu_{1}, \mathtt{var}_{1})\) and \(N(0, 1)\) through this process.

\paragraph{Copula Perturbation}
Following the pseudo-random sample generation, the distribution \(\overline{\textup{PR}}_{\rho}(Z_{1}, Z_{2}|\Xi_{1}, \Xi_{2})\) has a CDF given by the copula:
\[
C^{\mathtt{g}}_{\rho}(V_{1}, V_{2}|\Xi_{1}, \Xi_{2}) = \Phi_{\rho} \left(F_{1}^{-1}(\Xi_{1}(V_{1})), F_{\rho}^{-1}(\Xi_{2}(V_{2}))\right),
\]
where \(\Phi_{\rho}\) is the CDF of a bivariate normal distribution with correlation \(\rho\).

We say that a pair of mechanisms \(\mathcal{M}_{1}(\gamma_{1})\) and \(\mathcal{M}_{2}(\gamma_{2})\) are \textit{perturbed by the copula} \(C^{\mathtt{g}}_{\rho}(\Xi_{1}, \Xi_{2})\) if the noises \(v_{1}\) and \(v_{2}\) used to perturb \(\mathcal{M}_{1}(\gamma_{1})\) and \(\mathcal{M}_{2}(\gamma_{2})\), respectively, are sampled via
\[
(v_{1}, v_{2}) = \textup{PsedR}_{\rho}(z_{1}, z_{2}|\Xi_{1}, \Xi_{2}).
\]

Consider an arbitrary function \(\tilde{\eta}_{S}:\mathcal{S} \to \mathcal{B}\) that produces a univariate output \(\tilde{\eta}_{S}(s)\) for a given state \(s\). Let the sensitivity of \(\tilde{\eta}_{S}\) be defined as \(\mathtt{C}_{\textup{sen}} = \sup_{(s_{1}, s_{0}) \in \mathcal{Q}} |\tilde{\eta}_{S}(s_{0}) - \tilde{\eta}_{S}(s_{1})|\).

Suppose that for any \(s \in \mathcal{S}\),
\[
z_{1}(s) = \tilde{\eta}_{S}(s) + \tilde{z}_{1},
\]
where \(\tilde{z}_{1}\) is drawn from \(N\left(0, \mathtt{w}^2 \left(\frac{\mathtt{C}_{\textup{sen}}}{\epsilon_{\mathtt{c}}}\right)^2\right)\) with \(\mathtt{w} \geq 2\log(2/\delta_{\mathtt{c}})\), \(\epsilon_{\mathtt{c}} \geq 0\), and \(\delta_{\mathtt{c}} \in [0, 1)\), and \(\tilde{z}_{2}\) is drawn from \(N(0, 1)\).

Let \(\Xi_{1}\) and \(\Xi_{2}\) be the CDFs of two independent noise distributions inducing \(\gamma_{1}\) and \(\gamma_{2}\), respectively, such that the mechanisms \(\mathcal{M}_{1}(\gamma_{1})\) and \(\mathcal{M}_{2}(\gamma_{2})\) are \((\epsilon_{1}, \delta_{1})\)-CP and \((\epsilon_{2}, \delta_{2})\)-CP, respectively. Finally, let \(C^{\mathtt{g}}_{\rho}(\Xi_{1}, \Xi_{2}; \epsilon_{\mathtt{c}}, \delta_{\mathtt{c}})\) denote the induced copula. 

We then obtain the following theorem.

\begin{theorem}\label{thm:copula_1}
Suppose \(\mathcal{M}_{1}(\gamma_{1})\) and \(\mathcal{M}_{2}(\gamma_{2})\) are perturbed by the copula \(C^{\mathtt{g}}_{\rho}(\Xi_{1}, \Xi_{2}; \epsilon_{\mathtt{c}}, \delta_{\mathtt{c}})\) for some \(\epsilon_{\mathtt{c}} \geq 0\) and \(\delta_{\mathtt{c}} \in [0, 1)\). Then, the following holds:
\begin{itemize}
    \item[(i)] \(\mathcal{M}_{1}(\gamma_{1})\) and \(\mathcal{M}_{2}(\gamma_{2})\) remain \((\epsilon_{1}, \delta_{1})\)-CP and \((\epsilon_{2}, \delta_{2})\)-CP, respectively.
    \item[(ii)] For the composition of \(\mathcal{M}_{1}(\gamma_{1})\) and \(\mathcal{M}_{2}(\gamma_{2})\),
    \[
    \overline{\textbf{opt}}_{\delta_{g}} = \textup{OptComp}((\epsilon_{\mathtt{c}}, \delta_{\mathtt{c}}), \epsilon_{1}, \delta_{1}, \epsilon_{2}, \delta_{2}, \delta_{g}),
    \]
    for any \((s_{1}, s_{0}) \in \mathcal{Q}\) and \(\delta_{g} \in [0, 1)\), where \(\textup{OptComp}\) is given by Theorem 1.5 of \cite{murtagh2015complexity}.
\end{itemize}
\end{theorem}

Consider a collection of independent mechanisms \(\{\mathcal{M}_{i}(\gamma_{i})\}_{i}\), where each \(\mathcal{M}_{i}(\gamma_{i})\) is \((\epsilon_{i}, \delta_{i})\)-CP, guaranteed by noise perturbation with distribution \(\Xi_{i}\).

\begin{definition}[Invertible Noise Perturbation]
A noise perturbation scheme is \textup{invertible} if there exists a bijective mapping \(\mathtt{A}: \mathcal{V}_{i} \to \mathcal{Y}_{i}\) such that \(y_{i} = \mathtt{A}(v_{i})\) if and only if \(v_{i} = \mathtt{A}^{-1}(y_{i})\), for \(i \in \{1, 2\}\).
\end{definition}

Without loss of generality, select \(\mathcal{M}_{1}(\gamma_{1})\) and \(\mathcal{M}_{2}(\gamma_{2})\) such that they are perturbed by the copula \(C^{\mathtt{g}}_{\rho}(\Xi_{1}, \Xi_{2}; \epsilon_{\mathtt{c}}, \delta_{\mathtt{c}})\). Additionally, define
\[
\textbf{L}_{s_{0}, s_{1}}(Y_{1}, Y_{2}; c^{\mathtt{g}}_{\rho}) \equiv \log \frac{c^{\mathtt{g}}_{\rho}(\Psi_{1}(Y_1 \mid s_{1}), \Psi_{2}(Y_2 \mid s_{1}) \mid \epsilon_{\mathtt{c}}, \delta_{\mathtt{c}})}{c^{\mathtt{g}}_{\rho}(\Psi_{1}(Y_1 \mid s_{0}), \Psi_{2}(Y_2 \mid s_{0}) \mid \epsilon_{\mathtt{c}}, \delta_{\mathtt{c}})},
\]
where \(c^{\mathtt{g}}_{\rho}(\cdot \mid \epsilon_{\mathtt{c}}, \delta_{\mathtt{c}})\) is the Gaussian copula density of \(C^{\mathtt{g}}_{\rho}\), and \((Y_{1}, Y_{2}) \sim c^{\mathtt{g}}_{\rho}(\Psi_{1}(Y_1 \mid s_{1}), \Psi_{2}(Y_2 \mid s_{1}))\).

\begin{proposition}\label{prop:copula_2}
Suppose the noise perturbation is invertible. Under the composition of \(k\) mechanisms \(\{\mathcal{M}_{j}(\gamma_{j})\}_{j=1}^{k}\), where \(\{\mathcal{M}_{1}(\gamma_{1}), \mathcal{M}_{2}(\gamma_{2})\}\) are perturbed by the copula \(C^{\mathtt{g}}_{\rho}(\Xi_{1}, \Xi_{2}; \epsilon_{\mathtt{c}}, \delta_{\mathtt{c}})\), the privacy loss random variable (PLRV) is
\[
\widetilde{\mathbf{L}}_{b_{s_{0}}, b_{s_{1}}}(\vec{Y}; C^{\mathtt{g}}_{\rho}) = \textbf{L}_{b_{s_{0}}, b_{s_{1}}}(\vec{Y}) + \textbf{L}_{s_{0} \| s_{1}}(Y_{1}, Y_{2}; c^{\mathtt{g}}_{\rho}),
\]
where \(\textbf{L}_{b_{s_{0}}, b_{s_{1}}}(\vec{Y})\) is given by \((\ref{eq:privacy_loss_RV})\).
\end{proposition}

Proposition \ref{prop:copula_2} demonstrates that the final PLRV can be additively decomposed into the PLRV induced by the copula perturbation and the unperturbed PLRV of the composition.

In the copula perturbation, we use \( N\left(0, \mathtt{w}^2\left(\frac{\mathtt{C}_{\textup{sen}}}{\epsilon_{\mathtt{c}}}\right)^2\right) \) with \( \mathtt{w} \geq 2\log(2/\delta_{\mathtt{c}}) \) and \( N(0,1) \) to generate the latent variables \( z_{1} \) and \( z_{2} \), respectively. Additionally, the mechanism \( \tilde{\eta}_{S} \) introduces state dependence, where the latent variable \( z_{1} \) acts as Gaussian noise to perturb the output of \( \tilde{\eta}_{S} \). As a result, this Gaussian copula achieves \((\epsilon_{\mathtt{s}}, \delta_{\mathtt{c}})\)-CP and contributes to the total cumulative privacy loss of the mechanisms under composition.

The copula perturbation is straightforward to implement, with state dependence captured entirely by the deterministic mechanism \( \tilde{\eta}_{S} \), independent of the underlying Gaussian copula distributions. Alternative approaches to copula perturbation are also possible. For example, the parameters of the bivariate Gaussian distribution (e.g., correlation coefficient \( \rho \), mean, or variance) used in the copula framework could be selected based on the states, introducing state dependence directly into the copula.

During the pseudo-random sample generation process, noises \( v_{1} \) and \( v_{2} \) are jointly sampled to perturb the selected mechanisms \(\mathcal{M}_{1}(\gamma_{1})\) and \(\mathcal{M}_{2}(\gamma_{2})\). However, the copula perturbation does not require these two mechanisms to operate simultaneously. The key requirement is that the noise used to perturb each mechanism is generated via the pseudo-random sample generation process.

\subsection{Proof of Theorem \ref{thm:copula_1} in Appendix \ref{app:gaussian_copula_per}}

Part (i) of Theorem \ref{thm:copula_1} demonstrates that our modified Gaussian copula \( C^{\mathtt{g}}_{\rho}(\Xi_{1}, \Xi_{2}; \epsilon_{\mathtt{c}}, \delta_{\mathtt{c}}) \) is a valid copula that retains the marginals as \(\gamma_{1}\) and \(\gamma_{2}\).

Part (ii) of Theorem \ref{thm:copula_1} shows that the conservative upper bound provided by Theorem \ref{thm:order_composition} can be computed using Theorem 1.5 of \cite{murtagh2015complexity}, given \((\epsilon_{\mathtt{c}}, \delta_{\mathtt{c}})\), which parameterize the Gaussian distribution of \(Z_{1}\) in the process \(\textup{PsedR}_{\rho}(\cdot \mid \Xi_{1}, \Xi_{2})\).

For simplicity, we omit \(\epsilon_{\mathtt{c}}\) and \(\delta_{\mathtt{c}}\), and write \(C^{\mathtt{g}}_{\rho}(\Xi_{1}, \Xi_{2})\).

\subsubsection{\textbf{Proof of Part (i)}}

Part (i) of Theorem \ref{thm:copula_1} can be proved by showing that 
\[
C^{\mathtt{g}}_{\rho}(V_{1}, V_{2} \mid \Xi_{1}, \Xi_{2}) = \Phi_{\rho'} \left(F_{1}^{-1}(\Xi_{1}(V_{1})), F_{\rho}^{-1}(\Xi_{2}(V_{2}))\right)
\]
is a valid copula for any \(\rho' \in (-1, 1)\) and \(\rho \in (-1, 1) \setminus \{0\}\), ensuring that the marginals of the joint distribution \(\overline{\textup{PR}}(Z_{1}, Z_{2} \mid \Xi_{1}, \Xi_{2})\) are \(\Xi_{1}\) and \(\Xi_{2}\).

Let
\[
C^{g}_{\rho}(U_{1}, U_{2}) = \Phi_{\rho'}\left(F_{1}^{-1}(U_{1}), F_{\rho}^{-1}(U_{2})\right).
\]
The function \(C^{\mathtt{g}}_{\rho}(V_{1}, V_{2} \mid \Xi_{1}, \Xi_{2})\) is a valid copula if the following conditions hold:
\begin{itemize}
    \item \(C^{g}_{\rho}(U_{1}, U_{2})\) is a valid joint CDF.
    \item \(C^{g}_{\rho}(U_{1}, U_{2})\) has uniform marginals.
\end{itemize}
Condition (a) follows directly from the definition of \(\Phi_{\rho'}\), which is the joint CDF of a bivariate Gaussian distribution with zero mean and correlation coefficient \(\rho'\).

\textbf{Uniform Marginal for \(U_{1}\) }
Given any \(s\), the CDF of \(z_{1}(s)\) is given by
\[
F_{1}(z_{1}(s)) = \Phi\left(\frac{z_{1} - \tilde{\eta}_{S}(s)}{\sqrt{\mathtt{var}_{1}}}\right),
\]
where \(\Phi\) is the CDF of the standard Gaussian distribution, and \(\mathtt{var}_{1} = \mathtt{w}^2\left(\frac{\mathtt{C}_{\textup{sen}}}{\epsilon_{\mathtt{c}}}\right)^2\).
By the probability integral transform, \(U_{1} = F_{1}(Z_{1}(s))\) is uniformly distributed on \([0, 1]\) because \(F_{1}\) maps the real line to \([0, 1]\) and is a continuous, strictly increasing function.

\textbf{Uniform Marginal for \(U_{2}\) }
Given any \(s\), \(F_{\rho}\) is the CDF of the random variable 
\[
\hat{Z}_{2}(s) = \rho Z_{1}(s) + \sqrt{1 - \rho^{2}} Z_{2},
\]
where \(Z_{2} \sim N(0, 1)\). Therefore, 
\[
F_{\rho}(\hat{z}_{2}(s)) = \Phi\left(\frac{\hat{z}_{2}(s) - \rho\tilde{\eta}_{S}(s)}{\sqrt{\rho^{2} \mathtt{var}_{1} + (1 - \rho^{2})}}\right).
\]
By the probability integral transform, \(U_{2} = F_{\rho}(\hat{Z}_{2}(s))\) is uniformly distributed on \([0, 1]\) because \(F_{\rho}\) is a valid CDF that maps \(\hat{Z}_{2}(s)\) to a uniformly distributed random variable on \([0, 1]\).

Therefore, \(C^{\mathtt{g}}_{\rho}(V_{1}, V_{2} \mid \Xi_{1}, \Xi_{2})\) is a valid copula with marginal distributions \(\Xi_{1}\) and \(\Xi_{2}\). Here, each \(\Xi_{i}\) is independently chosen (without considering composition) such that the mechanism \(\mathcal{M}_{i}(\gamma_{i})\) is \((\epsilon_{i}, \delta_{i})\)-CP for \(i \in \{1, 2\}\). 

Furthermore, if an attacker only has access to one of the two mechanisms, each \(\mathcal{M}_{i}(\gamma_{i})\) of the \(C^{\mathtt{g}}_{\rho}(V_{1}, V_{2} \mid \Xi_{1}, \Xi_{2})\)-perturbed \(\{\mathcal{M}_{1}(\gamma_{1}), \mathcal{M}_{2}(\gamma_{2})\}\) maintains \((\epsilon_{i}, \delta_{i})\)-CP.

\subsubsection{\textbf{Proof of Part (ii)}}

Recall that \(\mathtt{C}_{\textup{sen}} = \sup_{(s_{1}, s_{0}) \in \mathcal{Q}} |\tilde{\eta}_{S}(s_{0}) - \tilde{\eta}_{S}(s_{1})|\) is the sensitivity of a function \(\tilde{\eta}_{S} : \mathcal{S} \to \mathcal{B}\), which outputs a univariate value \(\tilde{\eta}_{S}(s)\) given the state \(s\). 

For any state \(s \in \mathcal{S}\), let 
\[
z_{1}(s) = \tilde{\eta}_{S}(s) + \tilde{z}_{1},
\]
where \(\tilde{z}_{1}\) is drawn from \(N\left(0, \mathtt{w}^2\left(\frac{\mathtt{C}_{\textup{sen}}}{\epsilon_{\mathtt{c}}}\right)^2\right)\), with \(\mathtt{w} \geq 2\log(2/\delta_{\mathtt{c}})\), \(\epsilon_{\mathtt{c}} \geq 0\), and \(\delta_{\mathtt{c}} \in [0, 1)\).

It is straightforward to verify that \(z_{1}(s_{0})\) and \(z_{1}(s_{1})\) are normally distributed with means \(\tilde{\eta}_{S}(s_{0})\) and \(\tilde{\eta}_{S}(s_{1})\), respectively, and a common variance 
\[
\mathtt{var}_{1} = \mathtt{w}^2\left(\frac{\mathtt{C}_{\textup{sen}}}{\epsilon_{\mathtt{c}}}\right)^2.
\]

For any \(s\), let
\[
C^{g}_{\rho}(U_{1}(s), U_{2}(s)) = \Phi_{\rho'}\left(F^{-1}_{1}(U_{1}), F^{-1}_{\rho}(U_{2})\right),
\]
where \(U_{1}(s) = F_{1}(Z_{1}(s))\) and \(U_{2}(s) = F_{\rho}(\hat{Z}_{2}(s))\). 

By Theorem 4.3 in \cite{dwork2022differential}, \(z_{1}(s)\) satisfies \((\epsilon_{\mathtt{c}}, \delta_{\mathtt{c}})\)-SDP. By the post-processing property of differential privacy, \(C^{g}_{\rho}(U_{1}(s), U_{2}(s))\) is also \((\epsilon_{\mathtt{c}}, \delta_{\mathtt{c}})\)-SDP.

Here, the formulation \(\hat{Z}_{2}(s) = \rho Z_{1}(s) + \sqrt{1 - \rho^{2}} Z_{2}\) is deterministic and represents a linear combination of \(Z_{1}(s)\) and \(Z_{2}\), where \(\rho \neq -1\) or \(1\) to avoid degenerate cases. The CDFs are strictly monotonic functions. In addition, all inverse CDFs are assumed to be continuous and well-defined.

Therefore, the process \(C^{g}_{\rho}(U_{1}(s), U_{2}(s))\) is a deterministic one-to-one mapping from \((z_{1}(s), z_{2})\). Thus, the parameters \((\epsilon_{\mathtt{c}}, \delta_{\mathtt{c}})\) of \(C^{g}_{\rho}(U_{1}(s), U_{2}(s))\)'s CP are tight given \((z_{1}(s), z_{2})\).

Recall that the privacy loss random variable can be defined as
\[
\begin{aligned}
    \widehat{\textbf{L}}_{s_{0}, s_{1}}(U_{1}(s), U_{2}(s), Y_{1}, Y_{2}) &\equiv \textbf{L}_{\hat{c}_{s_{0}}, \hat{c}_{s_{1}}}(U_{1}(s), U_{2}(s)) \\
    &\quad + \sum_{i=1}^{2} \textbf{L}_{\psi_{i}^{s_{0}}, \psi_{i}^{s_{1}}}(Y_{i}),
\end{aligned}
\]
where \(\hat{c}_{s} = c^{\mathtt{g}}_{\rho}(U_{1}(s), U_{2}(s))\).

If the copula is treated as an independent mechanism accessing any dataset related to \(s\), we can derive the conservative bound 
\[
\overline{\textbf{opt}}_{\delta_{g}}= \textup{OptComp}((\epsilon_{\mathtt{c}}, \delta_{\mathtt{c}}), \epsilon_{1}, \delta_{1}, \epsilon_{2}, \delta_{2}, \delta_{g}),
\]
for any \((s_{1}, s_{0}) \in \mathcal{Q}\) and \(\delta_{g} \in [0, 1)\), where \(\textup{OptComp}\) is provided by Theorem 1.5 of \cite{murtagh2015complexity}.

\hfill $\square$

\subsection{Proof of Proposition \ref{prop:copula_2} in Appendix \ref{app:gaussian_copula_per}}

For simplicity, we omit \((\epsilon_{\mathtt{c}}, \delta_{\mathtt{c}})\) and write \(C^{\mathtt{g}}_{\rho}(\Xi_{1}, \Xi_{2}; \epsilon_{\mathtt{c}}, \delta_{\mathtt{c}})\).

Recall that \(\psi_{i} : \mathcal{S} \to \Delta(\mathcal{Y}_{i})\) is the density function induced by \(\gamma_{i}\) and \(\theta\), and \(\Psi_{i}\) is the corresponding CDF.

Suppose the noise perturbation is given by \(y_{i} = \mathtt{A}(v_{i})\), where \(\mathtt{A}\) is strictly increasing. Then, \(\mathtt{A}^{-1}(Y_{i}) \leq v_{i}\) is equivalent to \(Y_{i} \leq \mathtt{A}(v_{i})\). Thus, we have
\[
\Xi_{i}(v_{i}) = \textup{Pr}\left(\mathtt{A}^{-1}(Y_{i}) \leq v_{i}\right) = \Psi_{i}(\mathtt{A}(v_{i}) \mid s).
\]

Similarly, suppose the noise perturbation is given by \(y_{i} = \mathtt{A}(v_{i})\), where \(\mathtt{A}\) is strictly decreasing. Then, \(\mathtt{A}^{-1}(Y_{i}) \leq v_{i}\) is equivalent to \(Y_{i} \geq \mathtt{A}(v_{i})\). Thus, we have
\[
\Xi_{i}(v_{i}) = \textup{Pr}\left(\mathtt{A}^{-1}(Y_{i}) \geq v_{i}\right) = 1 - \Psi_{i}(\mathtt{A}(v_{i}) \mid s).
\]

Therefore, when the invertible \(\mathtt{A}\) is strictly increasing or strictly decreasing, we have:
\[
\begin{aligned}
    &\widehat{C}^{\mathtt{g}}_{\rho'}(V_{1}, V_{2} \mid \Xi_{1}, \Xi_{2}; s) \\
    &= \Phi_{\rho'}\left(F_{1}^{-1}(\Psi_{1}(\mathtt{A}(V_{1}) \mid s)), F_{\rho}^{-1}(\Psi_{2}(\mathtt{A}(V_{2}) \mid s))\right),
\end{aligned}
\]
and
\[
\begin{aligned}
    &\widetilde{C}^{\mathtt{g}}_{\rho'}(V_{1}, V_{2} \mid \Xi_{1}, \Xi_{2}; s) \\
    &= \Phi_{\rho'}\left(F_{1}^{-1}(1 - \Psi_{1}(\mathtt{A}(V_{1}) \mid s)), F_{\rho}^{-1}(1 - \Psi_{2}(\mathtt{A}(V_{2}) \mid s))\right).
\end{aligned}
\]

In addition, let \(\widehat{c}^{\mathtt{g}}_{\rho'}(V_{1}, V_{2} \mid s)\) and \(\widetilde{c}^{\mathtt{g}}_{\rho'}(V_{1}, V_{2} \mid s)\) denote the corresponding copula densities.

For the two mechanisms \(\mathcal{M}_{1}(\gamma_{1})\) and \(\mathcal{M}_{2}(\gamma_{2})\) perturbed by \(C^{\mathtt{g}}_{\rho}(V_{1}, V_{2} \mid \Xi_{1}, \Xi_{2})\), let \(\Psi_{[2]}(Y_{1}, Y_{2} \mid s, C^{\mathtt{g}}_{\rho})\) denote the joint CDF of the outputs \(Y_{1}\) and \(Y_{2}\).

\begin{lemma}\label{lemma:proof_copula_for_outputs}
The following holds:
\begin{itemize}
    \item[(i)] If the invertible \(\mathtt{A}\) is strictly increasing, then
    \[
    \Psi_{[2]}(Y_{1}, Y_{2} \mid s, C^{\mathtt{g}}_{\rho}) = \widehat{C}^{\mathtt{g}}_{\rho}(Y_{1}, Y_{2} \mid \Xi_{1}, \Xi_{2}; s).
    \]
    \item[(ii)] If the invertible \(\mathtt{A}\) is strictly decreasing, then
    \[
    \Psi_{[2]}(Y_{1}, Y_{2} \mid s, C^{\mathtt{g}}_{\rho}) = \widetilde{C}^{\mathtt{g}}_{\rho}(Y_{1}, Y_{2} \mid \Xi_{1}, \Xi_{2}; s).
    \]
\end{itemize}
\end{lemma}

\begin{proof}
Suppose the invertible \(\mathtt{A}\) is strictly increasing. Then,
\[
\begin{aligned}
    \Psi_{[2]}(y_{1}, y_{2} \mid s, C^{\mathtt{g}}_{\rho}) &= \textup{Pr}\left(Y_{1} \leq y_{1}, Y_{2} \leq y_{2}\right) \\
    &= \textup{Pr}\left(\mathtt{A}(V_{1}) \leq y_{1}, \mathtt{A}(V_{2}) \leq y_{2}\right) \\
    &= \textup{Pr}\left(V_{1} \leq \mathtt{A}^{-1}(y_{1}), V_{2} \leq \mathtt{A}^{-1}(y_{2})\right).
\end{aligned}
\]

Since \(V_{i} \leq \mathtt{A}^{-1}(y_{i})\) is equivalent to \(Y_{i} = \mathtt{A}(V_{i}) \leq y_{i}\), we have
\[
\Xi_{i}(\mathtt{A}^{-1}(y_{i})) = \textup{Pr}\left(V_{i} \leq \mathtt{A}^{-1}(y_{i})\right) = \Psi_{i}(y_{i}),
\]
for all \(i \in \{1, 2\}\).

Since \(V_{1}\) and \(V_{2}\) have the joint CDF \(C^{\mathtt{g}}_{\rho}(\cdot \mid \Xi_{1}, \Xi_{2})\), it follows that
\[
\begin{aligned}
    &\Psi_{[2]}(y_{1}, y_{2} \mid s, C^{\mathtt{g}}_{\rho}) = C^{\mathtt{g}}_{\rho}(\mathtt{A}^{-1}(y_{1}), \mathtt{A}^{-1}(y_{2}) \mid \Xi_{1}, \Xi_{2}) \\
    =& \Phi_{\rho'}\left(F_{1}^{-1}(\Psi_{1}(\mathtt{A}(\mathtt{A}^{-1}(y_{1})) \mid s)), F_{\rho}^{-1}(\Psi_{2}(\mathtt{A}(\mathtt{A}^{-1}(y_{2})) \mid s))\right) \\
    =& \Phi_{\rho'}\left(F_{1}^{-1}(\Psi_{1}(y_{1} \mid s)), F_{\rho}^{-1}(\Psi_{2}(y_{2} \mid s))\right).
\end{aligned}
\]

A similar procedure can be used to prove (ii).
\end{proof}

Lemma \ref{lemma:proof_copula_for_outputs} shows that the dependency between \(Y_{1}\) and \(Y_{2}\) due to the copula perturbation can be fully characterized by the same structure \(C^{\mathtt{g}}_{\rho'}\) as the noises \(V_{1}\) and \(V_{2}\).

Under the copula perturbation using \(C^{\mathtt{g}}_{\rho'}\), let \(b(\cdot \mid s; c^{\mathtt{g}}_{\rho'})\) denote the joint density. Then,
\begin{equation*}%\label{eq:joint_copula_density}
    \begin{aligned}
        &b\left(\vec{Y} \mid s\right) \\
        &= c\left(\Psi_{[2]}(Y_{1}, Y_{2} \mid s, C^{\mathtt{g}}_{\rho}), \Psi_{3}(Y_{3} \mid s), \dots, \Psi_{k}\left(Y_{k} \mid s\right) \mid s\right) \\
        &\quad \times \psi_{[2]}(Y_{1}, Y_{2} \mid s, C^{\mathtt{g}}_{\rho}) \prod_{i=3}^{k} \psi_{i}\left(Y_{i} \mid s\right) \\
        &= c\left(\Psi_{[2]}(Y_{1}, Y_{2} \mid s, C^{\mathtt{g}}_{\rho}), \Psi_{3}(Y_{3} \mid s), \dots, \Psi_{k}\left(Y_{k} \mid s\right) \mid s\right) \\
        &\quad \times c^{\mathtt{g}}_{\rho}(y_{1}, y_{2} \mid s) \prod_{i=1}^{k} \psi_{i}\left(Y_{i} \mid s\right),
    \end{aligned}
\end{equation*}
where \(\psi_{[2]}\) is the joint density of the joint CDF \(\Psi_{[2]}\).

For any \((s_{0}, s_{1}) \in \mathcal{Q}\), define
\[
\begin{aligned}
    \widehat{\textbf{L}}_{s_{0}, s_{1}}(Y_{1}, Y_{2}; \widehat{c}^{\mathtt{g}}_{\rho}) \equiv \log \frac{\widehat{c}^{\mathtt{g}}_{\rho}(\Psi_{1}(Y_{1} \mid s_{1}), \Psi_{2}(Y_{2} \mid s_{1}) \mid \epsilon_{\mathtt{c}}, \delta_{\mathtt{c}})}{\widehat{c}^{\mathtt{g}}_{\rho}(\Psi_{1}(Y_{1} \mid s_{0}), \Psi_{2}(Y_{2} \mid s_{0}) \mid \epsilon_{\mathtt{c}}, \delta_{\mathtt{c}})},
\end{aligned}
\]
and
\[
\begin{aligned}
    \widetilde{\textbf{L}}_{s_{0}, s_{1}}(Y_{1}, Y_{2}; \widetilde{c}^{\mathtt{g}}_{\rho}) \equiv \log \frac{\widetilde{c}^{\mathtt{g}}_{\rho}(\Psi_{1}(Y_{1} \mid s_{1}), \Psi_{2}(Y_{2} \mid s_{1}) \mid \epsilon_{\mathtt{c}}, \delta_{\mathtt{c}})}{\widetilde{c}^{\mathtt{g}}_{\rho}(\Psi_{1}(Y_{1} \mid s_{0}), \Psi_{2}(Y_{2} \mid s_{0}) \mid \epsilon_{\mathtt{c}}, \delta_{\mathtt{c}})}.
\end{aligned}
\]

By the product rule for logarithms, we can separate the copula densities \(\widehat{c}^{\mathtt{g}}_{\rho}\) and \(\widetilde{c}^{\mathtt{g}}_{\rho}\), obtaining
\[
\begin{aligned}
    \widetilde{\textbf{L}}_{b_{s_{0}}, b_{s_{1}}}(\vec{Y}; C^{\mathtt{g}}_{\rho}) = \textbf{L}_{b_{s_{0}}, b_{s_{1}}}(\vec{Y}) + \textbf{L}_{s_{0}, s_{1}}(Y_{1}, Y_{2}; \widehat{c}^{\mathtt{g}}_{\rho}),
\end{aligned}
\]
and
\[
\begin{aligned}
    \widetilde{\textbf{L}}_{b_{s_{0}}, b_{s_{1}}}(\vec{Y}; C^{\mathtt{g}}_{\rho}) = \textbf{L}_{b_{s_{0}}, b_{s_{1}}}(\vec{Y}) + \widetilde{\textbf{L}}_{s_{0}, s_{1}}(Y_{1}, Y_{2}; \widetilde{c}^{\mathtt{g}}_{\rho}).
\end{aligned}
\]

\hfill \(\square\)

%%%%%%%%%%%%%%%%%%%%%
\subsection{Proof of Proposition \ref{prop:invertible_G}}
%{prop:invertible_G}

When $\mathcal{G}$ is an invertible mapping, for any secret $s\in\mathcal{S}$, let $x_{s}\in\mathcal{X}$ be the dataset such that $P_{\theta}(x_{s}|s)=\frac{P_{\theta}(s,x_{s})}{P_{\theta}(s)}=1$.
Then, the density function of the effective mechanism $\mathcal{N}(\psi)$ becomes
\[
\begin{aligned}
    \psi_{i}(y_{i} \mid s) &= \int_{x_{s} \in \mathcal{X}} \gamma_{i}(y_{i} \mid x_{s}) \frac{P_{\theta}(s, x_{s})}{P_{\theta}(s)} \, dx_{s}\\
    & =\gamma_{i}(y_{i} \mid x_{s}).
\end{aligned}
\]
The joint probability under composition given by (\ref{eq:likelihood_composition}) becomes
\[
\begin{aligned}
    \textbf{Pr}^{s}_{\vec{\gamma}}\left[\vec{y} \in \vec{\mathcal{W}}\right] &= \int_{x \in \mathcal{X}} \int_{\vec{y} \in \vec{\mathcal{W}}} \prod_{i=1}^{k} \gamma_{i}(y_{i} \mid x) \frac{P_{\theta}(s, x)}{P_{\theta}(s)} \, d\vec{y} \, dx\\
    &= \int_{\vec{y} \in \vec{\mathcal{W}}} \prod_{i=1}^{k} \gamma_{i}(y_{i} \mid x_{s})  d\vec{y}.
\end{aligned}
\]
Thus, the joint density under composition becomes the product of individual densities $\vec{\gamma}$.
Then, the PLRV of the composition becomes the sum of individual PLRVs. That is,
\[
\begin{aligned}
    \textbf{L}_{b_{s_{0}}, b_{s_{1}}}(\vec{Y}) =  \textbf{L}^{\textup{id}}_{\psi_{s_{0}}, \psi_{s_{1}} }(\vec{Y}), 
\end{aligned}
\]
which implies $\textbf{L}^{\mathcal{G}}_{c_{s_{0}}, c_{s_{1}}}(\vec{Y})=0$.

\hfill $\square$

%%%%%%%%%%%%%%%%%%%%%%%

\subsection{Proof of Part (i) of Theorem \ref{thm:order_composition_pre} and Proposition \ref{prop:existence_composition_DCP}}

For any $\theta\in\Theta$, let $\mathcal{Q}\equiv\{(s,s')|\mathtt{D}(s,s'|\theta)\leq\mathtt{d}\}$ for simplicity.
Our proof depends on a hypothesis test result.

\subsubsection{Hypothesis Testing}\label{app_sec:hypothesis_testing}

The notion of $f$-differential privacy \cite{dong2022gaussian} generalizes the differential privacy by interpreting the worst-case scenario differential privacy as a Neyman-Pearson optimal hypothesis test.
It is not hard to extend $f$-differential privacy ($f$-DP) to our confounding privacy settings.
For $(s_{0}, s_{1})\in\mathcal{Q}$, consider the following binary hypothesis testing: 
\begin{center}
    $H_{0}$: the state is $s_{0}$. versus $H_{1}$: the state is $s_{1}$.
\end{center}

We rewrite $\textbf{Pr}^{\theta,s}_{\gamma}\left[y \in \mathcal{W}\right] $ in (\ref{eq:prior_likelihood}) by $\textbf{Pr}\left[\cdot\middle|H_{t},\gamma\right]$, for $t\in\{0,1\}$.
For any rejection rule $\phi$, define the Type-I error $\alpha_{\phi}$ and the Type-II error $\beta_{\phi}$ by 
\[
\alpha_{\phi}=\mathbb{E}[\phi|H_{0},\gamma], \beta_{\phi}=1-\mathbb{E}[\phi|H_{1},\gamma].
\]

\begin{definition}[Optimal Symmetric Trade-off Function \cite{dong2022gaussian}]
    A function $f:[0,1]\mapsto [0,1]$ is a \textup{symmetric trade-off function} if it is convex, continuous, non-increasing, $f(\alpha)\leq 1- \alpha$, and $f(\alpha)=\inf\{t\in[0,1]:f(t)\leq \alpha\}$ for $\alpha\in[0,1]$.
\end{definition}

\begin{lemma}\label{lemma:existence_fDP}
Fix any $\alpha\in[0,1]$. For any $\theta\in\Theta$, there exists a symmetric trade-off function $f$ and $\epsilon\geq 0$ for the mechanism $\mathcal{M}(\gamma)$ such that $\mathcal{M}$ is $(\epsilon, \delta(\epsilon))$-DCP, where 
    \[
    \begin{aligned}
        &\delta(\epsilon) = 1+f^{*}(-e^{\epsilon}) \textup{ with } f^{*}(y)=\sup\nolimits_{0\leq \alpha' \leq 1} y\alpha' - f(\alpha').
    \end{aligned}
    \]
\end{lemma}

\begin{proof}

Lemma \ref{lemma:existence_fDP} is a corollary of Proposition 2.4 and Proposition 2.12  of \cite{dong2022gaussian}. 
To prove Lemma \ref{lemma:existence_fDP}, we first need to show that there exists uniformly most powerful test for the problem:
\begin{center}
    $H_{0}$: the state is $s_{0}$. versus $H_{1}$: the state is $s_{1}$.
\end{center}
It is straightforward to see that this is a simple binary hypothesis testing problem. 
Hence, Neyman-Pearson lemma implies that the likelihood-ratio test is the UMP test.
Therefore, for any significance level $\alpha\in[0,1]$, there is 
\[
\beta^{\alpha} = \inf\nolimits_{\phi}\left\{\beta_{\phi}: \alpha_{\phi}\leq \alpha\right\},
\]
such that $\alpha = \mathbb{E}[\phi|H_{0}, \gamma]$ and $\beta^{\alpha}=1-\mathbb{E}[\phi|H_{1}, \gamma]$.
Therefore, the mechanism $\mathcal{M}(\gamma)$ is $f$-differentially private (Definition 2.3 of \cite{dong2022gaussian}).
Then, by Proposition 2.4 of \cite{dong2022gaussian}, $\mathcal{M}(\gamma)$ is also $f^{*}$-differentially private, where $f^{*}(y)=\sup\nolimits_{0\leq \alpha' \leq 1} y\alpha' - f(\alpha')$. 
Then, we obtain $\delta(\epsilon) = 1+f^{*}(-e^{\epsilon})$ from Proposition 2.12 of \cite{dong2022gaussian}.

\end{proof}

Lemma \ref{lemma:existence_fDP} shows that any randomized mechanism with well-defined probability distributions satisfies $(\epsilon, \delta)$-DCP for some $\epsilon\geq 0$ and $\delta\in[0,1]$.
It is a corollary of Proposition 2.4 and Proposition 2.12  of \cite{dong2022gaussian}.

\begin{lemma}\label{lemma:composition_of_cp_is_cp}
    Fix any $\theta$. Suppose that each $\mathcal{M}_{i}(\gamma_{i})$ is $(\epsilon_{i}, \delta_{i})$-DCP with $\epsilon_{i}\geq 0$ and $\delta_{i}\in[0,1)$ for all $i\in[k]$, and $\mathcal{M}(\vec{\gamma})$ is the $k$-fold composition with $f_{[k]}$ as the optimal symmetric trade-off function. Then, the mechanism $\mathcal{M}(\vec{\gamma})$ is $(\epsilon'_{g}, \delta(\epsilon'_{g}))$-DCP for some $\epsilon'_{g}\geq 0$.
\end{lemma}

\begin{proof}

Consider the following hypothesis testing: $H_{0}$: the state is $s_{0}$. versus $H_{1}$: the state is $s_{1}$.
Under the composition, given any $\vec{y}$, the likelihood ratio is given by
\[
\begin{aligned}
    \textup{LR}(\vec{y}) = \frac{ \ell(\vec{y}|H_{0})}{\ell(\vec{y}|H_{1})}, 
\end{aligned}
\]
where
\[
\ell(\vec{y}|H_{z}) = \int\limits_{x}\prod_{i=1}^{k}\gamma_{i}(y_{i}|x)\frac{P_{\theta}(s_{z},x)}{P_{\theta}(s_{z})}dx,
\]
for $z\in\{0,1\}$, which is well-defined.
Then, by Neyman-Pearson lemma, the likelihood-ratio test is the UMP test and the proof follows Lemma \ref{lemma:existence_fDP}.

\end{proof}

%%%%%%%%%%%%%%%%%%%%%
\subsection{Proof of Parts (ii) and (iii) of Theorem \ref{thm:order_composition_pre} and Theorem \ref{thm:order_composition}}

For any \(\theta \in \Theta\), we define \(\mathcal{Q} \equiv \{(s, s') \mid \mathtt{D}(s, s' \mid \theta) \leq \mathtt{d}\}\) for simplicity. The proof of this theorem involves notions and definitions from hypothesis testing, as outlined in Appendix \ref{app_sec:hypothesis_testing}.
We start by proving part (i) of Theorem \ref{thm:order_composition}.

%%%%%%%%%%%%%%%%%%%%%%%%%%%%%%%%%%%%%%%%%%%%%%%%%%%%%%
\subsubsection{\textbf{Proof of Part (i) of Theorem \ref{thm:order_composition}}}

When $\textbf{L}^{\mathcal{G}}_{c_{s_{0}}, c_{s_{1}}}(\vec{Y})=0$ for all $(s_{0}, s_{1})\in\mathcal{Q}$, the PLRV under composition becomes $\textbf{L}^{\textup{id}}_{\psi_{s_{0}}, \psi_{s_{1}} }(\vec{Y})$, the sum of independent PLRVs. 
In addition, the basic composition of DP implies that the composition $\mathcal{M}(\vec{\gamma})$ is $(\hat{\epsilon}_{g}, \delta^{\dagger}_{g})$-DCP, where $\hat{\epsilon}_{g} = \sum^{k}_{i=1}\epsilon_{i}$ and $\delta^{\dagger}_{g}=\sum^{k}_{i=1}\delta_{i}$.

Let
\[
 \hat{\delta}_{g} \equiv\mathbb{E}_{\textbf{L}\sim \textup{PLD}^{\mathtt{in}}}\left[\left(1-e^{\hat{\epsilon}_{g}-\textbf{L}}\right)^{+}\right],
\]
where $\textup{PLD}^{\mathtt{in}}$ is the PLD of $\textbf{L}^{\textup{id}}_{\psi_{s_{0}}, \psi_{s_{1}} }(\vec{Y})$.
In addition, let $\epsilon^{\dagger}_{g}$ such that
\[
\mathbb{E}_{\textbf{L}\sim \textup{PLD}^{\mathtt{in}}}\left[\left(1-e^{\epsilon^{\dagger}_{g}-\textbf{L}}\right)^{+}\right] = \delta^{\dagger}_{g}\equiv\sum^{k}_{i=1}\delta_{i}.
\]
By Proposition 7 of \cite{steinke2022composition}, (when $\textbf{L}^{\mathcal{G}}_{c_{s_{0}}, c_{s_{1}}}(\vec{Y})=0$), the following holds:
\begin{itemize}
    \item If we fix the ratio at $e^{\hat{\epsilon}_{g}}$, the additive ‘excess’ that the event’s probability under $s$ can have—beyond that scaled probability under $s'$—is at most $\hat{\delta}_{g}$, where $\hat{\delta}_{g}\leq \delta^{\dagger}_{g}$.

    \item If we fix the additive $\delta^{\dagger}_{g}$, there is no attacker that can violate privacy in terms of epsilon-delta indistinguishability by more than $e^{\epsilon^{\dagger}_{g}}\leq e^{\hat{\epsilon}_{g}}$.
\end{itemize}

For any $D>0$, define
\[
F(D;\epsilon)\equiv \mathbb{E}_{\textbf{L}\sim \textup{PLD}^{\mathtt{in}}}\left[\left(1-\frac{e^{\epsilon-\textbf{L}}}{D}\right)^{+}\right].
\]
For any fixed $\epsilon\geq 0$, it is easy to verify that $F(D;\epsilon)$ is a continuous and strictly increasing function of $D$ for all $D>0$.
In addition, the Intermediate Value Theorem ensures that for any $W\in(0,1]$, there is a $D$ such that $F(D;\epsilon)=W$ for any fixed $\epsilon\geq 0$.
Furthermore, for any $\epsilon\geq 0$, there is a $D$ such that $F(D;\epsilon)=W$ for any fixed $W\in(0,1]$.

Consider:
\begin{itemize}
    \item Let $\hat{D}>1$ such that when $\hat{\epsilon}_{g}$ is fixed, $F(1;\hat{\epsilon}_{g})$ is increased to $F(\hat{D};\hat{\epsilon}_{g})=\delta^{\dagger}_{g}$.

    \item Let $D^{\dagger}>0$ such that when $\delta^{\dagger}_{g}$ is fixed, $\epsilon^{\dagger}$ is changed to $\hat{\epsilon}_{g}$ such that $\delta^{\dagger}_{g} = F(D^{\dagger};\hat{\epsilon}_{g})$.

\end{itemize}

Now, suppose that $\textbf{L}^{\mathcal{G}}_{c_{s_{0}}, c_{s_{1}}}(\vec{Y})\neq 0$ in general.
Then, we have, for any $\epsilon\geq 0$,
\[
\begin{aligned}
    \mathbb{E}_{\textbf{L}\sim \textup{PLD}}\left[\left(1-\frac{e^{\epsilon-\textbf{L}^{\mathtt{in}}}}{e^{\textbf{L}^{\mathcal{G}}}}\right)^{+}\right] = F(e^{\textbf{L}^{\mathcal{G}}};\epsilon).
\end{aligned}
\]
The basic composition theory holds for DCP when $\textbf{L}^{\mathcal{G}}_{c_{s_{0}}, c_{s_{1}}}(\vec{Y})\neq 0$ if and only if
\begin{itemize}
    \item For given $\hat{\epsilon}_{g} = \sum^{k}_{i=1}\epsilon_{i}$,   the basic composition theory holds for DCP if and only if $e^{\textbf{L}^{\mathcal{G}}}\leq \hat{D}$ or $\textbf{L}^{\mathcal{G}}\leq \log \hat{D}$.

    \item For given $\delta^{\dagger}_{g}=\sum^{k}_{i=1}\delta_{i}$, the basic composition theory holds for DCP if and only if $e^{\textbf{L}^{\mathcal{G}}}\leq D^{\dagger}$ or $\textbf{L}^{\mathcal{G}}\leq \log D^{\dagger}$.
\end{itemize}
However, $\textbf{L}^{\mathcal{G}}_{c_{s_{0}}, c_{s_{1}}}(\vec{Y})\neq 0$, in general, cannot guarantee $\textbf{L}^{\mathcal{G}}\leq \log \hat{D}$ or $\textbf{L}^{\mathcal{G}}\leq \log D^{\dagger}$.

%%%%%%%%%%%%%%%%%%%%%%%%%%%%%%%%%%%%%%%%%%%%%%%%%%%%%%
\subsubsection{\textbf{Proof of Part (ii) of Theorem \ref{thm:order_composition}}}

The following theorem aligns with part (ii) of Theorem \ref{thm:order_composition}.

\begin{theorem}\label{thm:standard_compose_fail_general}
    Fix any \(\theta\). Suppose that each \(\mathcal{M}_{i}(\gamma_{i})\) is \((\epsilon_{i}, \delta_{i})\)-CP with \(\epsilon_{i} \geq 0\) and \(\delta_{i} \in [0,1)\) for all \(i \in [k]\), and let \(\mathcal{M}(\vec{\gamma})\) be the composition mechanism with \(f_{[k]}\) as the optimal symmetric trade-off function. 
    Let 
    \[
    \epsilon^{*}(\delta_{g}) = \textup{OPT}\left(\mathcal{N}(\psi_{[k]}), \delta_{g}\right),
    \]
    for any \(\delta_{g} \in [0,1)\), where \(\vec{\psi}\) is induced by \(\vec{\gamma}\) and \(\textup{OPT}\) is given by (\ref{eq:def_opt_comp}). Additionally, let 
    \[
    \delta(\epsilon_{g}) = 1 + f^{*}_{[k]}(-e^{\epsilon_{g}}),
    \]
    for any \(\epsilon_{g} \geq 0\). Then, the following holds:
    \begin{itemize}
        \item[\textup{(i)}] For any \(\epsilon_{g} \geq 0\) such that \(\delta(\epsilon_{g}) \in [0,1)\), we have \(\epsilon_{g} \geq \epsilon^{*}\left(\delta(\epsilon_{g})\right)\).
        \item[\textup{(ii)}] For any \(\delta_{g} \in [0,1)\), \(\delta_{g} \geq \delta(\epsilon^{*}(\delta_{g}))\).
    \end{itemize}

    Equality holds in \textup{(i)} and \textup{(ii)} if and only if 
    \[
    \vec{\psi}(\vec{y} \mid s) = \prod_{i=1}^{k} \psi_{i}(y_{i} \mid s),
    \]
    for all \(s \in \mathcal{S}\), \(\vec{y} \in \prod_{i=1}^{k} \mathcal{Y}_{i}\), and \(i \in [k]\).
\end{theorem}

Theorem \ref{thm:standard_compose_fail_general} demonstrates that the optimal bound on privacy loss, as captured by the optimal composition of standard differential privacy \cite{murtagh2015complexity}, underestimates the privacy degradation in the composition of CP mechanisms, even when the individual mechanisms are independent of each other.

Both Theorems \ref{thm:standard_compose_fail_general} and \ref{thm:order_composition} illustrate that the tightest privacy bound under the optimal composition of DP underestimates the actual cumulative privacy loss of CP. We prove these two theorems together in this section.

For simplicity, we restrict our attention to the composition of two individual mechanisms, without loss of generality. Consider two mechanisms \(\mathcal{M}_{1}(\gamma_{1})\) and \(\mathcal{M}_{2}(\gamma_{2})\) that are \((\epsilon_{1}, \delta_{1})\)-CP and \((\epsilon_{2}, \delta_{2})\)-CP, respectively. Let \(\mathcal{N}_{i}(\psi_{i})\) be induced by \(\mathcal{M}_{i}(\gamma_{i})\) for \(i \in \{1, 2\}\). 

Recall that the density function \(\psi_{i}\) is given by
\[
\psi_{i}(y_{i} \mid s) = \int_{x} \gamma_{i}(y_{i} \mid x) \frac{P_{\theta}(s, x)}{P_{\theta_{S}}(s)} dx,
\]
which is induced by \(\gamma_{i}\) and \(\theta\).

Define
\[
\mathtt{D}_{1,2}(y_{1}, y_{2} \mid s) = \psi_{1}(y_{1} \mid s) \psi_{2}(y_{2} \mid s).
\]
Additionally, the density function associated with \(\mathcal{M}(\vec{\gamma})\) is given by
\[
\mathtt{D}_{[2]}(y_{1}, y_{2} \mid s) = \int_{x} \prod_{i=1}^{2} \gamma_{i}(y_{i} \mid x) \frac{P_{\theta}(s, x)}{P_{\theta_{S}}(s)} dx.
\]

We adopt the general notation of the trade-off function as given in Definition 2.1 of \cite{dong2022gaussian}.

\begin{definition}[Trade-off Function]
Let \(\alpha_{\phi}\) and \(\beta_{\phi}\) be the Type-I and Type-II errors associated with any rejection rule \(\phi\), respectively. For any two probability distributions \(P\) and \(Q\) on the same space, the trade-off function \(T(P, Q) : [0,1] \to [0,1]\) is defined as
\[
T(P, Q)(\alpha) = \inf \{\beta_\phi : \alpha_\phi \leq \alpha\},
\]
where the infimum is taken over all (measurable) rejection rules.
\end{definition}

The following Proposition \ref{prop:blackwell_ording} shows that for any \(\vec{\gamma}\), \(\mathtt{D}_{[2]}\) is more informative than \(\mathtt{D}_{1,2}\) about the state in the sense of Blackwell \cite{blackwell1951comparison,de2018blackwell}.

\begin{proposition}\label{prop:blackwell_ording}
For any \((s_{1}, s_{0}) \in \mathcal{Q}\) and \(\alpha \in [0,1]\), we have
\[
T(\mathtt{D}_{[2]}(\cdot \mid s_{0}), \mathtt{D}_{[2]}(\cdot \mid s_{1})) \leq T(\mathtt{D}_{1,2}(\cdot \mid s_{0}), \mathtt{D}_{1,2}(\cdot \mid s_{1})),
\]
where equality holds if and only if \(\mathtt{D}_{[2]} = \mathtt{D}_{1,2}\).
\end{proposition}

The proof of Proposition \ref{prop:blackwell_ording} relies on Lemmas \ref{lemma:KL_divergence_and_CEL} and \ref{lemma:blackwell_informative}. Specifically:
\begin{itemize}
    \item Lemma \ref{lemma:KL_divergence_and_CEL} shows that Bayesian inference based on \(\mathtt{D}_{[2]}\) results in a smaller expected cross-entropy loss compared to that based on \(\mathtt{D}_{1,2}\).
    \item Lemma \ref{lemma:blackwell_informative} establishes an equivalence between the ordering of expected cross-entropy losses and Blackwell's ordering of informativeness.
\end{itemize}

Thus, by applying Lemmas \ref{lemma:KL_divergence_and_CEL} and \ref{lemma:blackwell_informative}, we can prove Proposition \ref{prop:blackwell_ording}.

\textbf{Detailed Proofs}

By Bayes' rule, the posterior distributions are constructed as follows:
\[
\begin{aligned}
    &\textup{J}_{1,2}(s \mid y_{1}, y_{2}) = \frac{\mathtt{D}_{1,2}(y_{1}, y_{2} \mid s) p_{\theta_{S}}(s)}{\textup{P}_{1,2}(y_{1}, y_{2})}, \\
    &\textup{J}_{[2]}(s \mid y_{1}, y_{2}) = \frac{\mathtt{D}_{[2]}(y_{1}, y_{2} \mid s) p_{\theta_{S}}(s)}{\textup{P}_{[2]}(y_{1}, y_{2})},
\end{aligned}
\]
where \(\textup{P}_{1,2}(y_{1}, y_{2}) = \int_{s} \mathtt{D}_{1,2}(y_{1}, y_{2} \mid s)p_{\theta_{S}}(s) \, ds\) and \(\textup{P}_{[2]}(y_{1}, y_{2}) = \int_{s} \mathtt{D}_{[2]}(y_{1}, y_{2} \mid s)p_{\theta_{S}}(s) \, ds\) are the marginal likelihoods.

For any randomized inference strategy $\sigma : \mathcal{Y}_{1} \times \mathcal{Y}_{2} \to \Delta(\mathcal{S})$,
the expected cross-entropy loss (CEL) induced by \(\{\gamma_{1}, \gamma_{2}, \theta\}\) is given by
\[
\begin{aligned}
    \textup{CEL}(\sigma) =&-\int_{s, x}p_{\theta}(s, x) \int_{y_{1}, y_{2}} \prod_{i=1}^{2} \gamma_{i}(y_{i} \mid x) \\
    &\times\log \sigma(s \mid y_{1}, y_{2}) \, dx \, ds \, dy_{[2]}.
\end{aligned}
\]

\begin{lemma}\label{lemma:KL_divergence_and_CEL}
The following inequality holds:
\[
\textup{CEL}\left(\textup{J}_{[2]}\right) \leq \textup{CEL}\left(\textup{J}_{1,2}\right),
\]
with equality if and only if \(\textup{J}_{[2]} = \textup{J}_{1,2}\).
\end{lemma}

\begin{proof}
We consider the difference between the expected CELs:
\[
\begin{aligned}
    \textup{CEL}\left(\textup{J}_{1,2}\right) - \textup{CEL}\left(\textup{J}_{[2]}\right) &= -\int_{s,x} p_{\theta}(s,x) \int_{y_{1}, y_{2}} \prod_{i=1}^{2} \gamma_{i}(y_{i} \mid x) \\
    &\quad \times \log \frac{\textup{J}_{[2]}(s \mid y_{1}, y_{2})}{\textup{J}_{1,2}(s \mid y_{1}, y_{2})} \, dx \, ds \, dy_{[2]}.
\end{aligned}
\]

By the definitions of \(\mathtt{D}_{[2]}\) and \(\textup{J}_{[2]}\), we have:
\[
\textup{J}_{[2]}(s \mid y_{1}, y_{2}) \textup{P}_{[2]}(y_{1}, y_{2}) = \int_{x} \prod_{i=1}^{2} \gamma_{i}(y_{i} \mid x) p_{\theta}(s, x) \, dx.
\]

Using this, the difference \(\textup{CEL}\left(\textup{J}_{1,2}\right) - \textup{CEL}\left(\textup{J}_{[2]}\right)\) can be rewritten as the Kullback–Leibler (KL) divergence between \(\textup{J}_{[2]}\) and \(\textup{J}_{1,2}\):
\[
\begin{aligned}
    \textup{CEL}\left(\textup{J}_{1,2}\right) - \textup{CEL}\left(\textup{J}_{[2]}\right) &= \textup{KL}\left(\textup{J}_{[2]} \| \textup{J}_{1,2}\right),
\end{aligned}
\]
where
\[
\begin{aligned}
    \textup{KL}\left(\textup{J}_{[2]} \| \textup{J}_{1,2}\right) &= \int_{s} \int_{y_{1}, y_{2}} \textup{J}_{[2]}(s \mid y_{1}, y_{2}) \\
    &\quad \times \log \frac{\textup{J}_{[2]}(s \mid y_{1}, y_{2})}{\textup{J}_{1,2}(s \mid y_{1}, y_{2})} \textup{P}_{[2]}(y_{1}, y_{2}) \, dy_{1} \, dy_{2} \, ds.
\end{aligned}
\]

Since \(\textup{KL}\left(\textup{J}_{[2]} \| \textup{J}_{1,2}\right) \geq 0\), with equality if and only if \(\textup{J}_{[2]} = \textup{J}_{1,2}\), it follows that:
\[
\textup{CEL}\left(\textup{J}_{[2]}\right) \leq \textup{CEL}\left(\textup{J}_{1,2}\right).
\]
\end{proof}

\begin{lemma}\label{lemma:blackwell_informative}
Let \(\mathcal{A}(\sigma)\) with \(\sigma : \mathcal{Y}_{1} \times \mathcal{Y}_{2} \to \Delta(\mathcal{S})\) be any attack model. Let \(\textup{J}_{[2]}\) and \(\textup{J}_{1,2}\) be the posterior distributions induced by \(\mathtt{D}_{[2]}\) and \(\mathtt{D}_{1,2}\), respectively. Then, the following two statements are equivalent:
\begin{itemize}
    \item[(i)] \(\textup{CEL}\left(\textup{J}_{[2]}\right) \leq \textup{CEL}\left(\textup{J}_{1,2}\right)\).
    \item[(ii)] For any \((s_{1}, s_{0}) \in \mathcal{Q}\) and \(\alpha \in [0,1]\),
    \[
    T(\mathtt{D}_{[2]}(\cdot \mid s_{0}), \mathtt{D}_{[2]}(\cdot \mid s_{1})) \leq T(\mathtt{D}_{1,2]}(\cdot \mid s_{0}), \mathtt{D}_{1,2}(\cdot \mid s_{1})).
    \]
\end{itemize}
Equality holds if and only if \(\mathtt{D}_{[2]} = \mathtt{D}_{1,2}\).
\end{lemma}

\begin{proof}
By Theorem 1 of \cite{de2018blackwell}, part (ii) of Lemma \ref{lemma:blackwell_informative} is equivalent to the following:
\begin{itemize}
    \item There exists a randomized mechanism \(\mathtt{Proc} : \mathcal{Y}_{1} \times \mathcal{Y}_{2} \to \mathcal{Z}\) such that 
    \[
    \begin{aligned}
        \mathtt{D}_{1,2}(s_{0}) = \mathtt{Proc}(\mathtt{D}_{[2]}(s_{0})),\; \mathtt{D}_{1,2}(s_{1}) = \mathtt{Proc}(\mathtt{D}_{[2]}(s_{1})),
    \end{aligned}
    \]
    for all \((s_{0}, s_{1}) \in \mathcal{Q}\).
\end{itemize}
Then, by Theorem 10 in \cite{blackwell1951comparison} (see also Theorem 2.10 in \cite{dong2022gaussian} and Theorem 2.5 in \cite{kairouz2015composition}), we establish the equivalence of parts (i) and (ii) of Lemma \ref{lemma:blackwell_informative}.
\end{proof}

Recall that a function \(f : [0,1] \to [0,1]\) is a \textit{symmetric trade-off function} if it is convex, continuous, non-increasing, satisfies \(f(\alpha) \leq 1 - \alpha\) and 
\[
f(\alpha) = \inf\{t \in [0,1] : f(t) \leq \alpha\},
\]
for \(\alpha \in [0,1]\) \cite{dong2022gaussian}. 

The convex conjugate of \(f\) is defined as
\[
f^{*}(y) = \sup\nolimits_{0 \leq \alpha' \leq 1} \, (y\alpha' - f(\alpha')).
\]
Let \(\epsilon^{*}(\delta_{g}) = \textup{OPT}\left(\mathcal{N}(\psi_{[2]}), \delta_{g}\right)\) for any \(\delta_{g} \in [0,1)\), where \(\psi_{[2]}\) is induced by \(\gamma_{[2]}\). Additionally, define 
\[
\delta_{[2]}(\epsilon_{g}) = 1 + f^{*}_{[2]}(-e^{\epsilon_{g}}),
\]
\[
\delta_{1,2}(\epsilon_{g}) = 1 + f^{*}_{1,2}(-e^{\epsilon_{g}}),
\]
for any \(\epsilon_{g} \geq 0\).

When \(\mathtt{D}_{[2]} \neq \mathtt{D}_{1,2}\), Proposition \ref{prop:blackwell_ording} implies 
\[
T(\mathtt{D}_{[2]}(\cdot \mid s_{0}), \mathtt{D}_{[2]}(\cdot \mid s_{1})) < T(\mathtt{D}_{1,2}(\cdot \mid s_{0}), \mathtt{D}_{1,2}(\cdot \mid s_{1})).
\]
Consequently, \(f_{[2]}(\alpha) < f_{1,2}(\alpha)\) for all \(\alpha \in [0,1]\), which implies that for any fixed \(\epsilon_{g}\) and \(\epsilon\), we have \(\delta_{[2]}(\epsilon) > \delta_{1,2}(\epsilon)\).

Suppose \(\delta_{[2]}(\epsilon_{[2]}) = \delta_{1,2}(\epsilon_{1,2})\). Then, \(f^{*}_{[2]}(-e^{\epsilon^{[2]}_{g}}) = f^{*}_{1,2}(-e^{\epsilon^{1,2}_{g}})\). Since both \(f_{[2]}\) and \(f_{1,2}\) are convex, continuous, and non-increasing functions, their convex conjugates are concave, continuous, and non-decreasing functions. 

Furthermore, since \(\delta_{[2]}(\epsilon) > \delta_{1,2}(\epsilon)\) for all \(\epsilon \geq 0\), to achieve \(\delta_{[2]}(\epsilon_{[2]}) = \delta_{1,2}(\epsilon_{1,2})\), we must have \(\epsilon_{[2]} > \epsilon_{1,2}\).

Thus, we complete the proof of Theorem \ref{thm:standard_compose_fail_general}.

Theorem \ref{thm:standard_compose_fail_general} confirms the following parts of Theorem \ref{thm:order_composition}:
\begin{itemize}
    \item If \(\textbf{L}_{c_{s_{0}}, c_{s_{1}}}(\vec{Y}) \neq 0\), then
    \[
    \begin{aligned}
        &\underline{\textbf{opt}}_{\delta_{g}}< \textbf{opt}_{\delta_{g}}, \; \underline{\textbf{dt}}_{\epsilon_{g}}< \textbf{dt}_{\epsilon_{g}}.
    \end{aligned}
    \]
\end{itemize}
The second set of inequalities:
\[
\begin{aligned}
    &\textbf{opt}_{\delta_{g}} \leq \overline{\textbf{opt}}_{\delta_{g}},\; \textbf{dt}_{\epsilon_{g}} \leq \overline{\textbf{dt}}_{\epsilon_{g}},
\end{aligned}
\]
arise from the fact that, in the PLRV given by (\ref{eq:privacy_loss_RV}), the copula-dependent PLRV \(\textbf{L}^{\mathcal{G}}_{c_{s_{0}}, c_{s_{1}}}(\vec{Y})\) is correlated with \(\{\textbf{L}_{\psi_{i}^{s_{0}}, \psi_{i}^{s_{1}}}(Y_{i})\}_{i=1}^{k}\), where each element of \(\{\textbf{L}_{\psi_{i}^{s_{0}}, \psi_{i}^{s_{1}}}(Y_{i})\}_{i=1}^{k}\) is independent of the others. 

Consequently, even when each \(\psi_{i}\) achieves the bounds of its probabilistic indistinguishability, the copula may not attain its own bound. The maximum privacy loss (\(\overline{\textbf{opt}}_{\delta_{g}}\) or \(\overline{\textbf{dt}}^{\epsilon_{g}}\)) occurs when all the individual mechanisms and the copula term achieve their respective bounds of probabilistic indistinguishability, maximizing the cumulative privacy loss.

\hfill $\square$

% \newpage

% \[
% \begin{aligned}
%     L(\pi, s, \vec{y}_{\alpha}) = \log\left(\pi(s|\vec{y}_{\alpha})\right)
% \end{aligned}
% \]

% \[
% \text{Obj}(\alpha, \pi)=\mathbb{E}^{\alpha}_{\theta, \sigma}\left[ L(\pi, S, \vec{Y}_{\alpha}) \right]
% \]

% \newpage

% \textcolor{violet}{$$\bm{AAAAAAAAAAAA}$$}

%%%%%%%%%%%%%%%%%%%%
\subsection{Proof of Theorem \ref{thm:inverse_composition}}

In part (i), the condition $\Pi[\epsilon_{g}, \delta_{g}] \neq \emptyset$ is necessary because the target privacy budget $(\epsilon_{g}, \delta_{g})$ may already be violated by the existing $k$ mechanisms $\{\mathcal{M}_{i}(\gamma_{i})\}_{i=1}^{k}$. 

Part (ii) of Theorem~\ref{thm:order_composition_pre} (and part (i) of Theorem~\ref{thm:order_composition}) indicates that while basic composition theory is generally satisfied by DCP, it prevents directly using a conservative summation of privacy parameters to estimate the overall privacy budget. Consequently, \textbf{Task-1} requires a rough and conservative estimation of the total privacy loss to ensure that $\Pi[\epsilon_{g}, \delta_{g}] \neq \emptyset$.

In \textbf{Task-2}, $\epsilon_{g}$ becomes an additional decision variable. Ideally, if the optimization problem is solved perfectly, $\epsilon_{g}$ would correspond to the true, tightest privacy loss value for the composition $\mathcal{M}(\vec{\gamma})$.

The main proof task of Theorem~\ref{thm:inverse_composition} is to demonstrate that the constraints specified by $\Pi[\epsilon_{g}, \delta_{g}]$, for $\epsilon_{g} \geq 0$ and $\delta_{g} \in [0,1]$, are sufficient to guarantee that the privacy strategy $\alpha$—which induces the best response $\pi$ satisfying $\Pi[\epsilon_{g}, \delta_{g}]$—ensures that the composition satisfies $(\epsilon_{g}, \delta_{g})$-DCP.

% Note that in part (i), the condition $\Pi[\epsilon_{g}, \delta_{g}] \neq \emptyset$ is required because it is possible that a target privacy budget $(\epsilon_{g}, \delta_{g})$ may be already violated by the existing $k$ mechanisms $\{\mathcal{M}_{i}(\gamma_{i})\}^{k}_{i=1}$.
% Part (ii) of Theorem \ref{thm:order_composition_pre} (also part (i) of Theorem \ref{thm:order_composition}) implies that even basic composition theory is generally satisfied by DCP, preventing the direct conservative choice of privacy budget by the summation of privacy parameters.
% Thus, \textbf{Task-1} requires a rough and conservative estimation of the overall privacy loss to ensure that $\Pi[\epsilon_{g}, \delta_{g}] \neq \emptyset$.

% In \textbf{Task-2}, we have an additional decision variable $\epsilon_{g}$. If the optimization can be solved ideally, the $\epsilon_{g}$ would concide with the true and tightest epsilon value of the composition $\mathcal{M}(\vec{\gamma})$.

% The main proof tasks of Theorem \ref{thm:inverse_composition} is to show that the constraints specified by $\Pi[\epsilon_{g}, \delta_{g}]$ for $\epsilon_{g}\geq 0$ and $\delta_{g}\in[0,1]$ is sufficient to guarantee that the privacy strategy $\alpha$, which induces the best response $\pi$ that satisfies $\Pi[\epsilon_{g}, \delta_{g}]$ for $\epsilon_{g}\geq 0$, makes to the composition satisfy $(\epsilon_{g},\delta_{g})$-DCP.

% $$AAAAA$$

Given any choice of strictly proper scoring rules \(\mathcal{L}(\pi, s, \vec{y}_{\alpha})\), the posterior distribution \(\pi^{*} = \mu\) is the \textit{unique} minimizer of \(\mathbb{E}_{\alpha, \vec{\gamma}}^{\mu}[\mathcal{L}(\pi, s, \vec{y}_{\alpha})]\) \cite{gneiting2007strictly}. Here, \(\mu\) denotes the posterior distribution induced by the leader's privacy strategy \(\alpha\) and the composition \(\mathcal{M}(\vec{\gamma})\) (or equivalently, \(\mathcal{M}(\vec{\gamma}_{\alpha})\)).

To establish that the composition \(\mathcal{M}(\vec{\gamma}_{\alpha})\) satisfies \((\epsilon_{g}, \delta_{g})\)-DCP, we demonstrate that the induced posterior distribution fulfills the conditions in \(\Pi[\epsilon_{g}, \delta_{g}]\). Consequently, the composition \(\mathcal{M}(\vec{\gamma}_{\alpha})\) is proven to be \((\epsilon_{g}, \delta_{g})\)-DCP.

With abuse of notation, we let $\vec{\gamma}^{c}_{\alpha}$ denote the joint density function of $\mathcal{M}(\vec{\gamma}_{\alpha})$ with the underlying copula density $c$ capturing the dependencies among the mechanisms $\mathcal{M}_{1}(\gamma_{1}), \dots, \mathcal{M}_{k}(\gamma_{k})$.
In addition, let $\vec{\psi}^{c}_{\alpha}(\cdot|s)$ be the corresponding density of the effective mechanism.
Given the density $\vec{\gamma}^{c}_{\alpha}$ and any output profile $\vec{y}_{\alpha}\in \vec{\mathcal{Y}}\times \mathcal{Y}_{\alpha}$, the posterior belief is constructed by
\[
\begin{aligned}
    \text{J}\left(s\middle|\vec{y}_{\alpha}\right)=\frac{P(s, \vec{y}_{\alpha})}{P(s, \vec{y}_{\alpha})},
\end{aligned}
\]
where $P(s, \vec{y}_{\alpha})\equiv\int_{x}\vec{\gamma}^{c}_{\alpha}(\vec{y}_{\alpha}|x)P_{\theta}(s,x)dx$ and $P(s, \vec{y}_{\alpha})\equiv \int_{s} P(s, \vec{y}_{\alpha}) ds$.

By Lemma \ref{lemma:probabilistic_representation1} shown below, we first obtain that 
\[
\begin{aligned}
    \text{J}\left(s\middle|\vec{y}_{\alpha}\right) \geq e^{-\epsilon}P_{\theta}(s),
\end{aligned}
\]
and
\[
\mathbb{E}_{\text{J}}\left[\frac{\text{J}\left(S\middle|\vec{y}_{\alpha}\right)}{P_{\theta}(S)}\right] \leq \delta_{g} e^{\epsilon_{g}},
\]
implies
\[
\begin{aligned}
    \text{Pr}_{S\sim \text{J}\left(\cdot\middle|\vec{y}_{\alpha}\right)}\left[e^{-\epsilon_{g}}\leq \frac{\text{J}\left(S\middle|\vec{y}_{\alpha}\right)}{P_{\theta}(S)} \leq e^{\epsilon_{g}} \right] \geq 1-\delta_{g},
\end{aligned}
\]
for all $\vec{y}_{\alpha}\in \vec{\mathcal{Y}}\times \mathcal{Y}_{\alpha}$.
Then, if $\vec{Y}_{\alpha}\sim \vec{\psi}^{c}_{\alpha}(\cdot|s)$, it holds that 
\begin{equation}\label{eq:probability_posterior_prior_ratio}
    \text{Pr}_{\vec{Y}_{\alpha}\sim \vec{\psi}^{c}_{\alpha}(\cdot|s)}\left[e^{-\epsilon_{g}}\leq \frac{\text{J}\left(S\middle|\vec{Y}_{\alpha}\right)}{P_{\theta}(S)} \leq e^{\epsilon_{g}}  \right]\geq 1-\delta_{g}.
\end{equation}

From Lemma \ref{lemma:semantic_secure_1}, we have that
\[
\begin{aligned}
    e^{-\epsilon_{g}}\leq \frac{\text{J}\left(s\middle|\vec{y}_{\alpha}\right)}{P_{\theta}(s)} \leq e^{\epsilon_{g}}, \forall s\in \mathcal{S},
\end{aligned}
\]
implies
\[
\textbf{Pr}^{s}_{\vec{\gamma}^{c}_{\alpha}}\left[\vec{y}_{\alpha}\in\vec{\mathcal{W}}_{\alpha}\right]\leq e^{\varepsilon(\epsilon_{g}, \theta) }\textbf{Pr}^{s'}_{\vec{\gamma}^{c}_{\alpha}}\left[\vec{y}_{\alpha}\in\vec{\mathcal{W}}_{\alpha}\right],
\]
for all $(s,s')\in\mathcal{Q}$, $\vec{\mathcal{W}}_{\alpha}\subset \vec{\mathcal{Y}}_{\alpha}$, where $\varepsilon(\epsilon_{g}, \theta)\equiv \log\left( 1 + \frac{e^{\epsilon_{g}}-1}{P^{*}_{\theta}}\right)$ with $P^{*}_{\theta}=\min_{s}P_{\theta}(s)$.
Thus, by (\ref{eq:probability_posterior_prior_ratio}), we have
\begin{equation}\label{eq:probability_CP}
    \begin{aligned}
        \text{Pr}_{\vec{Y}_{\alpha}\sim \vec{\psi}^{c}_{\alpha}(\cdot|s)}\left[\frac{\textbf{Pr}^{s}_{\vec{\gamma}^{c}_{\alpha}}\left[\vec{y}_{\alpha}\in\vec{\mathcal{W}}_{\alpha}\right]}{\textbf{Pr}^{s'}_{\vec{\gamma}^{c}_{\alpha}}\left[\vec{y}_{\alpha}\in\vec{\mathcal{W}}_{\alpha}\right]} \leq e^{\epsilon_{g}}  \right]\geq 1-\delta_{g},
    \end{aligned}
\end{equation}
for all $(s,s')\in\mathcal{Q}$, $\vec{\mathcal{W}}_{\alpha}\subset \vec{\mathcal{Y}}_{\alpha}$.

Finally, Lemma \ref{lemma:probability_to_deterministic} shows that (\ref{eq:probability_CP}) implies 
\[
\textbf{Pr}^{s}_{\vec{\gamma}^{c}_{\alpha}}\left[\vec{y}_{\alpha}\in\vec{\mathcal{W}}_{\alpha}\right]\leq e^{\epsilon_{g}} \textbf{Pr}^{s'}_{\vec{\gamma}^{c}_{\alpha}}\left[\vec{y}_{\alpha}\in\vec{\mathcal{W}}_{\alpha}\right]+\delta_{g}.
\]
Therefore, we can conclude that the sufficiency of the conditions given by $\Pi[\epsilon_{g}, \delta_{g}]$ for any $\epsilon_{g}\geq 0$ and $\delta_{g}\in(0,1]$.
In addition, when $\delta_{g}=0$, Lemma \ref{lemma:semantic_secure_1} establishes the sufficiency.

\subsubsection{Related Lemmas}

\begin{lemma}\label{lemma:probabilistic_representation1}
Let $F:\mathcal{Y}\mapsto \Delta(\mathcal{S})$ be any conditional probability mass or density function, and let $P_{S}\in\Delta(\mathcal{S})$ be any prior distribution of the secret $S$. For any $\epsilon\geq 0$ and $\delta\in[0,1)$, if 
\[
F(s|y)\geq e^{-\epsilon}p_{\theta_{S}}(s) \textup{ and } \mathbb{E}\left[\frac{F\left(S\middle|y\right)}{P_{S}(S)}\right]\leq \delta e^{\epsilon},
\]
where the expectation is taken over the randomness of $S\sim F(\cdot|y)$, then
\[
\textup{Pr}\left[e^{-\epsilon}\leq\frac{F(S|y)}{P_{S}(S)}\leq e^{\epsilon}\right]\geq 1-\delta,
\]
for any $s\in\{s_{0}, s_{1}\}\in\mathcal{Q}$, $y\in\mathcal{Y}$.
\end{lemma}

\begin{proof}
First, $F(s|y)\geq e^{-\epsilon}p_{\theta_{S}}(s)$ implies 
\[
\begin{aligned}
    \textup{Pr}\left[e^{-\epsilon}\leq\frac{F(s|y)}{P_{S}(s)}\leq e^{\epsilon}\right]=1-\textup{Pr}\left[\frac{F(s|\vec{y}_{[k]})}{P_{S}(s)}\geq e^{\epsilon}\right].
\end{aligned}
\]
By Markov's inequality, we have 
\[
\textup{Pr}\left[\frac{F(s|\vec{y}_{[k]})}{P_{S}(s)}\geq e^{\epsilon}\right]\leq \frac{\mathbb{E}\left[\frac{F(s|\vec{y}_{[k]})}{P_{S}(s)}\right] }{e^{\epsilon}}.
\]
Since $\mathbb{E}\left[\frac{F\left(S\middle|y\right)}{P_{S}(S)}\right]\leq \delta e^{\epsilon}$, we have 
\[
\textup{Pr}\left[\frac{F(s|\vec{y}_{[k]})}{P_{S}(s)}\geq e^{\epsilon}\right]\leq \delta.
\]
Therefore, $\textup{Pr}\left[e^{-\epsilon}\leq\frac{F(s|\vec{y}_{[k]})}{P_{S}(s)}\leq e^{\epsilon}\right]\geq 1-\delta$.
\end{proof}

\begin{lemma}\label{lemma:semantic_secure_1}
Given $\theta\in \Theta$, a mechanism $\mathcal{M}(\gamma)$ is $(\varepsilon(\epsilon_{g}, \theta),0)$-DCP with $\varepsilon(\epsilon_{g}, \theta)= \log\left( 1 + \frac{e^{\epsilon_{g}}-1}{P^{*}_{\theta}}\right)$ where $P^{*}_{\theta}=\min_{s}P_{\theta}(s)$ if the posterior $\textup{J}$ satisfies  
    \begin{equation}\label{eq:lemma_posterior_inequality}
        e^{-\epsilon}\leq\frac{\text{J}\left(s\middle|y\right)}{P_{\theta}(s)}\leq e^{\epsilon},
    \end{equation}
    for all $s\in \mathcal{S}$.
\end{lemma}

Lemma \ref{lemma:semantic_secure_1} extends Claim 3 of \cite{dwork2006calibrating} to the general differential confounding privacy.

\begin{proof}

Let $\mathcal{N}(\psi):\mathcal{S}\mapsto \mathcal{Y}$ be the effective mechanism of $\mathcal{M}(\gamma)$, with well-defined density $\psi$, prior $P_{\theta}$, and induced posterior $\mathtt{J}$. Let $E$ be any measurable event.

Suppose that the mechanism $\mathcal{M}(\gamma)$ satisfies (\ref{eq:lemma_posterior_inequality}). 
Then, 
\[
e^{-\epsilon}\leq\frac{\text{J}\left(E\middle|y\right)}{P_{\theta}(E)}\leq e^{\epsilon}.
\]
Given $\theta\in\Theta$, construct a prior $P_{\theta}$ (with abuse of notation) that places mass only on $s_{0}$ and $s_{1}$, such that $P_{\theta}(s_{0}) + P_{\theta}(s_{1})=1$.
Set the event $E=\{s_{1}\}$. Then, 
\[
\mathtt{J}(E|y)=\frac{\psi(y|s_{1})P_{\theta}(s_{1})}{\psi(y|s_{0})P_{\theta}(s_{0}) + \psi(y|s_{1})P_{\theta}(s_{1})},
\]
which gives
\[
\begin{aligned}
    \frac{P_{\theta}(E)}{\text{J}(E|y)}&=P_{\theta}(s_{1}) + P_{\theta}(s_{0}) \frac{\psi(y|s_{0})}{\psi(y|s_{1})}\\
    &=1- P_{\theta}(s_{0})  + P_{\theta}(s_{0}) \frac{\psi(y|s_{0})}{\psi(y|s_{1})}.
\end{aligned}
\]
Since (\ref{eq:lemma_posterior_inequality}) is satisfied, we have
\[
\begin{aligned}
    e^{-\epsilon}\leq 1- P_{\theta}(s_{0})  + P_{\theta}(s_{0}) \frac{\psi(y|s_{0})}{\psi(y|s_{1})}\leq e^{\epsilon},
\end{aligned}
\]
which yields
\[
\begin{aligned}
   \max\{0, \frac{e^{-\epsilon} - 1}{P_{\theta}(s_{0})} + 1\} \leq \frac{\psi(y|s_{0})}{\psi(y|s_{1})}\leq \frac{e^{\epsilon} - 1}{P_{\theta}(s_{0})} + 1,
\end{aligned}
\]
for all $(s_{0}, s_{1})\in\mathcal{Q}$. 
Therefore, the maximum upper bound of the ratio $\frac{\psi(y|s_{0})}{\psi(y|s_{1})}$, for all $(s_{0},s_{1})\in\mathcal{Q}$, that we can have is $\varepsilon(\epsilon_{g}, \theta)= \log\left( 1 + \frac{e^{\epsilon_{g}}-1}{\min_{s}P_{\theta}(s) }\right)$, where $\min_{s}P_{\theta}(s) >0$ for all $s\in\{s_{0},s_{1}\}\in\mathcal{Q}$.
%
% Since $\frac{e^{-\epsilon} - 1}{P_{\theta}(s_{0})} + 1 \leq \frac{e^{\epsilon} - 1}{P_{\theta}(s_{0})} + 1$, the lower bound of $\frac{\psi(y|s_{0})}{\psi(y|s_{1})}$ is within $[0, \frac{e^{\epsilon} - 1}{P_{\theta}(s_{0})} + 1]$.
% Thus, the maximum distinguishability of $\frac{\psi(y|s_{0})}{\psi(y|s_{1})}$ is $\frac{e^{\epsilon} - 1}{P_{\theta}(s_{0})} + 1$.
% Therefore, the maximum upper bound of the ratio $\frac{\psi(y|s_{0})}{\psi(y|s_{1})}$, for all $(s_{0},s_{1})\in\mathcal{Q}$, that we can have is $\varepsilon(\epsilon_{g}, \theta)= \log\left( 1 + \frac{e^{\epsilon_{g}}-1}{\min_{s}P_{\theta}(s) }\right)$. 

\end{proof}

The condition (\ref{eq:probability_BCP}) in Lemma \ref{lemma:probability_to_deterministic} is equivalent to the \textit{strong Bayesian DP} introduced by \cite{triastcyn2020bayesian}.
By treating the DCP as a single-point DP, Lemma \ref{lemma:probability_to_deterministic} can be proved in the similar way as Proposition 1 of \cite{triastcyn2020bayesian}.

\begin{lemma}\label{lemma:probability_to_deterministic}
Given $\theta\in\Theta$, if
\begin{equation}\label{eq:probability_BCP}
    \begin{aligned}
        \text{Pr}_{\vec{Y}_{\alpha}\sim \vec{\psi}^{c}_{\alpha}(\cdot|s)}\left[\frac{\textbf{Pr}^{s}_{\vec{\gamma}^{c}_{\alpha}}\left[\vec{y}_{\alpha}\in\vec{\mathcal{W}}_{\alpha}\right]}{\textbf{Pr}^{s'}_{\vec{\gamma}^{c}_{\alpha}}\left[\vec{y}_{\alpha}\in\vec{\mathcal{W}}_{\alpha}\right]} \leq e^{\epsilon_{g}}  \right]\geq 1-\delta_{g},
    \end{aligned}
\end{equation}
for all $\vec{\mathcal{W}}_{\alpha}\subseteq \vec{\mathcal{Y}}\times\mathcal{Y}_{\alpha}$, $(s,s')\in \mathcal{Q}$, $\epsilon_{g}\geq 0$, $\delta_{g}\in[0,1)$, then we have
\[
\textbf{Pr}^{s}_{\vec{\gamma}^{c}_{\alpha}}\left[\vec{y}_{\alpha}\in\vec{\mathcal{W}}_{\alpha}\right]\leq e^{\epsilon_{g}} \textbf{Pr}^{s'}_{\vec{\gamma}^{c}_{\alpha}}\left[\vec{y}_{\alpha}\in\vec{\mathcal{W}}_{\alpha}\right]+\delta_{g}.\]
\end{lemma}

\subsection{Experimental Details}

\subsubsection{Network Configurations and Hyperparameters}

The Defender neural network is a generative model designed to process membership vectors and generate beacon modification decisions. The input layer is followed by a series of fully connected layers with activation functions applied after each layer. The first and second hidden layers use ReLU activation, while the third hidden layer employs LeakyReLU activation. Layer normalization is applied after the second and third hidden layers. The output layer utilizes a custom ScaledSigmoid activation function, producing a real value in a scaled range determined by this function. The network is trained using the AdamW optimizer with a learning rate of $1.5 \times 10^{-5}$ and weight decay of $1.5 \times 10^{-5}$.
Specific configurations for Defender performing single-shot re-design and copula perturbation are provided in Tables \ref{table:defender_single} and \ref{table:defender_copula}, respectively.

The Attacker neural network is a generative model designed to process beacons and noise, producing membership vectors. The input layer is followed by multiple fully connected layers, each equipped with batch normalization and ReLU activation. The output layer applies a standard sigmoid activation function. All Attacker models were trained using the AdamW optimizer, with a learning rate of $1.0 \times 10^{-4}$ and a weight decay of $1.0 \times 10^{-4}$.
Specific configurations for Attacker when Defender performing single-shot re-design and copula perturbation are provided in Tables \ref{table:attacker_single} and \ref{table:attacker_copula}, respectively.

% Second table: Defender (Single Re-Design)
\begin{table}[h]
\centering
\small
\begin{tabular}{|l|c|c|}
\hline
\textbf{Defender} & \textbf{Input Units} & \textbf{Output Units} \\ \hline
Input Layer      & 800 & 5000        \\ \hline
Hidden Layer 1   & 5000          & 9000        \\ \hline
Hidden Layer 2   & 9000          & 3000        \\ \hline
Output Layer     & 3000          & 1000 \\ \hline
\end{tabular}
\vspace{0.5cm}
\caption{\small Defender Neural Network Architecture (Single Re-Design)}\label{table:defender_single}
\end{table}

% First table: Defender (Copula Perturbation)
\begin{table}[h]
\centering
\small
\begin{tabular}{|l|c|c|}
\hline
\textbf{Defender} & \textbf{Input Units} & \textbf{Output Units} \\ \hline
Input Layer      & 830             & 3000        \\ \hline
Hidden Layer 1   & 3000          & 2000        \\ \hline
Hidden Layer 2   & 2000          & 1200        \\ \hline
Hidden Layer 3   & 1200          & 1200        \\ \hline
Output Layer     & 1200          & 1           \\ \hline
\end{tabular}
\caption{\small Defender Neural Network Architecture (Copula Perturbation)}\label{table:defender_copula}
\end{table}

% Fourth table: Attacker (Single Re-Design)
\begin{table}[h]
\centering
\small
\begin{tabular}{|l|c|c|}
\hline
\textbf{Attacker} & \textbf{Input Units} & \textbf{Output Units} \\ \hline
Input Layer      & 1500 & 14000       \\ \hline
Hidden Layer 1   & 14000         & 9500        \\ \hline
Hidden Layer 2   & 9500          & 4500        \\ \hline
Output Layer     & 4500          & 800 \\ \hline
\end{tabular}
\vspace{0.5cm}
\caption{\small  Attacker Neural Network Architecture (Single Re-Design)}\label{table:attacker_single}
\end{table}

% Third table: Attacker (Copula Perturbation)
\begin{table}[h]
\centering
\small
\begin{tabular}{|l|c|c|}
\hline
\textbf{Attacker} & \textbf{Input Units} & \textbf{Output Units} \\ \hline
Input Layer      & 1500             & 31000       \\ \hline
Hidden Layer 1   & 31000         & 25000       \\ \hline
Hidden Layer 2   & 25000         & 15000       \\ \hline
Output Layer     & 15000         & 800 \\ \hline
\end{tabular}
\vspace{0.5cm}
\caption{\small  Attacker Neural Network Architecture (Copula Perturbation)}\label{table:attacker_copula}
\end{table}

\subsubsection{AUC Values with Standard Deviations}

Tables \ref{table:auc_single} and \ref{table:auc_copula} show the AUC values shown in the plots.

\begin{table}[h]
\centering
\begin{tabular}{|l|l|c|}
\hline
\textbf{Target $\&$ Ind. $\epsilon$} & \textbf{Mechanisms} & \textbf{AUC ± std} \\ \hline
\multirow{2}{*}{$\epsilon_g=0.25$, $\epsilon_i=0.05$} & Composition & 0.7441 ± 0.0055 \\ \cline{2-3} 
                                                      & Single      & 0.7485 ± 0.0064 \\ \hline

\multirow{2}{*}{$\epsilon_g=0.5$, $\epsilon_i=0.1$} & Composition & 0.8093 ± 0.0042 \\ \cline{2-3} 
                                                    & Single      & 0.8165 ± 0.0044 \\ \hline

\multirow{2}{*}{$\epsilon_g=1.5$, $\epsilon_i=0.3$} & Composition & 0.8921 ± 0.0031 \\ \cline{2-3} 
                                                    & Single      & 0.9056 ± 0.0029 \\ \hline

\multirow{2}{*}{$\epsilon_g=3$, $\epsilon_i=0.6$} & Composition & 0.9138 ± 0.0031 \\ \cline{2-3} 
                                                  & Single      & 0.9175 ± 0.0031 \\ \hline

\multirow{2}{*}{$\epsilon_g=5$, $\epsilon_i=1$} & Composition & 0.9330 ± 0.0029 \\ \cline{2-3} 
                                                & Single      & 0.9308 ± 0.0028 \\ \hline
\end{tabular}
\vspace{0.5cm}
\caption{\small  AUC and standard deviation for single-shot re-design when the target $\delta_{g}$ and each individual $\delta_{i}$ are the same, i.e., $\delta_{g}=\delta_{i}=0.02$, and $\rho=0.5$}\label{table:auc_single}
\end{table}

\begin{table}[h]
\centering
\begin{tabular}{|l|l|c|}
\hline
\textbf{Target $\&$ Ind. $\epsilon$} & \textbf{Type} & \textbf{AUC ± std} \\ \hline
\multirow{2}{*}{$\epsilon_g=0.4$, $\epsilon_i=0.05$} & Composition & 0.5624 ± 0.0078 \\ \cline{2-3} 
                                                     & Single      & 0.5745 ± 0.0060 \\ \hline

\multirow{2}{*}{$\epsilon_g=0.6$, $\epsilon_i=0.1$} & Composition & 0.6206 ± 0.0071 \\ \cline{2-3} 
                                                    & Single      & 0.6388 ± 0.0053 \\ \hline

\multirow{2}{*}{$\epsilon_g=1$, $\epsilon_i=0.18$} & Composition & 0.6943 ± 0.0063 \\ \cline{2-3} 
                                                   & Single      & 0.7086 ± 0.0052 \\ \hline

\multirow{2}{*}{$\epsilon_g=2$, $\epsilon_i=0.3$} & Composition & 0.7429 ± 0.0055 \\ \cline{2-3} 
                                                  & Single      & 0.7324 ± 0.0051 \\ \hline

\multirow{2}{*}{$\epsilon_g=4$, $\epsilon_i=0.6$} & Composition & 0.8235 ± 0.0035 \\ \cline{2-3} 
                                                  & Single      & 0.8139 ± 0.0042 \\ \hline

\multirow{2}{*}{$\epsilon_g=6$, $\epsilon_i=1$} & Composition & 0.8238 ± 0.0043 \\ \cline{2-3} 
                                                & Single      & 0.8142 ± 0.0044 \\ \hline
\end{tabular}
\vspace{0.5cm}
\caption{ AUC and standard deviation for copula perturbation when the target $\delta_{g}$ and each individual $\delta_{i}$ are the same, i.e., $\delta_{g}=\delta_{i}=0.02$, and $\rho=0.5$}\label{table:auc_copula}
\end{table}

\end{document}